\documentclass[runningheads]{llncs}
\let\llncssubparagraph\subparagraph
\let\subparagraph\paragraph
\usepackage[compact]{titlesec}
\let\subparagraph\llncssubparagraph

\usepackage{marvosym}
\usepackage{amsmath}
\usepackage{amssymb}
\usepackage{palatino,mathpazo} 
\usepackage[T1]{fontenc}
\usepackage[utf8]{inputenc}
\usepackage[mathscr]{eucal}
\usepackage{titlesec}
\usepackage{bm}
\usepackage[all,cmtip,2cell]{xy}
\usepackage{enumerate}
\usepackage{graphicx}
\usepackage{stmaryrd}
\usepackage{rotating}
\usepackage{enumitem}
\usepackage{algorithm}
\usepackage{algorithmic}
\usepackage{hyperref}

\allowdisplaybreaks

\usepackage{makecell}

\newcounter{para}[section]

\usepackage{tikz}
\usetikzlibrary{arrows,snakes,automata,backgrounds,petri, positioning,calc,cd, automata}
\usetikzlibrary{shapes.geometric,shapes.misc,matrix,arrows.meta,decorations.markings}
\tikzset{x=1em, y=1.5ex, baseline=-0.5ex}
\tikzset{ihbase/.style={inner sep=0,circle,draw,fill=lightgray,minimum size=0.4em,node contents={}}}
\tikzset{ihblack/.style={ihbase,fill=black}}
\tikzset{ihwhite/.style={ihbase,fill=white}}
\tikzset{mat/.style={draw,fill=white,rectangle,node font=\scriptsize}}
\tikzset{ha/.style={mat,rounded rectangle,rounded rectangle left arc=none}}
\tikzset{haop/.style={mat,rounded rectangle,rounded rectangle right arc=none}}
\tikzset{blackha/.style={mat,rounded rectangle,rounded rectangle left arc=none,font=\color{white},fill=black}}
\tikzset{blackhaop/.style={mat,rounded rectangle,rounded rectangle right arc=none,font=\color{white},fill=black}}
\tikzset{anti/.style={inner sep=0,isosceles triangle,fill=black,draw=black, minimum width=0.75em, node contents={}}}
\tikzset{antiop/.style={anti,shape border rotate=180}}
\tikzset{antisq/.style={inner sep=0,rectangle,fill=black, minimum height=1em, minimum width=0.6em, node contents={}}}
\tikzset{count/.style={above,inner ysep=0.15em,font=\scriptsize}}
\tikzset{axiom/.style={above,font=\small}}
\tikzset{dir/.style={-Latex}}
\tikzset{st/.style={decoration={markings,
    mark={at position 0.5 with {\draw (0, 2pt) to (0, -2pt);}}},
    postaction=decorate}}

\newcommand{\objred}{{\color{red} \blacktriangleright}}
\newcommand{\objr}{\blacktriangleright}
\newcommand{\objl}{\blacktriangleleft}

\newcommand{\sem}[1]{\left\llbracket{#1}\right\rrbracket}
\newcommand{\semreg}[1]{\left\llbracket{#1}\right\rrbracket_R}
\newcommand{\transreg}[1]{\left\langle{#1}\right\rangle}
\newcommand{\RegExp}{\mathsf{RegExp}}

\newcommand{\arrowright}{
\tikzset{x=1em, y=2.1ex}
\begin{tikzpicture}
	\begin{pgfonlayer}{nodelayer}
		\node [style=none] (0) at (-1.5, 0) {};
		\node [style=none] (1) at (0, 0) {};
		\node [style=none] (2) at (1.25, 0) {};
	\end{pgfonlayer}
	\begin{pgfonlayer}{edgelayer}
		\draw [->] (0.center) to (1.center);
		\draw (1.center) to (2.center);
	\end{pgfonlayer}
\end{tikzpicture}
}
\tikzset{x=1em, y=1.5ex}
}
\newcommand{\arrowleft}{
\tikzset{x=1em, y=2.1ex}
\begin{tikzpicture}
	\begin{pgfonlayer}{nodelayer}
		\node [style=none] (1) at (-1.25, 0) {};
		\node [style=none] (2) at (0, 0) {};
		\node [style=none] (3) at (1.25, 0) {};
	\end{pgfonlayer}
	\begin{pgfonlayer}{edgelayer}
		\draw (1.center) to (2.center);
		\draw [->] (3.center) to (2.center);
	\end{pgfonlayer}
\end{tikzpicture}
}
\tikzset{x=1em, y=1.5ex}
}

\newcommand{\genericcomult}[2]{
  \begin{tikzpicture}
    \node at (1, 0) [ihbase,solid,name=copy,#1];
    \draw[#2] (copy) .. controls (1.25, 0.75) .. (2, 0.75);
    \draw[#2] (0, 0) -- (copy);
    \draw[#2] (copy) .. controls (1.25, -0.75) .. (2, -0.75);
  \end{tikzpicture}
}

\newcommand{\genericcounit}[2]{
  \tikz \draw[#2] (0, 0) -- (1, 0) node[ihbase,#1, solid];
}
\newcommand{\genericcounitn}[2]{
  \tikz \draw (0, 0) -- node[count] {#2} (1, 0) node[ihbase,#1];
}
\newcommand{\genericmult}[2]{
  \tikz {
    \node at (1,0) (copy) [ihbase,#1,solid];
    \draw[#2] (0,  0.75) .. controls (0.75,  0.75) .. (copy);
    \draw[#2] (0, -0.75) .. controls (0.75, -0.75) .. (copy);
    \draw[#2] (copy) -- (2, 0);
  }
}

\newcommand{\genericunit}[2]{
  \tikz \draw[#2] (0, 0) node[ihbase,#1, solid] -- (1, 0);
}

\newcommand{\Bcomult}{\genericcomult{ihblack}{}}

\newcommand{\Bcounit}{\genericcounit{ihblack}{}}

\newcommand{\Bmult}{\genericmult{ihblack}{}}

\newcommand{\Bunit}{\genericunit{ihblack}{}}

\newcommand{\Wcounitn}[1]{\genericcounitn{ihwhite}}

\newcommand\scalar[1]{
  \tikz {
    \node[ha] (ha) {$#1$};
    \draw (ha.west) -- ++(-0.75, 0);
    \draw (ha.east) -- ++(0.75, 0);
  }
}

\tikzset{x=1em, y=1.5ex}
}

\pgfdeclarelayer{edgelayer}
\pgfdeclarelayer{nodelayer}
\pgfsetlayers{edgelayer,nodelayer,main}

\definecolor{light-gray}{gray}{.5}
\tikzstyle{none}=[inner sep=0pt]
\tikzstyle{plain}=[inner sep=0pt]
\tikzstyle{black}=[circle, draw=black, fill=black, inner sep=0pt, minimum size=4pt]
\tikzstyle{black-faded}=[circle, draw=light-gray, fill=light-gray, inner sep=0pt, minimum size=4pt]
\tikzstyle{white}=[circle, draw=black, fill=white, inner sep=0pt, minimum size=4.5pt]
\tikzstyle{white-faded}=[circle, draw=light-gray, fill=white, inner sep=0pt, minimum size=4.5pt]
\tikzstyle{delay}=[fill=black, regular polygon, regular polygon sides=3,rotate=-90, scale=.55]
\tikzstyle{delay-op}=[fill=black, regular polygon, regular polygon sides=3,rotate=90, scale=.55]
\tikzstyle{reg}=[draw, fill=white, rounded rectangle, rounded rectangle left arc=none, minimum height=1.2em, minimum width=1.4em, node font={\scriptsize}]
\tikzstyle{coreg}=[draw, fill=white, rounded rectangle, rounded rectangle right arc=none, minimum height=1.2em, minimum width=1.4em, node font={\scriptsize}]
\tikzstyle{rcoreg}=[draw=red, fill=white, rounded rectangle, rounded rectangle right arc=none, minimum height=1.2em, minimum width=1.4em, node font={\scriptsize}]
\tikzstyle{regb}=[draw, fill=black, rounded rectangle, rounded rectangle left arc=none, minimum height=1.2em, minimum width=1.4em, node font={\scriptsize}]
\tikzstyle{coregb}=[draw, fill=black, rounded rectangle, rounded rectangle right arc=none, minimum height=1.2em, minimum width=1.4em, node font={\scriptsize}]
\tikzstyle{rn}=[circle, draw=red, fill=red, inner sep=0pt, minimum size=4pt]
\tikzstyle{wrn}=[circle, draw=red, fill=white, inner sep=0pt, minimum size=4pt]
\tikzstyle{place}=[circle, draw=black, fill=white, inner sep=0pt, minimum size=9pt]
\tikzstyle{act}=[circle, draw=black, fill=white, inner sep=0pt, minimum size=4.5pt]
\tikzstyle{coact}=[draw, fill=white, rounded rectangle, rounded rectangle right arc=none, minimum height=.7em, minimum width=.9em, node font={\scriptsize}]

\tikzstyle{pl}=[circle,thick,draw=black!75,fill=white,minimum size=17pt]
\tikzstyle{port}=[circle, fill,inner sep=1.2pt]
\tikzstyle{transition}=[rectangle,thick,draw=black!75,
  			  fill=black!20,minimum size=7pt]

\tikzstyle{arrow}=[->]

\newcommand{\diagbox}[3]{
\tikzset{x=1em, y=2.1ex}
\begin{tikzpicture}
	\begin{pgfonlayer}{nodelayer}
		\node [style=none] (0) at (-0.75, 0.5) {};
		\node [style=none] (1) at (-0.25, 1) {};
		\node [style=none] (2) at (-0.75, -0.5) {};
		\node [style=none] (3) at (0.75, -0.5) {};
		\node [style=none] (4) at (-0.25, -1) {};
		\node [style=none] (5) at (0.75, 0.5) {};
		\node [style=none] (6) at (2.5, -0) {};
		\node [style=none] (7) at (0.75, -0) {};
		\node [style=none] (8) at (0.25, -1) {};
		\node [style=none] (9) at (0.25, 1) {};
		\node [style=none] (10) at (0, -0) {$#1$};
		\node [style=none] (11) at (2.25, 0.5) {\scriptsize $#3$};
		\node [style=none] (12) at (-2.25, 0.5) {\scriptsize $#2$};
		\node [style=none] (13) at (-2.5, -0) {};
		\node [style=none] (14) at (-0.75, -0) {};
	\end{pgfonlayer}
	\begin{pgfonlayer}{edgelayer}
		\draw [in=180, out=0, looseness=1.25] (7.center) to (6.center);
		\draw [semithick, in=0, out=-90, looseness=1.00] (3.center) to (8.center);
		\draw [semithick, in=-90, out=180, looseness=1.00] (4.center) to (2.center);
		\draw [semithick, in=180, out=90, looseness=1.00] (0.center) to (1.center);
		\draw [semithick, in=90, out=0, looseness=1.00] (9.center) to (5.center);
		\draw [semithick] (1.center) to (9.center);
		\draw [semithick] (5.center) to (3.center);
		\draw [semithick] (8.center) to (4.center);
		\draw [semithick] (2.center) to (0.center);
		\draw [in=180, out=0, looseness=1.25] (13.center) to (14.center);
	\end{pgfonlayer}
\end{tikzpicture}
\tikzset{x=1em, y=1.5ex}
}

\newcommand{\diagstate}[4]{
\tikzset{x=1em, y=2.1ex}
\begin{tikzpicture}
	\begin{pgfonlayer}{nodelayer}
		\node [style=none] (0) at (-0.75, 0.5) {};
		\node [style=none] (1) at (-0.25, 1) {};
		\node [style=none] (2) at (-0.75, -0.5) {};
		\node [style=none] (3) at (0.75, -0.5) {};
		\node [style=none] (4) at (-0.25, -1) {};
		\node [style=none] (5) at (0.75, 0.5) {};
		\node [style=none] (6) at (2.5, -0.5) {};
		\node [style=none] (7) at (0.75, -0.5) {};
		\node [style=none] (8) at (0.25, -1) {};
		\node [style=none] (9) at (0.25, 1) {};
		\node [style=none] (10) at (0, -0) {$#1$};
		\node [style=none] (11) at (2.5, -1) {\scriptsize $#3$};
		\node [style=none] (12) at (-2.5, -1) {\scriptsize $#2$};
		\node [style=none] (13) at (-2.5, -0.5) {};
		\node [style=none] (14) at (-0.75, -0.5) {};
		\node [style=none] (15) at (0.75, 0.5) {};
		\node [style=none] (16) at (2.5, 0.5) {};
		\node [style=none] (17) at (-2.5, 0.5) {};
		\node [style=none] (18) at (-0.75, 0.5) {};
		\node [style=none] (19) at (-2.5, 1) {\scriptsize $#4$};
		\node [style=none] (20) at (2.5, 1) {\scriptsize $#4$};
	\end{pgfonlayer}
	\begin{pgfonlayer}{edgelayer}
		\draw [in=180, out=0, looseness=1.25] (7.center) to (6.center);
		\draw [semithick, in=0, out=-90, looseness=1.00] (3.center) to (8.center);
		\draw [semithick, in=-90, out=180, looseness=1.00] (4.center) to (2.center);
		\draw [semithick, in=180, out=90, looseness=1.00] (0.center) to (1.center);
		\draw [semithick, in=90, out=0, looseness=1.00] (9.center) to (5.center);
		\draw [semithick] (1.center) to (9.center);
		\draw [semithick] (5.center) to (3.center);
		\draw [semithick] (8.center) to (4.center);
		\draw [semithick] (2.center) to (0.center);
		\draw [in=180, out=0, looseness=1.25] (13.center) to (14.center);
		\draw [in=180, out=0, looseness=1.25] (15.center) to (16.center);
		\draw [in=180, out=0, looseness=1.25] (17.center) to (18.center);
	\end{pgfonlayer}
\end{tikzpicture}
\tikzset{x=1em, y=1.5ex}
}

\newcommand{\traceform}[4]{
\tikzset{x=1em, y=2.1ex}
\begin{tikzpicture}
	\begin{pgfonlayer}{nodelayer}
		\node [style=none] (0) at (-0.75, 1.25) {};
		\node [style=none] (1) at (0.75, 1) {};
		\node [style=none] (2) at (-0.25, 1.75) {};
		\node [style=none] (3) at (-0.75, -0.25) {};
		\node [style=none] (4) at (0.75, -0.25) {};
		\node [style=none] (5) at (-0.25, -0.75) {};
		\node [style=none] (6) at (-0.75, 1) {};
		\node [style=none] (7) at (0.75, 1.25) {};
		\node [style=none] (8) at (3.25, 0) {};
		\node [style=none] (9) at (0.75, 0) {};
		\node [style=none] (10) at (0.25, -0.75) {};
		\node [style=none] (11) at (0.25, 1.75) {};
		\node [style=none] (12) at (0, 0.5) {$#1$};
		\node [style=none] (13) at (3, 0.5) {\scriptsize $#3$};
		\node [style=none] (14) at (1.5, 2.5) {};
		\node [style=none] (15) at (-1.5, 2.5) {};
		\node [style=none] (16) at (-2.75, 0.5) {\scriptsize $#2$};
		\node [style=none] (17) at (-3, 0) {};
		\node [style=none] (18) at (-0.75, 0) {};
		\node [style=none] (20) at (1.5, 1) {};
		\node [style=none] (21) at (2.75, 2.35) {\scriptsize $#4$};
		\node [style=none] (22) at (-1.5, 1) {};
	\end{pgfonlayer}
	\begin{pgfonlayer}{edgelayer}
		\draw [in=180, out=0, looseness=1.25] (9.center) to (8.center);
		\draw [semithick, in=0, out=-90] (4.center) to (10.center);
		\draw [semithick, in=-90, out=180] (5.center) to (3.center);
		\draw [semithick, in=180, out=90] (0.center) to (2.center);
		\draw [semithick, in=90, out=0] (11.center) to (7.center);
		\draw [semithick] (2.center) to (11.center);
		\draw [semithick] (7.center) to (4.center);
		\draw [semithick] (10.center) to (5.center);
		\draw [semithick] (3.center) to (0.center);
		\draw (15.center) to (14.center);
		\draw [in=180, out=0, looseness=1.25] (17.center) to (18.center);
		\draw (6.center) to (22.center);
		\draw [bend right=90, looseness=1.75] (20.center) to (14.center);
		\draw [bend left=90, looseness=1.75] (22.center) to (15.center);
		\draw (1.center) to (20.center);
	\end{pgfonlayer}
\end{tikzpicture}
\tikzset{x=1em, y=1.5ex}
}
\newcommand{\traceaction}[5]{
\tikzset{x=1em, y=2.1ex}
\begin{tikzpicture}
	\begin{pgfonlayer}{nodelayer}
		\node [style=none] (0) at (-0.75, 1.25) {};
		\node [style=none] (1) at (0.75, 1) {};
		\node [style=none] (2) at (-0.25, 1.75) {};
		\node [style=none] (3) at (-0.75, -0.25) {};
		\node [style=none] (4) at (0.75, -0.25) {};
		\node [style=none] (5) at (-0.25, -0.75) {};
		\node [style=none] (6) at (-0.75, 1) {};
		\node [style=none] (7) at (0.75, 1.25) {};
		\node [style=none] (8) at (4.25, 0) {};
		\node [style=none] (9) at (0.75, 0) {};
		\node [style=none] (10) at (0.25, -0.75) {};
		\node [style=none] (11) at (0.25, 1.75) {};
		\node [style=none] (12) at (0, 0.5) {$#1$};
		\node [style=none] (13) at (4, 0.5) {\scriptsize $#3$};
		\node [style=none] (14) at (2.5, 2.5) {};
		\node [style=none] (15) at (-1.5, 2.5) {};
		\node [style=none] (16) at (-2.75, 0.5) {\scriptsize $#2$};
		\node [style=none] (17) at (-3, 0) {};
		\node [style=none] (18) at (-0.75, 0) {};
		\node [style=none] (20) at (2.5, 1) {};
		\node [style=none] (21) at (3.5, 2.75) {\scriptsize $#4$};
		\node [style=none] (22) at (-1.5, 1) {};
		\node [style=reg] (23) at (1.75, 1) {$#5$};
	\end{pgfonlayer}
	\begin{pgfonlayer}{edgelayer}
		\draw [in=180, out=0, looseness=1.25] (9.center) to (8.center);
		\draw [semithick, in=0, out=-90] (4.center) to (10.center);
		\draw [semithick, in=-90, out=180] (5.center) to (3.center);
		\draw [semithick, in=180, out=90] (0.center) to (2.center);
		\draw [semithick, in=90, out=0] (11.center) to (7.center);
		\draw [semithick] (2.center) to (11.center);
		\draw [semithick] (7.center) to (4.center);
		\draw [semithick] (10.center) to (5.center);
		\draw [semithick] (3.center) to (0.center);
		\draw (15.center) to (14.center);
		\draw [in=180, out=0, looseness=1.25] (17.center) to (18.center);
		\draw (6.center) to (22.center);
		\draw [bend right=90, looseness=1.75] (20.center) to (14.center);
		\draw [bend left=90, looseness=1.75] (22.center) to (15.center);
		\draw (1.center) to (20.center);
	\end{pgfonlayer}
\end{tikzpicture}
\tikzset{x=1em, y=1.5ex}}

\newcommand{\myeq}[1]{\mathrel{\overset{\makebox[0pt]{\mbox{\normalfont\tiny\sffamily (#1)}}}{=}}}

\newcommand{\N}{\mathbb{N}}

\newcommand{\from}{\mathrel{:}\,}

\newcommand{\Aut}{\mathsf{Aut_{\scriptscriptstyle{\Sigma}}}}

\newcommand{\poi}{\,;\,}

\newcommand{\adjto}{\,\lower1pt\hbox{$\dashv$}\,}

\newcommand{\Rel}{\mathsf{Rel}}

\newcommand{\BProf}{\mathsf{Prof}_\mathbb{B}}

\newcommand{\eqKa}{=_{\scriptscriptstyle KAA}}

\makeatletter
\def\moverlay{\mathpalette\mov@rlay}
\def\mov@rlay#1#2{\leavevmode\vtop{%
\baselineskip\z@skip \lineskiplimit-\maxdimen
\ialign{\hfil$#1##$\hfil\cr#2\crcr}}}
\makeatother

\newcommand\twarr[2]{%
\mathrel{\mathop{\moverlay{\scriptstyle\xrightarrow{\,#1\,}\cr{\lower.2em\hbox{$\scriptstyle{}_{#2}$}}}}}}
\newcommand\twarrw[2]{%
\mathrel{\mathop{\moverlay{\scriptstyle\Longrightarrow\cr{\lower-.6em\hbox{$\scriptstyle{}_{#1}$}}
\cr{\lower.3em\hbox{$\scriptstyle{}_{#2}$}}}}}}

\newcommand{\dtransw}[2]{\raise1pt\hbox{$\;\twarrw{#1}{#2}\;$}}

\newcommand{\atom}[1]{
\begin{tikzpicture}
	\begin{pgfonlayer}{nodelayer}
		\node [style=wrn] (37) at (0.75, 0) {};
		\node [style=none] (44) at (2, 0) {};
		\node [style=none] (45) at (0.75, 1) {{\color{red} $\scriptstyle #1$}};
	\end{pgfonlayer}
	\begin{pgfonlayer}{edgelayer}
		\draw [red] (44.center) to (37);
	\end{pgfonlayer}
\end{tikzpicture}}

\newcommand{\kaId}{
\begin{tikzpicture}
	\begin{pgfonlayer}{nodelayer}
		\node [style=none] (2) at (1.25, 0) {};
		\node [style=none] (4) at (-1.25, 0) {};
	\end{pgfonlayer}
	\begin{pgfonlayer}{edgelayer}
		\draw [red] (4.center) to (2.center);
	\end{pgfonlayer}
\end{tikzpicture}}

\newcommand{\diagregexp}[1]{
\begin{tikzpicture}
	\begin{pgfonlayer}{nodelayer}
		\node [style=none] (0) at (1.5, 0) {};
		\node [style=rcoreg] (1) at (0, 0) {{\color{red} $e$}};
	\end{pgfonlayer}
	\begin{pgfonlayer}{edgelayer}
		\draw [red] (1) to (0.center);
	\end{pgfonlayer}
\end{tikzpicture}}

\newcommand{\smalldiag}[1]{\begin{tikzpicture}
	\begin{pgfonlayer}{nodelayer}
		\node [style=none] (0) at (-0.5, 0.5) {};
		\node [style=none] (1) at (-0.25, 0.75) {};
		\node [style=none] (2) at (-0.5, -0.5) {};
		\node [style=none] (3) at (0.5, -0.5) {};
		\node [style=none] (4) at (-0.25, -0.75) {};
		\node [style=none] (5) at (0.5, 0.5) {};
		\node [style=none] (6) at (1.5, 0) {};
		\node [style=none] (7) at (0.5, 0) {};
		\node [style=none] (8) at (0.25, -0.75) {};
		\node [style=none] (9) at (0.25, 0.75) {};
		\node [style=none] (10) at (0, 0) {$#1$};
		\node [style=none] (11) at (-1.5, 0) {};
		\node [style=none] (12) at (-0.5, 0) {};
	\end{pgfonlayer}
	\begin{pgfonlayer}{edgelayer}
		\draw [in=180, out=0, looseness=1.25] (7.center) to (6.center);
		\draw [semithick, in=0, out=-90] (3.center) to (8.center);
		\draw [semithick, in=-90, out=180] (4.center) to (2.center);
		\draw [semithick, in=180, out=90] (0.center) to (1.center);
		\draw [semithick, in=90, out=0] (9.center) to (5.center);
		\draw [semithick] (1.center) to (9.center);
		\draw [semithick] (5.center) to (3.center);
		\draw [semithick] (8.center) to (4.center);
		\draw [semithick] (2.center) to (0.center);
		\draw [in=180, out=0, looseness=1.25] (11.center) to (12.center);
	\end{pgfonlayer}
\end{tikzpicture}
}

\begin{document}
\title{A String Diagrammatic Axiomatisation \\ of Finite-State Automata}
%
%
\author{
Robin Piedeleu\inst{1}
\and Fabio Zanasi\inst{1}$^\star$(\Letter)
}
\authorrunning{R. Piedeleu and F. Zanasi}
%
\institute{
University College London, UK,
\email{\{r.piedeleu, f.zanasi\}@ucl.ac.uk}
}
\maketitle              
\begin{abstract}
We develop a fully diagrammatic approach to the theory of finite-state automata, based on reinterpreting their usual state-transition graphical representation as a two-dimensional syntax of string diagrams. In this setting, we are able to provide a sound and complete equational theory for language equivalence, with two notable features. First, the proposed axiomatisation is finite--- a result which is provably impossible to obtain for the one-dimensional syntax of regular expressions. Second, the Kleene star is a derived concept, as it can be decomposed into more primitive algebraic blocks.
\end{abstract}


\section{Introduction}\label{sec:intro}


Finite-state automata are one of the most studied structures in theoretical computer science, with an illustrious history and roots reaching far beyond, in the work of biologists, psychologists, engineers and mathematicians. Kleene~\cite{kleene1951representation} introduced regular expressions to give finite-state automata an algebraic presentation, motivated by the study of (biological) neural networks~\cite{mcculloch1943logical}. They are the terms freely generated by the following grammar:
\begin{equation}
\label{regexp}
e, f ::= e + f \mid ef \mid e^* \mid 0 \mid 1 \mid a\in A
\end{equation}
Equational properties of regular expressions were studied by Conway~\cite{conway2012regular} who introduced the term \emph{Kleene algebra}: this is an idempotent semiring with an operation $(-)^*$ for iteration, called the (Kleene) star. The equational theory of Kleene algebra is now well-understood, and multiple complete axiomatisations, both for language and relational models, have been given. Crucially, Kleene algebra is not finitely-based: no finite equational theory can appropriately capture the behaviour of the star~\cite{redko1964defining}. Instead, there are purely  equational infinitary axiomatisations~\cite{krob1991complete,bloom1993equational} and Kozen's finitary implicational theory~\cite{kozen1994completeness}.

Since then, much research has been devoted to extending Kleene algebra with additional operations, in order to capture richer patterns of behaviour, useful in program verification. Examples include conditional branching (Kleene algebra with tests~\cite{kozen1997kleene}, and its recent guarded version~\cite{smolka2019guarded}),  concurrent computation (CKA~\cite{hoare2009concurrent,KappeB0Z18}), and specification of message-passing behaviour in networks (NetKAT~\cite{anderson2014netkat}). 	

The meta-theory of the formalisms above essentially rests on the same three ingredients: (1) given an operational model (e.g., finite-state automata), (2) devise a syntax (regular expressions) that is sufficiently expressive to capture the class of behaviours of the operational model (regular languages), and (3) find a complete axiomatisation (Kleene algebra) for the given semantics.

In this paper, we open up a direct path from (1) to (3). Instead of thinking of automata as a combinatorial model, we formalise them as a bona-fide (two-dimensional) syntax, using the well-established mathematical theory of \emph{string diagrams} and monoidal categories~\cite{Selinger2009}. This approach lets us axiomatise the behaviour of automata directly, freeing us from the necessity of compressing them down to a one-dimensional notation like regular expressions. 

This perspective not only sheds new light on a venerable topic, but has significant consequences. First, as our most important contribution, we are able to provide a \emph{finite and purely equational} axiomatisation of finite-state automata, up to language equivalence. Intriguingly, this does not contradict the impossibility of finding a finite basis for Kleene algebra, as the algebraic setting is different: our result gives a finite presentation as a symmetric monoidal category, while the impossibility result prevents any such presentation to exist as an algebraic theory (in the standard sense). In other words, there is no finite axiomatisation based on terms (\emph{tree}-like structures), but we demonstrate that there is one based on string diagrams (\emph{graph}-like structures).

Secondly, embracing the two-dimensional nature of automata guarantees a strong form of compositionality, that the one-dimensional syntax of regular expressions does not have. In the string diagrammatic setting, automata may have multiple inputs and outputs and, as a result, can be decomposed into subcomponents that retain a meaningful interpretation. For example, if we split the automata below left, the resulting components are still valid string diagrams within our syntax, below right:
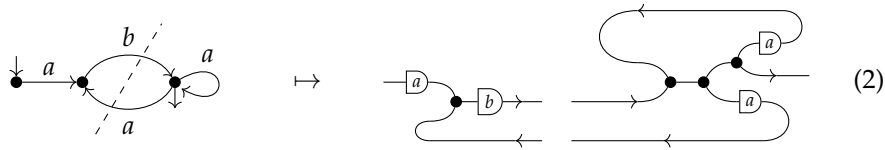
\begin{equation}\label{ex:decompose-automaton}

\tikzset{x=1em, y=2.1ex}
\InputIfFileExists{ex-automaton-graph-split.tikz}{}{\input{./tikz/ex-automaton-graph-split.tikz}}
\tikzset{x=1em, y=1.5ex}
 \qquad\mapsto\qquad 
\tikzset{x=1em, y=2.1ex}
\InputIfFileExists{ex-automaton-diagram-multiple.tikz}{}{\input{./tikz/ex-automaton-diagram-multiple.tikz}}
\tikzset{x=1em, y=1.5ex}
\quad 
\tikzset{x=1em, y=2.1ex}
\InputIfFileExists{ex-automaton-diagram-multiple-1.tikz}{}{\input{./tikz/ex-automaton-diagram-multiple-1.tikz}}
\tikzset{x=1em, y=1.5ex}

\end{equation}
In line with the compositional approach, it is significant that the Kleene star can be decomposed into more elementary building blocks (which come together to form a feedback loop):
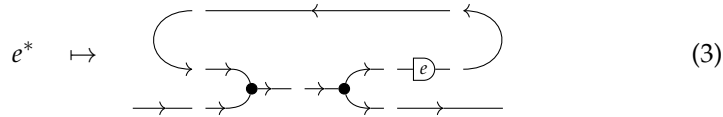
\begin{equation}\label{eq:star-decomposed}
e^* \quad \mapsto\quad 
\tikzset{x=1em, y=2.1ex}
\InputIfFileExists{star-decomposed.tikz}{}{\input{./tikz/star-decomposed.tikz}}
\tikzset{x=1em, y=1.5ex}

\end{equation}
This property opens up for interesting possibilities when studying extensions of Kleene algebra within the same approach. We elaborate further on this in the Discussion (Section~\ref{sec:conclusion}). 

Finally, we believe our proof of completeness is of independent interest, as it relies on fully diagrammatic reformulation of Brzozowski's minimisation procedure~\cite{brzozowski1962canonical}. In the string diagrammatic setting, the symmetries of the equational theory give this procedure a particularly elegant and simple form. Because all of the axioms involved in the determinisation procedure come with a dual, a co-determinisation procedure can be defined immediately by simply reversing the former. This reduces the proof of completeness to determinisation. 

\medskip

We should also note that this is not the first time that automata and regular languages are recast into a categorical mould. The \emph{iteration theories}~\cite{bloom1993iteration} of Bloom and {\'E}sik, the \emph{sharing graphs}~\cite{Hasegawa97recursionfrom} of Hasegawa or the \emph{Network algebras}~\cite{stefanescu2000network} of Stefanescu are all categorical frameworks designed to reason about iteration or recursion, that have found fruitful applications in this domain. They are based on a notion of parameterised fixed-point operation which defines a categorical \emph{trace} in the sense of~\cite{Joyal_tracedcategories}. While our proposal bears resemblance to (and is inspired by) this prior work, it goes beyond in one fundamental aspect: it is the first to give a \emph{finite} complete axiomatisation of automata up to language equivalence. 

A second difference is methodological: our syntax (see~\eqref{eq:syn1} below) does not feature any primitive for iteration or recursion. In particular, the Kleene star is a derivative concept, in the sense that it is decomposable into more elementary operation~\eqref{eq:star-decomposed}. Categorically, our starting point is a compact-closed rather than a traced category. 

 We elaborate further on the relation between our paper and existing work in the Discussion (Section~\ref{sec:conclusion}).

\paragraph{Outline.} Section~\ref{sec:syntax-semantics} lays out the diagrammatic syntax and its semantics. Section~\ref{sec:axioms} introduces the equational theory that we will prove complete. In Section~\ref{sec:encoding} we explain the precise link between our syntax and the traditional formalisms of regular expressions and finite-state automata. We also show how a simple change to the syntax captures context-free languages. Section~\ref{sec:completeness} is dedicated to the proof of completeness. We rely on a normal form argument, which implements Brzozowski's minimisation algorithm equationally and whose main ingredient is a determinisation procedure for diagrams. Omitted proofs can be found in Appendix.

\section{Syntax and semantics}\label{sec:syntax-semantics}

\paragraph{Syntax.} Let us fix an alphabet $\Sigma$ of letters $a \in \Sigma$. We call $\Aut$ the symmetric strict monoidal category freely generated by the following objects and morphisms:
\begin{itemize}
\item three generating objects $\objred$ (`action'), $\objr$ (`right') and $\objl$ (`left') with their identity morphisms depicted respectively as
\begin{equation}\label{eq:syn1}
	\kaId \qquad \arrowright \qquad \arrowleft
\end{equation}
\item the following generating morphisms, depicted as \emph{string diagrams}~\cite{Selinger2009}:
\begin{equation}
\label{eq:syn2}
\begin{gathered}

\tikzset{x=1em, y=2.1ex}
\InputIfFileExists{ka-copy.tikz}{}{\input{./tikz/ka-copy.tikz}}
\tikzset{x=1em, y=1.5ex}
\quad
\tikzset{x=1em, y=2.1ex}
\begin{tikzpicture}
	\begin{pgfonlayer}{nodelayer}
		\node [style=rn] (37) at (0.75, 0) {};
		\node [style=none] (44) at (-0.5, 0) {};
	\end{pgfonlayer}
	\begin{pgfonlayer}{edgelayer}
		\draw [red] (44.center) to (37);
	\end{pgfonlayer}
\end{tikzpicture}
}
\tikzset{x=1em, y=1.5ex}
\quad 
\tikzset{x=1em, y=2.1ex}
\begin{tikzpicture}
	\begin{pgfonlayer}{nodelayer}
		\node [style=wrn] (37) at (0.75, 0) {};
		\node [style=none] (44) at (2, 0) {};
		\node [style=none] (45) at (-0.5, 0) {};
	\end{pgfonlayer}
	\begin{pgfonlayer}{edgelayer}
		\draw [red] (44.center) to (37);
		\draw [red] (45.center) to (37);
	\end{pgfonlayer}
\end{tikzpicture}
}
\tikzset{x=1em, y=1.5ex}
\quad

\tikzset{x=1em, y=2.1ex}
\InputIfFileExists{ka-product.tikz}{}{\input{./tikz/ka-product.tikz}}
\tikzset{x=1em, y=1.5ex}
\quad 
\tikzset{x=1em, y=2.1ex}
\begin{tikzpicture}
	\begin{pgfonlayer}{nodelayer}
		\node [style=wrn] (37) at (0.75, 0) {};
		\node [style=none] (44) at (2, 0) {};
	\end{pgfonlayer}
	\begin{pgfonlayer}{edgelayer}
		\draw [red] (44.center) to (37);
	\end{pgfonlayer}
\end{tikzpicture}
}
\tikzset{x=1em, y=1.5ex}
\quad 
\tikzset{x=1em, y=2.1ex}
\InputIfFileExists{ka-sum.tikz}{}{\input{./tikz/ka-sum.tikz}}
\tikzset{x=1em, y=1.5ex}
 \quad
\tikzset{x=1em, y=2.1ex}
\begin{tikzpicture}
	\begin{pgfonlayer}{nodelayer}
		\node [style=rn] (37) at (0.75, 0) {};
		\node [style=none] (44) at (2, 0) {};
	\end{pgfonlayer}
	\begin{pgfonlayer}{edgelayer}
		\draw [red] (44.center) to (37);
	\end{pgfonlayer}
\end{tikzpicture}
}
\tikzset{x=1em, y=1.5ex}
\quad \atom{a} \quad(a\in\Sigma)
\\

\tikzset{x=1em, y=2.1ex}
\InputIfFileExists{action.tikz}{}{\input{./tikz/action.tikz}}
\tikzset{x=1em, y=1.5ex}
\quad
\tikzset{x=1em, y=2.1ex}
\InputIfFileExists{lr-copy.tikz}{}{\input{./tikz/lr-copy.tikz}}
\tikzset{x=1em, y=1.5ex}
\quad
\tikzset{x=1em, y=2.1ex}
\begin{tikzpicture}
	\begin{pgfonlayer}{nodelayer}
		\node [style=black] (37) at (0.75, 0) {};
		\node [style=none] (43) at (0.25, 0) {};
		\node [style=none] (44) at (-0.5, 0) {};
	\end{pgfonlayer}
	\begin{pgfonlayer}{edgelayer}
		\draw (43.center) to (37);
		\draw [->] (44.center) to (43.center);
	\end{pgfonlayer}
\end{tikzpicture}
}
\tikzset{x=1em, y=1.5ex}
\quad
\tikzset{x=1em, y=2.1ex}
\InputIfFileExists{lr-merge.tikz}{}{\input{./tikz/lr-merge.tikz}}
\tikzset{x=1em, y=1.5ex}
\quad
\tikzset{x=1em, y=2.1ex}
\begin{tikzpicture}
	\begin{pgfonlayer}{nodelayer}
		\node [style=black] (37) at (-0.5, 0) {};
		\node [style=none] (43) at (0.25, 0) {};
		\node [style=none] (44) at (0.75, 0) {};
	\end{pgfonlayer}
	\begin{pgfonlayer}{edgelayer}
		\draw [->] (37) to (43.center);
		\draw (44.center) to (43.center);
	\end{pgfonlayer}
\end{tikzpicture}
}
\tikzset{x=1em, y=1.5ex}
 \quad
\tikzset{x=1em, y=2.1ex}
\InputIfFileExists{cap-down.tikz}{}{\input{./tikz/cap-down.tikz}}
\tikzset{x=1em, y=1.5ex}
 \quad
\tikzset{x=1em, y=2.1ex}
\InputIfFileExists{cup-down.tikz}{}{\input{./tikz/cup-down.tikz}}
\tikzset{x=1em, y=1.5ex}

\end{gathered}
\end{equation}
\end{itemize}
Freely generating $\Aut$ from these data (usually called a \emph{symmetric monoidal theory}~\cite{ZanasiThesis,bonchi2018deconstructing}) means that morphisms of $\Aut$ will be the string diagrams obtained by pasting together (by sequential composition and monoidal product in $\Aut$) the basic components in \eqref{eq:syn1}-\eqref{eq:syn2}, and then quotienting by the laws of symmetric monoidal categories. For instance, \eqref{eq:star-decomposed} is a morphism of $\Aut$ of type $\objr \to \objr$, and $
\tikzset{x=1em, y=2.1ex}
\InputIfFileExists{action-product.tikz}{}{\input{./tikz/action-product.tikz}}
\tikzset{x=1em, y=1.5ex}
$ is one of type $\objred \objred \objr \ \to \ \objr$.

\paragraph{Semantics.}
We first define the semantics for string diagrams simply as a function, and then discuss how to extend it to a functor from $\Aut$ to another category. Our interpretation maps generating morphisms to relations between regular expressions and languages over $\Sigma$:
\begin{gather}
\label{def:ka-star} \nonumber
\sem{{\kaId}} = \left\{((e, e) \mid e \in \RegExp \right\}
\qquad\sem{
\tikzset{x=1em, y=2.1ex}
\begin{tikzpicture}
	\begin{pgfonlayer}{nodelayer}
		\node [style=wrn] (37) at (0.75, 0) {};
		\node [style=none] (44) at (2, 0) {};
		\node [style=none] (45) at (-0.5, 0) {};
	\end{pgfonlayer}
	\begin{pgfonlayer}{edgelayer}
		\draw [red] (44.center) to (37);
		\draw [red] (45.center) to (37);
	\end{pgfonlayer}
\end{tikzpicture}
}
\tikzset{x=1em, y=1.5ex}
} = \left\{(e, e^*) \mid e \in \RegExp \right\}
\\
\label{def:ka-copy} \nonumber
\sem{
\tikzset{x=1em, y=2.1ex}
\InputIfFileExists{ka-copy.tikz}{}{\input{./tikz/ka-copy.tikz}}
\tikzset{x=1em, y=1.5ex}
\;} = \left\{\big(e, (e,e)\big) \mid e \in \RegExp \right\} \qquad \sem{
\tikzset{x=1em, y=2.1ex}
\begin{tikzpicture}
	\begin{pgfonlayer}{nodelayer}
		\node [style=rn] (37) at (0.75, 0) {};
		\node [style=none] (44) at (-0.5, 0) {};
	\end{pgfonlayer}
	\begin{pgfonlayer}{edgelayer}
		\draw [red] (44.center) to (37);
	\end{pgfonlayer}
\end{tikzpicture}
}
\tikzset{x=1em, y=1.5ex}
\;} = \left\{(e, \bullet) \mid e \in \RegExp\right\}
\\
\label{def:ka-product} \nonumber
\sem{
\tikzset{x=1em, y=2.1ex}
\InputIfFileExists{ka-product.tikz}{}{\input{./tikz/ka-product.tikz}}
\tikzset{x=1em, y=1.5ex}
} = \left\{((e, f), ef) \mid e,f \in \RegExp \right\} \quad \sem{\,
\tikzset{x=1em, y=2.1ex}
\begin{tikzpicture}
	\begin{pgfonlayer}{nodelayer}
		\node [style=wrn] (37) at (0.75, 0) {};
		\node [style=none] (44) at (2, 0) {};
	\end{pgfonlayer}
	\begin{pgfonlayer}{edgelayer}
		\draw [red] (44.center) to (37);
	\end{pgfonlayer}
\end{tikzpicture}
}
\tikzset{x=1em, y=1.5ex}
} = \left\{(\bullet, 1)\right\}\quad \sem{\,\atom{a}\,} = \big\{(\bullet, a)\big\}
\\
\label{def:ka-sum} \nonumber
\sem{
\tikzset{x=1em, y=2.1ex}
\InputIfFileExists{ka-sum.tikz}{}{\input{./tikz/ka-sum.tikz}}
\tikzset{x=1em, y=1.5ex}
} = \left\{((e, f), e+f) \mid e,f \in \RegExp\right\} \qquad \sem{\,
\tikzset{x=1em, y=2.1ex}
\begin{tikzpicture}
	\begin{pgfonlayer}{nodelayer}
		\node [style=rn] (37) at (0.75, 0) {};
		\node [style=none] (44) at (2, 0) {};
	\end{pgfonlayer}
	\begin{pgfonlayer}{edgelayer}
		\draw [red] (44.center) to (37);
	\end{pgfonlayer}
\end{tikzpicture}
}
\tikzset{x=1em, y=1.5ex}
} = \left\{(\bullet, 0)\right\}
\\
\label{def:lr-copy} \nonumber
\sem{
\tikzset{x=1em, y=2.1ex}
\InputIfFileExists{lr-copy.tikz}{}{\input{./tikz/lr-copy.tikz}}
\tikzset{x=1em, y=1.5ex}
} = \left\{\big(L, (K_1,K_2)\big)\mid L\subseteq K_i,\, i=1,2 \text{ and } L,K_1, K_2 \subseteq \Sigma^{\star} \right\}
\\
\label{def:cup} \nonumber
\quad \sem{
\tikzset{x=1em, y=2.1ex}
\begin{tikzpicture}
	\begin{pgfonlayer}{nodelayer}
		\node [style=black] (37) at (0.75, 0) {};
		\node [style=none] (43) at (0.25, 0) {};
		\node [style=none] (44) at (-0.5, 0) {};
	\end{pgfonlayer}
	\begin{pgfonlayer}{edgelayer}
		\draw (43.center) to (37);
		\draw [->] (44.center) to (43.center);
	\end{pgfonlayer}
\end{tikzpicture}
}
\tikzset{x=1em, y=1.5ex}
\;} = \left\{(L, \bullet)\mid L \subseteq \Sigma^{\star} \right\} \qquad
\sem{
\tikzset{x=1em, y=2.1ex}
\InputIfFileExists{cup-down.tikz}{}{\input{./tikz/cup-down.tikz}}
\tikzset{x=1em, y=1.5ex}
} = \left\{(\bullet, (L,K)) \mid L\subseteq K \mid L,K \subseteq \Sigma^{\star} \right\}
\\
\label{def:blackid} \nonumber
\sem{\arrowright} = \left\{((L, K), L\subseteq K)\mid L,K \subseteq \Sigma^{\star} \right\}\quad \sem{\arrowleft} = \left\{((L, K), K\subseteq L)\mid L,K \subseteq \Sigma^{\star} \right\}
\\
\label{def:ka-action}
\sem{
\tikzset{x=1em, y=2.1ex}
\InputIfFileExists{action.tikz}{}{\input{./tikz/action.tikz}}
\tikzset{x=1em, y=1.5ex}
} = \left\{((e, L), K)\mid \semreg{e}L\subseteq K \text{ and } e \in \RegExp, L, K \subseteq \Sigma^{\star} \right\}
\end{gather}
and the converse relations for the mirror black generators. In \eqref{def:ka-action}, the semantics $\semreg{e} \in 2^{A^*}$ of a regular expression $e \in \RegExp$ is defined inductively on $e$ (see~\eqref{regexp}), in the standard way:
\begin{align*}
\semreg{e+f} = \semreg{e} \cup \semreg{f} \quad 
\semreg{ef} = \{vw \mid v\in \semreg{e}, w\in \semreg{f}\} \\ 
\semreg{1} = \{\varepsilon\} \qquad \semreg{0} = \varnothing 
\qquad \semreg{a} =  \{a\}   \qquad 
\semreg{e^*} = \bigcup_{n\in\N}\semreg{e}^n \quad 
\end{align*}
where $\semreg{e}^{n+1} := \semreg{e}\cdot \semreg{e}^{n}$ and $\semreg{e}^0 = \{\varepsilon\}$.
The semantics highlights the different roles played by red and black generators. In a nutshell, red generators stand for regular expressions ($
\tikzset{x=1em, y=2.1ex}
\InputIfFileExists{ka-sum.tikz}{}{\input{./tikz/ka-sum.tikz}}
\tikzset{x=1em, y=1.5ex}
$ the sum, $
\tikzset{x=1em, y=2.1ex}
\begin{tikzpicture}
	\begin{pgfonlayer}{nodelayer}
		\node [style=rn] (37) at (0.75, 0) {};
		\node [style=none] (44) at (2, 0) {};
	\end{pgfonlayer}
	\begin{pgfonlayer}{edgelayer}
		\draw [red] (44.center) to (37);
	\end{pgfonlayer}
\end{tikzpicture}
}
\tikzset{x=1em, y=1.5ex}
$ is $0$, $
\tikzset{x=1em, y=2.1ex}
\InputIfFileExists{ka-product.tikz}{}{\input{./tikz/ka-product.tikz}}
\tikzset{x=1em, y=1.5ex}
$ the product, $
\tikzset{x=1em, y=2.1ex}
\begin{tikzpicture}
	\begin{pgfonlayer}{nodelayer}
		\node [style=wrn] (37) at (0.75, 0) {};
		\node [style=none] (44) at (2, 0) {};
	\end{pgfonlayer}
	\begin{pgfonlayer}{edgelayer}
		\draw [red] (44.center) to (37);
	\end{pgfonlayer}
\end{tikzpicture}
}
\tikzset{x=1em, y=1.5ex}
$ is $1$, $
\tikzset{x=1em, y=2.1ex}
\begin{tikzpicture}
	\begin{pgfonlayer}{nodelayer}
		\node [style=wrn] (37) at (0.75, 0) {};
		\node [style=none] (44) at (2, 0) {};
		\node [style=none] (45) at (-0.5, 0) {};
	\end{pgfonlayer}
	\begin{pgfonlayer}{edgelayer}
		\draw [red] (44.center) to (37);
		\draw [red] (45.center) to (37);
	\end{pgfonlayer}
\end{tikzpicture}
}
\tikzset{x=1em, y=1.5ex}
$ the Kleene star,  and $\;\atom{a}$ the letters of $\Sigma$), and black generators for operations on the set of languages ($
\tikzset{x=1em, y=2.1ex}
\InputIfFileExists{lr-copy.tikz}{}{\input{./tikz/lr-copy.tikz}}
\tikzset{x=1em, y=1.5ex}
$ is copy, $
\tikzset{x=1em, y=2.1ex}
\begin{tikzpicture}
	\begin{pgfonlayer}{nodelayer}
		\node [style=black] (37) at (0.75, 0) {};
		\node [style=none] (43) at (0.25, 0) {};
		\node [style=none] (44) at (-0.5, 0) {};
	\end{pgfonlayer}
	\begin{pgfonlayer}{edgelayer}
		\draw (43.center) to (37);
		\draw [->] (44.center) to (43.center);
	\end{pgfonlayer}
\end{tikzpicture}
}
\tikzset{x=1em, y=1.5ex}
$ is delete, $
\tikzset{x=1em, y=2.1ex}
\InputIfFileExists{cup-down.tikz}{}{\input{./tikz/cup-down.tikz}}
\tikzset{x=1em, y=1.5ex}
$ and $
\tikzset{x=1em, y=2.1ex}
\InputIfFileExists{cap-down.tikz}{}{\input{./tikz/cap-down.tikz}}
\tikzset{x=1em, y=1.5ex}
$ feed back outputs into inputs, in a way made more precise later). These two perspectives, which are usually merged, are kept distinct in our approach and only allowed to communicate via 
\tikzset{x=1em, y=2.1ex}
\InputIfFileExists{action.tikz}{}{\input{./tikz/action.tikz}}
\tikzset{x=1em, y=1.5ex}
, which represents the product action of regular expressions (the red wire) on languages. 

In order for this mapping to be functorial from $\Aut$, we now introduce a suitable target semantic category. Interestingly, this will not be the category $\Rel$ of sets and relations: indeed, the identity morphisms $\arrowright$ and $\arrowleft$ are not interpreted as identities of $\Rel$. Instead, the semantic domain will be the category $\BProf$ of \emph{Boolean(-enriched) profunctors}~\cite{fong2018seven} (also called in the literature relational profunctors~\cite{hyland2003glueing} or weakening relations~\cite{moshier2015coherence}).

\begin{definition}\label{def:bool-prof}
Given two preorders $(X, \leq_X)$ and $(Y, \leq_Y)$, a \emph{Boolean profunctor} $R\from X \to Y$ is a relation $R\subseteq X \times Y$ such that if $(x,y)\in R \text{ and } x'\leq_X x,\; y\leq_Y y' \text{ then } (x',y')\in R$.

Preorders and Boolean profunctors form a  symmetric monoidal category $\BProf$ with composition given by relational composition, where the identity for an object $(X, \leq_X)$ is the order relation $\leq_X$ itself, and where the monoidal product is the usual product of preorders. 
\end{definition}

The rich features of our diagrammatic language are reflected in the profunctor interpretation. Indeed, the order relation is built into the wires $\arrowright$ and $\arrowleft$. The two possible directions represent the identities on the ordered set of languages and the same set with the reversed order, respectively. The additional red wire $\kaId$ represents 
the set $\RegExp$ of regular expressions, with \emph{equality} as the associated order relation.\footnote{Note that we can always consider any set with equality as a poset and that, therefore, $\Rel$ is a subcategory of $\BProf$, but not vice-versa, for the simple reason that the identity relation of an arbitrary poset in $\BProf$ is not mapped to the identity relation in $\Rel$.} It is clear that all monochromatic generators satisfy the condition of Definition~\ref{def:bool-prof}. Similarly, the action generator $
\tikzset{x=1em, y=2.1ex}
\InputIfFileExists{action.tikz}{}{\input{./tikz/action.tikz}}
\tikzset{x=1em, y=1.5ex}
$ is a Boolean profunctor: if $((e,L),K)$ are such that $\semreg{e}L\subseteq K$ and $L'\subseteq L$, $K\subseteq K'$ then we have $\semreg{e}L'\subseteq \semreg{e}L\subseteq K\subseteq K'$ by monotony of the product of languages. 
We can conclude that 
\begin{proposition}
The semantics $\sem{\cdot}$ defines a symmetric monoidal functor of type $\Aut \to \BProf$.
\end{proposition}
In particular, because $\Aut$ is free, we can unambiguously assign meaning to any composite diagram from the semantics of its components using composition and the monoidal product in $\BProf$: 
\begin{align*}
\sem{
\tikzset{x=1em, y=2.1ex}
\InputIfFileExists{seq-compose.tikz}{}{\input{./tikz/seq-compose.tikz}}
\tikzset{x=1em, y=1.5ex}
} &= \left\{(L,K)\mid \exists M \, (L,M)\in\sem{\smalldiag{c}}, (M,K)\in\sem{\smalldiag{{\scriptstyle d}}}\right\}\\\sem{
\tikzset{x=1em, y=2.1ex}
\InputIfFileExists{par-compose.tikz}{}{\input{./tikz/par-compose.tikz}}
\tikzset{x=1em, y=1.5ex}
} &= \left\{\big((L_1,L_2),(K_1, K_2)\big)\mid (L_i,K_i)\in\sem{\smalldiag{{\scriptstyle c_i}}},\, i=1,2\right\} 
\end{align*}

\begin{example}
We include here a worked out example to show how to compute the behaviour of a composite diagram which, as we will see, represents the action by concatenation of the regular language $a^*$. 
\begin{equation*}
\sem{
\tikzset{x=1em, y=2.1ex}
\InputIfFileExists{action-star-a.tikz}{}{\input{./tikz/action-star-a.tikz}}
\tikzset{x=1em, y=1.5ex}
} =  \{(L,K)\mid \exists M,N,O\text{ s.t. } \, L, O\subseteq N,\; \semreg{a}M \subseteq O,\; N\subseteq M, K\}
\end{equation*}
where $O$ is a variable name assigned to the top wire of the feedback loop, $M$ to the output wire of the action node, and $N$ is the name assigned to the wire joining $
\tikzset{x=1em, y=2.1ex}
\InputIfFileExists{lr-merge.tikz}{}{\input{./tikz/lr-merge.tikz}}
\tikzset{x=1em, y=1.5ex}
$ to $
\tikzset{x=1em, y=2.1ex}
\InputIfFileExists{lr-copy.tikz}{}{\input{./tikz/lr-copy.tikz}}
\tikzset{x=1em, y=1.5ex}
$. Since $\semreg{a} = \{a\}$ we continue with
\begin{align*}
& =  \{(L,K)\mid \exists M,N,O\text{ s.t. } \, L, O\subseteq N,\; aM \subseteq O,\; N\subseteq M, K\}\\
& =  \{(L,K)\mid \exists M,O\text{ s.t. } \,  aM \subseteq O,\;L, O\subseteq M,\; L,O\subseteq K\}\\
& =  \{(L,K)\mid \exists M\text{ s.t. } \,  aM \subseteq M,\;L\subseteq M,\; L,M\subseteq K\}\\
& =  \{(L,K)\mid \exists M\text{ s.t. } \,  L\cup aM \subseteq M,\; L,M\subseteq K\}\\
& =  \{(L,K)\mid \exists M\text{ s.t. } \,  a^*L \subseteq M,\; L,M\subseteq K\}\\
& =  \{(L,K)\mid\,  a^*L \subseteq K\}
\end{align*}
where the penultimate step is an application of Arden lemma~\cite{arden1961delayed}: $a^*L$ is the least solution of the language inequality $L\cup aX\subseteq X$. 
\end{example}

\section{Equational theory}\label{sec:axioms}

\begin{figure}
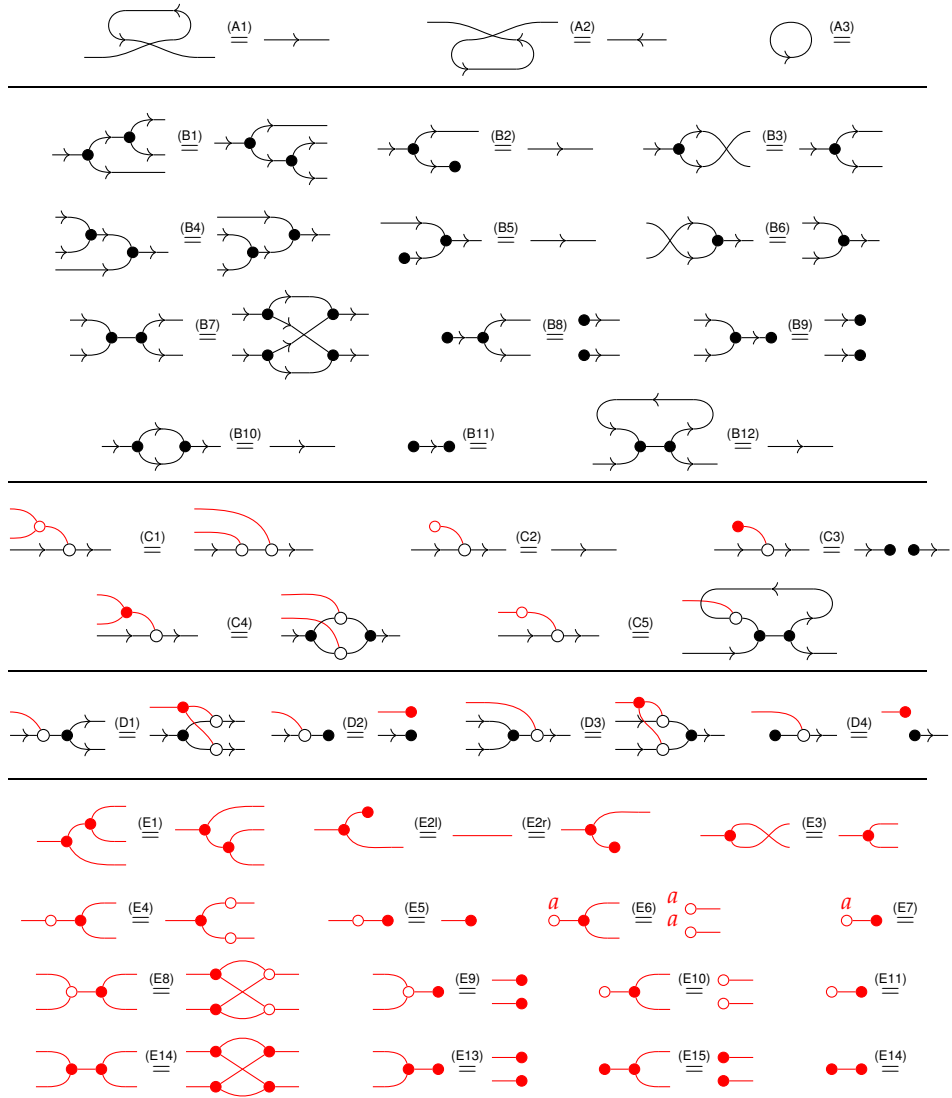

\vspace{-1cm}
\begin{equation*}\label{ax:compact-closed}

\tikzset{x=1em, y=2.1ex}
\InputIfFileExists{yanking-bis.tikz}{}{\input{./tikz/yanking-bis.tikz}}
\tikzset{x=1em, y=1.5ex}
\;\myeq{A1}\;\arrowright\quad \qquad \quad
\tikzset{x=1em, y=2.1ex}
\InputIfFileExists{yanking.tikz}{}{\input{./tikz/yanking.tikz}}
\tikzset{x=1em, y=1.5ex}
\;\myeq{A2}\;\arrowleft\quad \qquad \quad 
\tikzset{x=1em, y=2.1ex}
\InputIfFileExists{loop.tikz}{}{\input{./tikz/loop.tikz}}
\tikzset{x=1em, y=1.5ex}
\;\myeq{A3} 
\end{equation*}

\hrule

\begin{equation*}\label{ax:rel}

\tikzset{x=1em, y=2.1ex}
\InputIfFileExists{co-associativity.tikz}{}{\input{./tikz/co-associativity.tikz}}
\tikzset{x=1em, y=1.5ex}
 \; \myeq{B1} \; 
\tikzset{x=1em, y=2.1ex}
\InputIfFileExists{co-associativity-1.tikz}{}{\input{./tikz/co-associativity-1.tikz}}
\tikzset{x=1em, y=1.5ex}
 \qquad 
\tikzset{x=1em, y=2.1ex}
\InputIfFileExists{right-co-unitality.tikz}{}{\input{./tikz/right-co-unitality.tikz}}
\tikzset{x=1em, y=1.5ex}
 \; \myeq{B2} \; \arrowright  \qquad 
\tikzset{x=1em, y=2.1ex}
\InputIfFileExists{co-commutativity.tikz}{}{\input{./tikz/co-commutativity.tikz}}
\tikzset{x=1em, y=1.5ex}
\; \myeq{B3}\; 
\tikzset{x=1em, y=2.1ex}
\InputIfFileExists{large-copy.tikz}{}{\input{./tikz/large-copy.tikz}}
\tikzset{x=1em, y=1.5ex}

\end{equation*}
\begin{equation*}

\tikzset{x=1em, y=2.1ex}
\InputIfFileExists{associativity.tikz}{}{\input{./tikz/associativity.tikz}}
\tikzset{x=1em, y=1.5ex}
 \; \myeq{B4} \; 
\tikzset{x=1em, y=2.1ex}
\InputIfFileExists{associativity-1.tikz}{}{\input{./tikz/associativity-1.tikz}}
\tikzset{x=1em, y=1.5ex}
 \qquad  
\tikzset{x=1em, y=2.1ex}
\InputIfFileExists{right-unitality.tikz}{}{\input{./tikz/right-unitality.tikz}}
\tikzset{x=1em, y=1.5ex}
\; \myeq{B5} \; \arrowright \qquad 
\tikzset{x=1em, y=2.1ex}
\InputIfFileExists{commutativity.tikz}{}{\input{./tikz/commutativity.tikz}}
\tikzset{x=1em, y=1.5ex}
\; \myeq{B6}\; 
\tikzset{x=1em, y=2.1ex}
\InputIfFileExists{large-merge.tikz}{}{\input{./tikz/large-merge.tikz}}
\tikzset{x=1em, y=1.5ex}

\end{equation*}
\begin{equation*}

\tikzset{x=1em, y=2.1ex}
\InputIfFileExists{bimonoid.tikz}{}{\input{./tikz/bimonoid.tikz}}
\tikzset{x=1em, y=1.5ex}
\; \myeq{B7} \;
\tikzset{x=1em, y=2.1ex}
\InputIfFileExists{bimonoid-1.tikz}{}{\input{./tikz/bimonoid-1.tikz}}
\tikzset{x=1em, y=1.5ex}
 \qquad \quad
\tikzset{x=1em, y=2.1ex}
\InputIfFileExists{copy-co-delete.tikz}{}{\input{./tikz/copy-co-delete.tikz}}
\tikzset{x=1em, y=1.5ex}
\; \myeq{B8} \;
\tikzset{x=1em, y=2.1ex}
\InputIfFileExists{copy-co-delete-1.tikz}{}{\input{./tikz/copy-co-delete-1.tikz}}
\tikzset{x=1em, y=1.5ex}
\qquad\quad
\tikzset{x=1em, y=2.1ex}
\InputIfFileExists{merge-delete.tikz}{}{\input{./tikz/merge-delete.tikz}}
\tikzset{x=1em, y=1.5ex}
\; \myeq{B9} \;
\tikzset{x=1em, y=2.1ex}
\InputIfFileExists{merge-delete-1.tikz}{}{\input{./tikz/merge-delete-1.tikz}}
\tikzset{x=1em, y=1.5ex}

\end{equation*}
\begin{equation*}
 
\tikzset{x=1em, y=2.1ex}
\InputIfFileExists{idempotence.tikz}{}{\input{./tikz/idempotence.tikz}}
\tikzset{x=1em, y=1.5ex}
\; \myeq{B10} \;\arrowright \qquad \quad 
\tikzset{x=1em, y=2.1ex}
\InputIfFileExists{bone.tikz}{}{\input{./tikz/bone.tikz}}
\tikzset{x=1em, y=1.5ex}
\; \myeq{B11} \quad \quad \qquad 
\tikzset{x=1em, y=2.1ex}
\InputIfFileExists{feedback.tikz}{}{\input{./tikz/feedback.tikz}}
\tikzset{x=1em, y=1.5ex}
\; \myeq{B12} \; \arrowright
\end{equation*}

\hrule

\begin{equation*}\label{ax:action}

\tikzset{x=1em, y=2.1ex}
\InputIfFileExists{action-product.tikz}{}{\input{./tikz/action-product.tikz}}
\tikzset{x=1em, y=1.5ex}
 \quad \myeq{C1} \quad 
\tikzset{x=1em, y=2.1ex}
\InputIfFileExists{action-product-1.tikz}{}{\input{./tikz/action-product-1.tikz}}
\tikzset{x=1em, y=1.5ex}
 \qquad\qquad  
\tikzset{x=1em, y=2.1ex}
\InputIfFileExists{action-1.tikz}{}{\input{./tikz/action-1.tikz}}
\tikzset{x=1em, y=1.5ex}
 \ \myeq{C2} \ \arrowright \qquad\qquad 
\tikzset{x=1em, y=2.1ex}
\InputIfFileExists{action-0.tikz}{}{\input{./tikz/action-0.tikz}}
\tikzset{x=1em, y=1.5ex}
 \  \myeq{C3} \ 
\tikzset{x=1em, y=2.1ex}
\InputIfFileExists{delete-zero.tikz}{}{\input{./tikz/delete-zero.tikz}}
\tikzset{x=1em, y=1.5ex}

\end{equation*}
\begin{equation*}

\tikzset{x=1em, y=2.1ex}
\InputIfFileExists{action-sum.tikz}{}{\input{./tikz/action-sum.tikz}}
\tikzset{x=1em, y=1.5ex}
 \quad \myeq{C4} \quad 
\tikzset{x=1em, y=2.1ex}
\InputIfFileExists{action-sum-1.tikz}{}{\input{./tikz/action-sum-1.tikz}}
\tikzset{x=1em, y=1.5ex}
 \qquad\qquad  
\tikzset{x=1em, y=2.1ex}
\InputIfFileExists{action-star.tikz}{}{\input{./tikz/action-star.tikz}}
\tikzset{x=1em, y=1.5ex}
 \quad \myeq{C5} \quad 
\tikzset{x=1em, y=2.1ex}
\InputIfFileExists{action-star-1.tikz}{}{\input{./tikz/action-star-1.tikz}}
\tikzset{x=1em, y=1.5ex}

\end{equation*}

\hrule

\begin{equation*}\label{ax:action-homomorphism}

\tikzset{x=1em, y=2.1ex}
\InputIfFileExists{action-copy.tikz}{}{\input{./tikz/action-copy.tikz}}
\tikzset{x=1em, y=1.5ex}
 \  \myeq{D1} \  
\tikzset{x=1em, y=2.1ex}
\InputIfFileExists{action-copy-1.tikz}{}{\input{./tikz/action-copy-1.tikz}}
\tikzset{x=1em, y=1.5ex}
\quad  
\tikzset{x=1em, y=2.1ex}
\InputIfFileExists{action-delete.tikz}{}{\input{./tikz/action-delete.tikz}}
\tikzset{x=1em, y=1.5ex}
 \  \myeq{D2} \  
\tikzset{x=1em, y=2.1ex}
\InputIfFileExists{action-delete-1.tikz}{}{\input{./tikz/action-delete-1.tikz}}
\tikzset{x=1em, y=1.5ex}

\qquad

\tikzset{x=1em, y=2.1ex}
\InputIfFileExists{action-merge.tikz}{}{\input{./tikz/action-merge.tikz}}
\tikzset{x=1em, y=1.5ex}
 \ \myeq{D3} \ 
\tikzset{x=1em, y=2.1ex}
\InputIfFileExists{action-merge-1.tikz}{}{\input{./tikz/action-merge-1.tikz}}
\tikzset{x=1em, y=1.5ex}
\quad  
\tikzset{x=1em, y=2.1ex}
\InputIfFileExists{action-generate.tikz}{}{\input{./tikz/action-generate.tikz}}
\tikzset{x=1em, y=1.5ex}
 \ \myeq{D4} \ 
\tikzset{x=1em, y=2.1ex}
\InputIfFileExists{action-generate-1.tikz}{}{\input{./tikz/action-generate-1.tikz}}
\tikzset{x=1em, y=1.5ex}

\end{equation*}

\hrule

\begin{equation*}\label{ax:regexp}

\tikzset{x=1em, y=2.1ex}
\InputIfFileExists{ka-copy-co-assoc.tikz}{}{\input{./tikz/ka-copy-co-assoc.tikz}}
\tikzset{x=1em, y=1.5ex}
\; \myeq{E1}\; 
\tikzset{x=1em, y=2.1ex}
\InputIfFileExists{ka-copy-co-assoc-1.tikz}{}{\input{./tikz/ka-copy-co-assoc-1.tikz}}
\tikzset{x=1em, y=1.5ex}
\qquad 
\tikzset{x=1em, y=2.1ex}
\InputIfFileExists{ka-left-co-unit.tikz}{}{\input{./tikz/ka-left-co-unit.tikz}}
\tikzset{x=1em, y=1.5ex}
\; \myeq{E2l}\; \kaId \; \myeq{E2r}\; 
\tikzset{x=1em, y=2.1ex}
\InputIfFileExists{ka-right-co-unit.tikz}{}{\input{./tikz/ka-right-co-unit.tikz}}
\tikzset{x=1em, y=1.5ex}
\qquad 
\tikzset{x=1em, y=2.1ex}
\InputIfFileExists{ka-co-commut.tikz}{}{\input{./tikz/ka-co-commut.tikz}}
\tikzset{x=1em, y=1.5ex}
\;\myeq{E3}\;
\tikzset{x=1em, y=2.1ex}
\InputIfFileExists{ka-copy.tikz}{}{\input{./tikz/ka-copy.tikz}}
\tikzset{x=1em, y=1.5ex}

\end{equation*}
\begin{equation*}

\tikzset{x=1em, y=2.1ex}
\InputIfFileExists{ka-copy-star.tikz}{}{\input{./tikz/ka-copy-star.tikz}}
\tikzset{x=1em, y=1.5ex}
\; \myeq{E4}\; 
\tikzset{x=1em, y=2.1ex}
\InputIfFileExists{ka-copy-star-1.tikz}{}{\input{./tikz/ka-copy-star-1.tikz}}
\tikzset{x=1em, y=1.5ex}
\qquad\quad 
\tikzset{x=1em, y=2.1ex}
\InputIfFileExists{ka-star-delete.tikz}{}{\input{./tikz/ka-star-delete.tikz}}
\tikzset{x=1em, y=1.5ex}
\; \myeq{E5}\; 
\tikzset{x=1em, y=2.1ex}
\InputIfFileExists{ka-delete.tikz}{}{\input{./tikz/ka-delete.tikz}}
\tikzset{x=1em, y=1.5ex}
 \qquad\quad 
\tikzset{x=1em, y=2.1ex}
\InputIfFileExists{ka-atom-copy.tikz}{}{\input{./tikz/ka-atom-copy.tikz}}
\tikzset{x=1em, y=1.5ex}
\; \myeq{E6}\; 
\tikzset{x=1em, y=2.1ex}
\InputIfFileExists{ka-atom-copy-1.tikz}{}{\input{./tikz/ka-atom-copy-1.tikz}}
\tikzset{x=1em, y=1.5ex}
\qquad\quad\qquad 
\tikzset{x=1em, y=2.1ex}
\InputIfFileExists{ka-atom-delete.tikz}{}{\input{./tikz/ka-atom-delete.tikz}}
\tikzset{x=1em, y=1.5ex}
\; \myeq{E7}
\end{equation*}
\begin{equation*}

\tikzset{x=1em, y=2.1ex}
\InputIfFileExists{ka-copy-product.tikz}{}{\input{./tikz/ka-copy-product.tikz}}
\tikzset{x=1em, y=1.5ex}
\; \myeq{E8}\; 
\tikzset{x=1em, y=2.1ex}
\InputIfFileExists{ka-copy-product-1.tikz}{}{\input{./tikz/ka-copy-product-1.tikz}}
\tikzset{x=1em, y=1.5ex}
\qquad\quad 
\tikzset{x=1em, y=2.1ex}
\InputIfFileExists{ka-product-delete.tikz}{}{\input{./tikz/ka-product-delete.tikz}}
\tikzset{x=1em, y=1.5ex}
\; \myeq{E9}\; 
\tikzset{x=1em, y=2.1ex}
\InputIfFileExists{ka-product-delete-1.tikz}{}{\input{./tikz/ka-product-delete-1.tikz}}
\tikzset{x=1em, y=1.5ex}
\qquad \quad
\tikzset{x=1em, y=2.1ex}
\InputIfFileExists{ka-1-copy.tikz}{}{\input{./tikz/ka-1-copy.tikz}}
\tikzset{x=1em, y=1.5ex}
\; \myeq{E10}\; 
\tikzset{x=1em, y=2.1ex}
\InputIfFileExists{ka-1-copy-1.tikz}{}{\input{./tikz/ka-1-copy-1.tikz}}
\tikzset{x=1em, y=1.5ex}
\qquad\quad 
\tikzset{x=1em, y=2.1ex}
\InputIfFileExists{ka-1-delete.tikz}{}{\input{./tikz/ka-1-delete.tikz}}
\tikzset{x=1em, y=1.5ex}
\; \myeq{E11}
\end{equation*}
\begin{equation*}

\tikzset{x=1em, y=2.1ex}
\InputIfFileExists{ka-sum-copy.tikz}{}{\input{./tikz/ka-sum-copy.tikz}}
\tikzset{x=1em, y=1.5ex}
\; \myeq{E14}\; 
\tikzset{x=1em, y=2.1ex}
\InputIfFileExists{ka-sum-copy-1.tikz}{}{\input{./tikz/ka-sum-copy-1.tikz}}
\tikzset{x=1em, y=1.5ex}
\qquad\quad 
\tikzset{x=1em, y=2.1ex}
\InputIfFileExists{ka-sum-delete.tikz}{}{\input{./tikz/ka-sum-delete.tikz}}
\tikzset{x=1em, y=1.5ex}
\; \myeq{E13}\; 
\tikzset{x=1em, y=2.1ex}
\InputIfFileExists{ka-sum-delete-1.tikz}{}{\input{./tikz/ka-sum-delete-1.tikz}}
\tikzset{x=1em, y=1.5ex}
\qquad \quad
\tikzset{x=1em, y=2.1ex}
\InputIfFileExists{ka-0-copy.tikz}{}{\input{./tikz/ka-0-copy.tikz}}
\tikzset{x=1em, y=1.5ex}
\; \myeq{E15}\; 
\tikzset{x=1em, y=2.1ex}
\InputIfFileExists{ka-0-copy-1.tikz}{}{\input{./tikz/ka-0-copy-1.tikz}}
\tikzset{x=1em, y=1.5ex}
\qquad\quad 
\tikzset{x=1em, y=2.1ex}
\InputIfFileExists{ka-0-delete.tikz}{}{\input{./tikz/ka-0-delete.tikz}}
\tikzset{x=1em, y=1.5ex}
\; \myeq{E14}
\end{equation*}
\caption{Equational theory $\eqKa$ of Kleene action algebra.}\label{fig:axioms}
\vspace{-0.5cm}
\end{figure}

In Figure~\ref{fig:axioms} we introduce $\eqKa$, the (finite) equational theory of \emph{Kleene Action Algebra}, on $\Aut$. It will be later shown to be \emph{complete} for the given semantics. We explain some salient features of $\eqKa$ below.

\begin{itemize}
\item (A1)-(A2) relate $
\tikzset{x=1em, y=2.1ex}
\InputIfFileExists{cap-down.tikz}{}{\input{./tikz/cap-down.tikz}}
\tikzset{x=1em, y=1.5ex}
$ and $
\tikzset{x=1em, y=2.1ex}
\InputIfFileExists{cup-down.tikz}{}{\input{./tikz/cup-down.tikz}}
\tikzset{x=1em, y=1.5ex}
$, allowing us to bend and straighten wires at will. This makes the full subcategory of $\Aut$ on $\objr$ and $\objl$, modulo (A1)-(A2), a \emph{compact closed} category~\cite{kelly1980compactclosed}. (A3) allows us to eliminate isolated loops. Note that the whole category is not compact closed because $\objred$ does not have a dual.
\item The B block states that $
\tikzset{x=1em, y=2.1ex}
\InputIfFileExists{lr-copy.tikz}{}{\input{./tikz/lr-copy.tikz}}
\tikzset{x=1em, y=1.5ex}
,
\tikzset{x=1em, y=2.1ex}
\begin{tikzpicture}
	\begin{pgfonlayer}{nodelayer}
		\node [style=black] (37) at (0.75, 0) {};
		\node [style=none] (43) at (0.25, 0) {};
		\node [style=none] (44) at (-0.5, 0) {};
	\end{pgfonlayer}
	\begin{pgfonlayer}{edgelayer}
		\draw (43.center) to (37);
		\draw [->] (44.center) to (43.center);
	\end{pgfonlayer}
\end{tikzpicture}
}
\tikzset{x=1em, y=1.5ex}
$ forms a cocommutative comonoid (B1)-(B3), while $
\tikzset{x=1em, y=2.1ex}
\InputIfFileExists{lr-merge.tikz}{}{\input{./tikz/lr-merge.tikz}}
\tikzset{x=1em, y=1.5ex}
,
\tikzset{x=1em, y=2.1ex}
\begin{tikzpicture}
	\begin{pgfonlayer}{nodelayer}
		\node [style=black] (37) at (-0.5, 0) {};
		\node [style=none] (43) at (0.25, 0) {};
		\node [style=none] (44) at (0.75, 0) {};
	\end{pgfonlayer}
	\begin{pgfonlayer}{edgelayer}
		\draw [->] (37) to (43.center);
		\draw (44.center) to (43.center);
	\end{pgfonlayer}
\end{tikzpicture}
}
\tikzset{x=1em, y=1.5ex}
$ form a commutative monoid (B4)-(B6). Moreover, $
\tikzset{x=1em, y=2.1ex}
\InputIfFileExists{lr-copy.tikz}{}{\input{./tikz/lr-copy.tikz}}
\tikzset{x=1em, y=1.5ex}
,
\tikzset{x=1em, y=2.1ex}
\begin{tikzpicture}
	\begin{pgfonlayer}{nodelayer}
		\node [style=black] (37) at (0.75, 0) {};
		\node [style=none] (43) at (0.25, 0) {};
		\node [style=none] (44) at (-0.5, 0) {};
	\end{pgfonlayer}
	\begin{pgfonlayer}{edgelayer}
		\draw (43.center) to (37);
		\draw [->] (44.center) to (43.center);
	\end{pgfonlayer}
\end{tikzpicture}
}
\tikzset{x=1em, y=1.5ex}
,
\tikzset{x=1em, y=2.1ex}
\InputIfFileExists{lr-merge.tikz}{}{\input{./tikz/lr-merge.tikz}}
\tikzset{x=1em, y=1.5ex}
,
\tikzset{x=1em, y=2.1ex}
\begin{tikzpicture}
	\begin{pgfonlayer}{nodelayer}
		\node [style=black] (37) at (-0.5, 0) {};
		\node [style=none] (43) at (0.25, 0) {};
		\node [style=none] (44) at (0.75, 0) {};
	\end{pgfonlayer}
	\begin{pgfonlayer}{edgelayer}
		\draw [->] (37) to (43.center);
		\draw (44.center) to (43.center);
	\end{pgfonlayer}
\end{tikzpicture}
}
\tikzset{x=1em, y=1.5ex}
$ form an idempotent bimonoid (B7)-(B11). (B12)  allows us to eliminate trivial feedback loops.

\item The C block axiomatises the behaviour of the action of regular expressions on languages. These laws mimic the usual definition of the action of a semiring on a set, except for (C5) which is novel and captures the interaction with the Kleene star. Here lies a distinctive feature of our theory: the behaviour of the Kleene star is derived from its decomposition as the feedback loop on the right of (C5). 
\item The D block forces the action to be a comonoid ((D1)-(D2)) and monoid ((D1)-(D2)) homomorphism.
\item The E block axiomatises the purely red fragment. Remarkably, these axioms do not describe any of the actual Kleene algebra structure: they just state that $
\tikzset{x=1em, y=2.1ex}
\InputIfFileExists{ka-copy.tikz}{}{\input{./tikz/ka-copy.tikz}}
\tikzset{x=1em, y=1.5ex}
$ and $
\tikzset{x=1em, y=2.1ex}
\begin{tikzpicture}
	\begin{pgfonlayer}{nodelayer}
		\node [style=rn] (37) at (0.75, 0) {};
		\node [style=none] (44) at (-0.5, 0) {};
	\end{pgfonlayer}
	\begin{pgfonlayer}{edgelayer}
		\draw [red] (44.center) to (37);
	\end{pgfonlayer}
\end{tikzpicture}
}
\tikzset{x=1em, y=1.5ex}
$ form a commutative comonoid ((E1)-(E3)) and that all other red generators are comonoid homomorphisms ((E4)-(E15)). This means that the red fragment is actually the \emph{free} (cartesian) algebraic theory (\emph{cf.}~\cite{ZanasiThesis,bonchi2018deconstructing}) on generators $
\tikzset{x=1em, y=2.1ex}
\begin{tikzpicture}
	\begin{pgfonlayer}{nodelayer}
		\node [style=wrn] (37) at (0.75, 0) {};
		\node [style=none] (44) at (2, 0) {};
		\node [style=none] (45) at (-0.5, 0) {};
	\end{pgfonlayer}
	\begin{pgfonlayer}{edgelayer}
		\draw [red] (44.center) to (37);
		\draw [red] (45.center) to (37);
	\end{pgfonlayer}
\end{tikzpicture}
}
\tikzset{x=1em, y=1.5ex}
,

\tikzset{x=1em, y=2.1ex}
\InputIfFileExists{ka-product.tikz}{}{\input{./tikz/ka-product.tikz}}
\tikzset{x=1em, y=1.5ex}
, 
\tikzset{x=1em, y=2.1ex}
\begin{tikzpicture}
	\begin{pgfonlayer}{nodelayer}
		\node [style=wrn] (37) at (0.75, 0) {};
		\node [style=none] (44) at (2, 0) {};
	\end{pgfonlayer}
	\begin{pgfonlayer}{edgelayer}
		\draw [red] (44.center) to (37);
	\end{pgfonlayer}
\end{tikzpicture}
}
\tikzset{x=1em, y=1.5ex}
, 
\tikzset{x=1em, y=2.1ex}
\InputIfFileExists{ka-sum.tikz}{}{\input{./tikz/ka-sum.tikz}}
\tikzset{x=1em, y=1.5ex}
, 
\tikzset{x=1em, y=2.1ex}
\begin{tikzpicture}
	\begin{pgfonlayer}{nodelayer}
		\node [style=rn] (37) at (0.75, 0) {};
		\node [style=none] (44) at (2, 0) {};
	\end{pgfonlayer}
	\begin{pgfonlayer}{edgelayer}
		\draw [red] (44.center) to (37);
	\end{pgfonlayer}
\end{tikzpicture}
}
\tikzset{x=1em, y=1.5ex}
, 
\atom{a} (a\in\Sigma)$, where the remaining generators $
\tikzset{x=1em, y=2.1ex}
\InputIfFileExists{ka-copy.tikz}{}{\input{./tikz/ka-copy.tikz}}
\tikzset{x=1em, y=1.5ex}
$ and $
\tikzset{x=1em, y=2.1ex}
\begin{tikzpicture}
	\begin{pgfonlayer}{nodelayer}
		\node [style=rn] (37) at (0.75, 0) {};
		\node [style=none] (44) at (-0.5, 0) {};
	\end{pgfonlayer}
	\begin{pgfonlayer}{edgelayer}
		\draw [red] (44.center) to (37);
	\end{pgfonlayer}
\end{tikzpicture}
}
\tikzset{x=1em, y=1.5ex}
$ act as copy and discard of variables. 
\end{itemize}

Let $\eqKa$ be the smallest equational theory containing all equations in Fig.~\ref{fig:axioms}. Their \emph{soundness} for the chosen semantics is not difficult to show and, for space reasons, we omit the proof. We now state our \emph{completeness} result, whose proof will be discussed in Section~\ref{sec:completeness}.

\begin{theorem}[Completeness]\label{thm:completeness}
For morphisms $d$, $e$ in $\Aut$
, $d \eqKa e$ if and only if $\sem{d} = \sem{e}$. 
\end{theorem}


\begin{remark}\label{rmk:order}
Some axiomatisations of Kleene algebra make use of a partial order between terms, which can be defined from the idempotent monoid structure: $f\leq e$ iff $e+f = e$. At the semantic level, it corresponds to inclusion of languages. Similarly, using the idempotent bimonoid structure of our equational theory, we can define a partial order on $\objr\to\objr$ diagrams: $f\leq e$ iff $
\tikzset{x=1em, y=2.1ex}
\InputIfFileExists{sum.tikz}{}{\input{./tikz/sum.tikz}}
\tikzset{x=1em, y=1.5ex}
= \scalar{e}$.
This partial order structure can also be extended to all morphisms $\objr^n\to\objr^m$ by using the vertical composition of $n$ copies of $
\tikzset{x=1em, y=2.1ex}
\InputIfFileExists{lr-copy.tikz}{}{\input{./tikz/lr-copy.tikz}}
\tikzset{x=1em, y=1.5ex}
$ and $m$ copies of $
\tikzset{x=1em, y=2.1ex}
\InputIfFileExists{lr-merge.tikz}{}{\input{./tikz/lr-merge.tikz}}
\tikzset{x=1em, y=1.5ex}
$ instead. 
\end{remark}
\begin{remark}\label{rmk:different-monoid}
There are no specific equations relating the atomic actions $\atom{a} \,(a\in\Sigma)$. This is because, as we study finite-state automata, we are interested in the \emph{free} monoid $\Sigma^*$ over $\Sigma$. However, nothing would prevent us from modelling other structures. For instance, free commutative monoids (powers of $\N$), whose rational subsets correspond to semilinear sets~\cite[Chapter 11]{conway2012regular}, would be of particular interest.
\end{remark}


\section{Encoding regular expressions and automata}\label{sec:encoding}

A major appeal of our approach is that both regular expressions and automata can be uniformly represented in the graphical language of string diagrams, and the translation of one into the other becomes an equational derivation in $\eqKa$. In fact, we will see there is a close resemblance between automata and the shape of the string diagrams interpreting them --- the main difference being that string diagrams are \emph{composable} structures. 

In this section we describe how regular expressions (resp. automata) can be encoded as string diagrams, such that their semantics corresponds in a precise way to the languages that they describe (resp. recognise). 
 
 To define these encodings, it is convenient to introduce the following syntactic sugar. For any regular expression $e$, one may always construct a `red' string diagram $\diagregexp{e} \colon 0 \to \objred$ such that  $\sem{\,\diagregexp{e}} = \{(\bullet,e)\}$. We will write $\scalar{e}$ for its composite with the action, as defined below left, with the particular case of a letter $a\in\Sigma$ on the right:
\begin{equation}\label{def:scalar}
\scalar{e} := 
\tikzset{x=1em, y=2.1ex}
\InputIfFileExists{regexp-action.tikz}{}{\input{./tikz/regexp-action.tikz}}
\tikzset{x=1em, y=1.5ex}
 \qquad\qquad\qquad \scalar{a} := 
\tikzset{x=1em, y=2.1ex}
\InputIfFileExists{a-action.tikz}{}{\input{./tikz/a-action.tikz}}
\tikzset{x=1em, y=1.5ex}

\end{equation} 

\subsection{From regular expressions to string diagrams}

In a sense, regular expressions are already part of the graphical syntax, as the red generators. However, these alone are meaningless, since their image under the semantics is simply the free term algebra $\RegExp$ (see~\eqref{eq:regexpstar}) . They acquire meaning as they \emph{act} on the black wire, whose semantics is the set of languages over $\Sigma$. Using this action, we can inductively define an encoding $\transreg{-}$ of regular expressions into string diagrams of $\Aut$, as follows:  
\begin{align*}
\transreg{e+f}&\; =\;  
\tikzset{x=1em, y=2.1ex}
\InputIfFileExists{sum-action.tikz}{}{\input{./tikz/sum-action.tikz}}
\tikzset{x=1em, y=1.5ex}
  \myeq{C4}_{\scriptscriptstyle KAA}  
\tikzset{x=1em, y=2.1ex}
\InputIfFileExists{sum.tikz}{}{\input{./tikz/sum.tikz}}
\tikzset{x=1em, y=1.5ex}
\quad &\transreg{0} \;=\; 
\tikzset{x=1em, y=2.1ex}
\InputIfFileExists{action-0.tikz}{}{\input{./tikz/action-0.tikz}}
\tikzset{x=1em, y=1.5ex}
   \myeq{C3}_{\scriptscriptstyle KAA}  
\tikzset{x=1em, y=2.1ex}
\begin{tikzpicture}
	\begin{pgfonlayer}{nodelayer}
		\node [style=black] (0) at (0.5, 0) {};
		\node [style=black] (1) at (-0.5, 0) {};
		\node [style=none] (2) at (-1.75, 0) {};
		\node [style=none] (3) at (1.75, 0) {};
	\end{pgfonlayer}
	\begin{pgfonlayer}{edgelayer}
		\draw (1) to (2.center);
		\draw (3.center) to (0);
	\end{pgfonlayer}
\end{tikzpicture}
}
\tikzset{x=1em, y=1.5ex}
  \nonumber\\
\transreg{ef}&\;= \;
\tikzset{x=1em, y=2.1ex}
\InputIfFileExists{product-action.tikz}{}{\input{./tikz/product-action.tikz}}
\tikzset{x=1em, y=1.5ex}
    \myeq{C1}_{\scriptscriptstyle KAA}   
\tikzset{x=1em, y=2.1ex}
\InputIfFileExists{product.tikz}{}{\input{./tikz/product.tikz}}
\tikzset{x=1em, y=1.5ex}
  &\transreg{1} \;=\; 
\tikzset{x=1em, y=2.1ex}
\InputIfFileExists{action-1.tikz}{}{\input{./tikz/action-1.tikz}}
\tikzset{x=1em, y=1.5ex}
   \myeq{C2}_{\scriptscriptstyle KAA}   
\tikzset{x=1em, y=2.1ex}
\begin{tikzpicture}
	\begin{pgfonlayer}{nodelayer}
		\node [style=none] (0) at (-1.5, 0) {};
		\node [style=none] (1) at (0, 0) {};
		\node [style=none] (2) at (1.25, 0) {};
	\end{pgfonlayer}
	\begin{pgfonlayer}{edgelayer}
		\draw [->] (0.center) to (1.center);
		\draw (1.center) to (2.center);
	\end{pgfonlayer}
\end{tikzpicture}
}
\tikzset{x=1em, y=1.5ex}
 \nonumber
\end{align*}
\begin{equation}\label{eq:regexpstar}
\transreg{e^*}\; =\;
\tikzset{x=1em, y=2.1ex}
\InputIfFileExists{star-action.tikz}{}{\input{./tikz/star-action.tikz}}
\tikzset{x=1em, y=1.5ex}
   \myeq{C5}_{\scriptscriptstyle KAA}  
\tikzset{x=1em, y=2.1ex}
\InputIfFileExists{star.tikz}{}{\input{./tikz/star.tikz}}
\tikzset{x=1em, y=1.5ex}
  \quad\; \transreg{a} \;=\; 
\tikzset{x=1em, y=2.1ex}
\InputIfFileExists{a-action.tikz}{}{\input{./tikz/a-action.tikz}}
\tikzset{x=1em, y=1.5ex}
  =:   \scalar{a}  
\end{equation}
For example, $\transreg{ab(a+ab)^*} \; =$
\begin{equation}\label{ex:regexp-diagram}

\tikzset{x=1em, y=2.1ex}
\InputIfFileExists{ex-regexp-red.tikz}{}{\input{./tikz/ex-regexp-red.tikz}}
\tikzset{x=1em, y=1.5ex}
 \;\eqKa\; 
\tikzset{x=1em, y=2.1ex}
\InputIfFileExists{ex-regexp.tikz}{}{\input{./tikz/ex-regexp.tikz}}
\tikzset{x=1em, y=1.5ex}

\end{equation}

\noindent As expected, the translation preserves the language semantics of regular expressions in a sense that the following proposition makes precise. 
\begin{proposition}\label{thm:diagram-regexp}
For any regular expression $e$, $\sem{\transreg{e}} = \left\{(L,K)\mid \semreg{e} L \,  \subseteq K\right\}$.
\end{proposition}

\subsection{From automata to string diagrams...}\label{sec:automata-diag}

Example~\eqref{ex:regexp-diagram} suggests that the string diagram $\transreg{e}$ corresponding to a regular expression $e$ looks a lot like a nondeterministic finite-state automaton (NFA) for $e$. In fact, the translation $\transreg{-}$ can be seen as the diagrammatic counterpart of Thompson's construction~\cite{thompson1968programming} that builds an NFA from a given regular expression. 

Instead of starting from regular expressions, one may translate NFA into string diagrams directly. There are at least two ways to do that. The first is to encode an NFA as the diagrammatic counterpart of its transition relation. The second is to translate directly its combinatorial representation as a graph into the diagrammatic syntax. 

\paragraph{Encoding the transition relation.} This is a simple variant of the  translation of matrices over semirings that has appeared in several places in the literature~\cite{Lack2004a,ZanasiThesis}. 

Let $A$ be an NFA with set of states $Q$, initial state $q_0\in Q$, accepting states $F\subseteq Q$ and transition relation $\delta\subseteq Q\times \Sigma\times Q$.  
We can represent $\delta$
as a string diagram $d$ with $|Q|$ incoming wires on the left and $|Q|$ outgoing wires on the right.
 The left $j$th port of $d$ is connected to the $i$th port on the right through an $\scalar{a}$ whenever $(q_i,a,q_j)\in\delta$. To accommodate nondeterminism, when the same two ports are connected by several different letters of $\Sigma$, we join these using 
\tikzset{x=1em, y=2.1ex}
\InputIfFileExists{lr-copy.tikz}{}{\input{./tikz/lr-copy.tikz}}
\tikzset{x=1em, y=1.5ex}
 and 
\tikzset{x=1em, y=2.1ex}
\InputIfFileExists{lr-merge.tikz}{}{\input{./tikz/lr-merge.tikz}}
\tikzset{x=1em, y=1.5ex}
. When $(q_i, \epsilon, q_j)\in\delta$, the two ports are simply connected via a plain identity wire. If there is no tuple in $\delta$ such that $(q_i, a, q_j)\in\delta$ for any $a$, the two corresponding ports are disconnected.

\noindent\begin{minipage}{0.7\textwidth}
	  For example, the transition relation of an NFA with three states and $\delta = \{((q_0, a, q_1), (q_1, b, q_2), (q_2, a, q_1), (q_2, a, q_2))\} $ (disregarding the initial and accepting states for the moment) is depicted on the right. Conversely, given such a diagram, we can recover $\delta$ by collecting $\Sigma$-weighted paths from left to right ports.
\end{minipage}
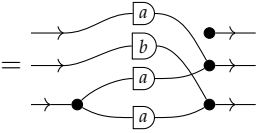
\begin{minipage}{0.3\textwidth}
	\[d = 
\tikzset{x=1em, y=2.1ex}
\InputIfFileExists{ex-matrix-nfa.tikz}{}{\input{./tikz/ex-matrix-nfa.tikz}}
\tikzset{x=1em, y=1.5ex}
\]
\end{minipage}

To deal with the initial state, we add an additional incoming wire connected to the right port corresponding to the initial state of the automaton. Similarly, for accepting states we

\noindent\begin{minipage}{0.7\textwidth} add an additional outgoing wire, connected to the left ports corresponding to each accepting state, via 
\tikzset{x=1em, y=2.1ex}
\InputIfFileExists{lr-merge.tikz}{}{\input{./tikz/lr-merge.tikz}}
\tikzset{x=1em, y=1.5ex}
 if there is more than one. Finally, we trace out the $|Q|$ wires of the diagrammatic transition relation to obtain the associated string diagram. In other words, for a NFA with initial state  $q_0$, set of accepting states $F$, transition relation $\delta$, we obtain the string diagram on the right, where  $d$ is the diagrammatic counterpart
\end{minipage}
\begin{minipage}{0.3\textwidth}
\[
\tikzset{x=1em, y=2.1ex}
\InputIfFileExists{nfa-rep.tikz}{}{\input{./tikz/nfa-rep.tikz}}
\tikzset{x=1em, y=1.5ex}
\]
\end{minipage}	
of $\delta$ as defined above, $e_0$ is the injection of a single wire as the first amongst $|Q|$ wires,
and $f$ deletes all wires that are not associated to states in $F$ with $
\tikzset{x=1em, y=2.1ex}
\begin{tikzpicture}
	\begin{pgfonlayer}{nodelayer}
		\node [style=black] (37) at (0.75, 0) {};
		\node [style=none] (43) at (0.25, 0) {};
		\node [style=none] (44) at (-0.5, 0) {};
	\end{pgfonlayer}
	\begin{pgfonlayer}{edgelayer}
		\draw (43.center) to (37);
		\draw [->] (44.center) to (43.center);
	\end{pgfonlayer}
\end{tikzpicture}
}
\tikzset{x=1em, y=1.5ex}
$, and applies $
\tikzset{x=1em, y=2.1ex}
\InputIfFileExists{lr-merge.tikz}{}{\input{./tikz/lr-merge.tikz}}
\tikzset{x=1em, y=1.5ex}
$ to merge them into a single outgoing wire.

For example, if $A$ with $\delta$ as above has initial state $q_0$ and accepting state $\{q_2\}$, we get the diagram below left; instead, if all states are accepting, we obtain the diagram below right:
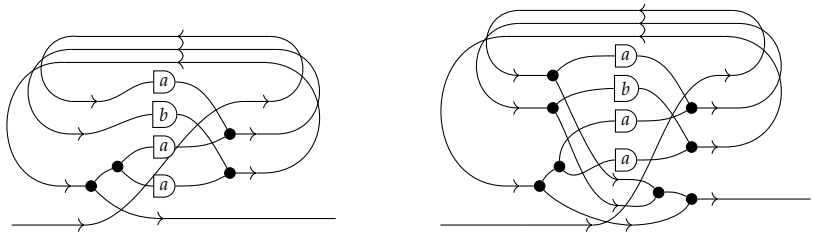
\begin{equation*}\label{eq:nfa-reps}

\tikzset{x=1em, y=2.1ex}
\InputIfFileExists{ex-nfa-rep-single-accepting.tikz}{}{\input{./tikz/ex-nfa-rep-single-accepting.tikz}}
\tikzset{x=1em, y=1.5ex}
 \qquad\quad  
\tikzset{x=1em, y=2.1ex}
\InputIfFileExists{ex-nfa-rep-multiple-accepting.tikz}{}{\input{./tikz/ex-nfa-rep-multiple-accepting.tikz}}
\tikzset{x=1em, y=1.5ex}

\end{equation*}
The correctness of this simple translation is justified by a semantic correspondence between the language recognised by a given NFA $A$ and the denotation of the corresponding string diagram. 
\begin{proposition}\label{thm:nfa-to-diag}
Given an NFA $A$ which recognises the language $L$, let $d_A$ be its associated string diagram, constructed as above. Then $\sem{d_A} = \left\{(K,K')\mid LK \subseteq K'\right\}$. 
\end{proposition}

\paragraph{From graphs to string diagrams.} The second way of translating automata into string diagrams mimics more directly the combinatorial representation of automata. The idea (which should be sufficiently intuitive to not need to be made formal here) is, for each state, to use 
\tikzset{x=1em, y=2.1ex}
\InputIfFileExists{lr-merge.tikz}{}{\input{./tikz/lr-merge.tikz}}
\tikzset{x=1em, y=1.5ex}
 to represent incoming edges, and  
\tikzset{x=1em, y=2.1ex}
\InputIfFileExists{lr-copy.tikz}{}{\input{./tikz/lr-copy.tikz}}
\tikzset{x=1em, y=1.5ex}
 to represent outgoing edges. As above, labels $a \in A$ will be modelled using $\scalar{a}$. For example, the graph and the associated string diagram corresponding with the NFA above are
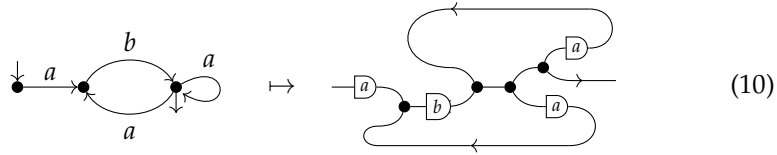
\begin{equation}\label{eq:ex-automaton-diagram}

\tikzset{x=1em, y=2.1ex}
\InputIfFileExists{ex-automaton-graph.tikz}{}{\input{./tikz/ex-automaton-graph.tikz}}
\tikzset{x=1em, y=1.5ex}
 \quad \mapsto \quad 
\tikzset{x=1em, y=2.1ex}
\InputIfFileExists{ex-automaton-diagram.tikz}{}{\input{./tikz/ex-automaton-diagram.tikz}}
\tikzset{x=1em, y=1.5ex}

\end{equation}
Note the initial state of the automaton corresponds to the left interface of the string diagram, and the accepting state to the right interface. As before, when there are multiple accepting states, they all connect to a single right interface, via $
\tikzset{x=1em, y=2.1ex}
\InputIfFileExists{lr-merge.tikz}{}{\input{./tikz/lr-merge.tikz}}
\tikzset{x=1em, y=1.5ex}
$. For example, if we make all states accepting in the automaton above, we get the following diagrammatic representation:
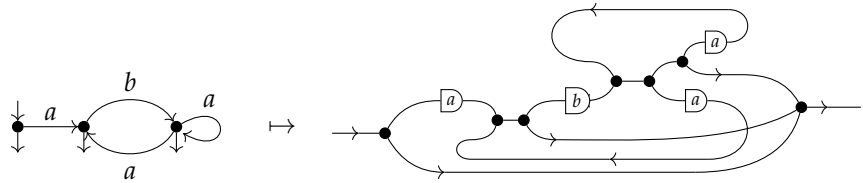
\begin{align*}

\tikzset{x=1em, y=2.1ex}
\InputIfFileExists{ex-automaton-graph-multiple-accept.tikz}{}{\input{./tikz/ex-automaton-graph-multiple-accept.tikz}}
\tikzset{x=1em, y=1.5ex}
\quad & \mapsto \quad 
\tikzset{x=1em, y=2.1ex}
\InputIfFileExists{ex-automaton-diagram-multiple-accept.tikz}{}{\input{./tikz/ex-automaton-diagram-multiple-accept.tikz}}
\tikzset{x=1em, y=1.5ex}

\end{align*}


\subsection{...and back}

The previous discussion shows how NFAs can be seen as string diagrams of type $\objr\to \objr$. The converse is also true: we now show how to extract an automaton from any string diagram $d \colon \objr\to \objr$, such that the language the automaton recognises matches the denotation of $d$.

In order to phrase this correspondence formally, we need to introduce some terminology. We call \emph{left-to-right} those string diagrams whose domain and codomain contain only $\objr$, i.e. their type is of the form $\objr^{n} \to \objr^{m}$. The idea is that, in any such string diagram, the $n$ left interfaces act as \emph{inputs} of the computation, and the $m$ right interfaces act as \emph{outputs}. For instance, \eqref{eq:ex-automaton-diagram} is a left-to-right diagram $\objr \to \objr$.

A string diagram $d$ is \emph{atomic} if the only red generators occurring in $d$ are of the form $\atom{a}$. By \emph{unfolding} all red components $\diagregexp{e}$ in any left-to-right diagram, using axioms (C1)-(C5), we can prove the following statement.
\begin{proposition}\label{thm:atomic}
Any left-to-right diagram is $\eqKa$-equivalent to an atomic one.
\end{proposition} 
For instance, the string diagram on the left of \eqref{ex:regexp-diagram} is $\eqKa$-equivalent to the atomic one on the right.

We call \emph{block} of a certain subset of generators a vertical composite of these generators followed by some permutations of the wires. 
\begin{definition}\label{def:matrix-diagram}
A \emph{matrix-diagram} is a left-to-right diagram that factors as a block of $
\tikzset{x=1em, y=2.1ex}
\InputIfFileExists{lr-copy.tikz}{}{\input{./tikz/lr-copy.tikz}}
\tikzset{x=1em, y=1.5ex}
, 
\tikzset{x=1em, y=2.1ex}
\begin{tikzpicture}
	\begin{pgfonlayer}{nodelayer}
		\node [style=black] (37) at (0.75, 0) {};
		\node [style=none] (43) at (0.25, 0) {};
		\node [style=none] (44) at (-0.5, 0) {};
	\end{pgfonlayer}
	\begin{pgfonlayer}{edgelayer}
		\draw (43.center) to (37);
		\draw [->] (44.center) to (43.center);
	\end{pgfonlayer}
\end{tikzpicture}
}
\tikzset{x=1em, y=1.5ex}
$, followed by a block of $\scalar{a} (a\in \Sigma)$ and finally, a block of $
\tikzset{x=1em, y=2.1ex}
\InputIfFileExists{lr-merge.tikz}{}{\input{./tikz/lr-merge.tikz}}
\tikzset{x=1em, y=1.5ex}
, 
\tikzset{x=1em, y=2.1ex}
\begin{tikzpicture}
	\begin{pgfonlayer}{nodelayer}
		\node [style=black] (37) at (-0.5, 0) {};
		\node [style=none] (43) at (0.25, 0) {};
		\node [style=none] (44) at (0.75, 0) {};
	\end{pgfonlayer}
	\begin{pgfonlayer}{edgelayer}
		\draw [->] (37) to (43.center);
		\draw (44.center) to (43.center);
	\end{pgfonlayer}
\end{tikzpicture}
}
\tikzset{x=1em, y=1.5ex}
$. 
\end{definition}
To each matrix-diagram $d$ we can associate a unique transition relation $\delta$ by gathering paths from each input to each output: $(q_i, a, q_j)\in\delta$ if there is $\scalar{a}$ joining the $i$th input to the $j$th output.

\noindent \begin{minipage}{0.7\textwidth} 
A transition relation is \emph{$\epsilon$-free} if it does not contain the empty word. It is \emph{deterministic} if it is $\epsilon$-free and, for each $i$ and each $a\in\Sigma$ there is at most one $j$ such that $(q_i, a, q_j)\in\delta$. We will apply these terms to matrix-diagrams and the associated transition relation inter-
\end{minipage}
\begin{minipage}{0.3\textwidth}
\begin{center}

\tikzset{x=1em, y=2.1ex}
\InputIfFileExists{ex-matrix-nfa-blocks.tikz}{}{\input{./tikz/ex-matrix-nfa-blocks.tikz}}
\tikzset{x=1em, y=1.5ex}

\end{center}
\end{minipage}
changeably. The example of Section~\ref{sec:automata-diag} above, with the three blocks highlighted, is a matrix-diagram. It is $\epsilon$-free but not deterministic since there are two $a$-labelled transitions starting from the third input.

Given a matrix-diagram $d\from \objr^{l+n}\to \objr^{p+m}$, we will write $d_{ij}$, with $i=l,n$ and $j=p,m$, for the subdiagrams corresponding to the appropriate submatrices.
\begin{definition}\label{def:representation}
For any left-to-right diagram $d\from \objr^{n}\to \objr^{m}$, a \emph{representation} is a matrix-diagram $\hat{d}\from \objr^{l+n}\to \objr^{l+m}$, such that
$\diagbox{d}{n}{m}\; = \; \traceform{\hat{d}}{n}{m}{l}$ and
$\hat{d}_{ll}$, $\hat{d}_{nl}$ are $\epsilon$-free. It is a \emph{deterministic representation} if moreover $\hat{d}_{ll}$ is deterministic.
\end{definition}
For example, given the string diagram below on the left, the one on the right is a representation for it, whose highlighted matrix-diagram is the same as above. 
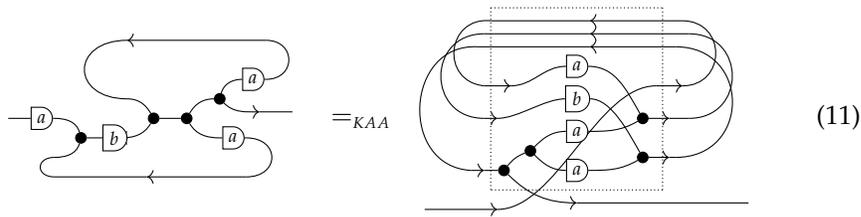
\begin{equation}

\tikzset{x=1em, y=2.1ex}
\InputIfFileExists{ex-automaton-diagram.tikz}{}{\input{./tikz/ex-automaton-diagram.tikz}}
\tikzset{x=1em, y=1.5ex}
\quad\eqKa
\tikzset{x=1em, y=2.1ex}
\InputIfFileExists{ex-rep.tikz}{}{\input{./tikz/ex-rep.tikz}}
\tikzset{x=1em, y=1.5ex}

\end{equation}
We will refer to the associated matrix-diagram $\hat{d}$ as the \emph{transition matrix} of a given representation.
From a $\objr\to \objr$ diagram with a representation $\hat{d}\from \objr^{l+1}\to \objr^{l+1}$, we can construct an NFA from its transition matrix $\hat{d}$ as follows: 
\begin{itemize}
\item its state set is $Q = \{q_1, \dots, q_l\}$, i.e., there is one state for each wire of $\hat{d}_{ll}$; 
\item its transition relation built from $\hat{d}_{ll}$ as described above;
\item its initial states $Q_0$ are those $q_i$ for which there exists an index $j$ such that the $ij$th coefficient of $\hat{d}_{1l}$ is non-zero (and therefore $\epsilon$);
\item its final states $F$ are those $q_j$ for which there exists an index $i$ such that the $ij$th coefficient of $\hat{d}_{l1}$ is non-zero (and therefore $\epsilon$);   
\end{itemize}
The construction above is the inverse of that of Section~\ref{sec:automata-diag}.
Moreover the connection between the constructed automaton and the original string diagram is summarised in the following statement, which is a straightforward corollary of Proposition~\ref{thm:nfa-to-diag}.
\begin{proposition}\label{thm:diagram-nfa}
For a diagram $d\from \objr\to \objr$ with a representation $\hat{d}$, let $A_{\hat{d}}$ be the associated automaton, constructed as above. Then $\hat{L}$ is the language recognised by $A_{\hat{d}}$ iff $\sem{d} = \left\{(K,K')\mid \hat{L}K \subseteq K'\right\}$. 
\end{proposition}
The following proposition states that we can extract a representation from any string diagram. Combined with Proposition~\ref{thm:diagram-nfa} it can also be read as a \emph{Kleene theorem} for our syntax.
\begin{proposition}[Kleene's Theorem for $\Aut$]\label{thm:traceform}
Any left-to-right diagram has a representation.
\end{proposition}

We established a correspondence between $\objr\to \objr$ diagrams and automata. What about arbitrary left-to-right diagrams $\objr^{n} \to \objr^{m}$? Their semantics is fully characterised by a single regular language \emph{for each pair of input-output port} (see Corollary~\ref{thm:1-to-1-restrict} below). As a result, the semantics of a given $\objr^{n} \to \objr^{m}$ diagram is fully characterised by an $m\times n$ array of languages.

\subsection{Interlude: from regular to context-free languages}\label{sec:context-free}

It is worth pointing out how a simple modification of the diagrammatic syntax takes us one notch up the Chomsky hierarchy, leaving the realm of regular languages for that of context-free grammars and languages. 

Our diagrammatic language allows to specify systems of language equations of the form $aX\subseteq Y$. In this context, feedback loops can be interpreted as fixed-points, in systems in which a variable may appear both on the left and on the right of equations. For example, the automaton below left, and its corresponding string diagram, below right, translate to the system of equations at the center:
\begin{equation}\label{eq:ex-system}

\tikzset{x=1em, y=2.1ex}
\InputIfFileExists{ex-automaton-graph.tikz}{}{\input{./tikz/ex-automaton-graph.tikz}}
\tikzset{x=1em, y=1.5ex}
 \quad \mapsto \; \begin{cases} \epsilon \subseteq X_0 \\ aX_0\subseteq X_1  \\ bX_1\subseteq X_2 \\ aX_2\subseteq X_1 \\ aX_2\subseteq X_2 \end{cases} \mapsfrom \quad 
\tikzset{x=1em, y=2.1ex}
\InputIfFileExists{ex-automaton-diagram.tikz}{}{\input{./tikz/ex-automaton-diagram.tikz}}
\tikzset{x=1em, y=1.5ex}
 
\end{equation} 
This translation can be obtained by simply labelling each state with a variable and adding one inequality of the form $aX_i\subseteq X_j$ for each $a$-transition from state $i$ to state $j$. The system we obtain corresponds very closely to the $\sem{-}$-semantics of the associated string diagram. 

The distinction between red and black wires can be understood as a type discipline that only allows linear uses of the product of languages. It is legitimate and enlightening to ask what would happen if we forgot about red wires and interpreted the action directly as the product. We would replace the action by a new generator $
\tikzset{x=1em, y=2.1ex}
\InputIfFileExists{lr-product.tikz}{}{\input{./tikz/lr-product.tikz}}
\tikzset{x=1em, y=1.5ex}
$ with  semantics $\sem{
\tikzset{x=1em, y=2.1ex}
\InputIfFileExists{lr-product.tikz}{}{\input{./tikz/lr-product.tikz}}
\tikzset{x=1em, y=1.5ex}
} = \{ \big((M,L),K\big) \mid ML\subseteq K \}$.

This would allow us to specify systems of language equations with unrestricted uses of the product on the left of inclusions, e.g. $UVW\subseteq X$.

\noindent\begin{minipage}{0.7\textwidth}
Equations of this form are similar to the production rules (e.g. $X\rightarrow UVW$) of context-free grammars and it is well-known that the least solutions of this class of systems are precisely \emph{context-free} languages~\cite[Chapter 10]{conway2012regular}. For example we could encode the Dyck language $X \rightarrow XX \mid (X) \mid \epsilon$ of properly matched parentheses as least solution of the system $\epsilon \subseteq X, (X) \subseteq X, XX \subseteq X$ which gives the diagram on the right.
\end{minipage}
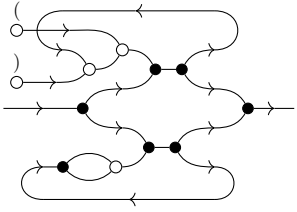
\begin{minipage}{0.3\textwidth}
\[

\tikzset{x=1em, y=2.1ex}
\InputIfFileExists{Dyck-diag.tikz}{}{\input{./tikz/Dyck-diag.tikz}}
\tikzset{x=1em, y=1.5ex}

\]
\end{minipage}

\section{Completeness and Determinisation}\label{sec:completeness}

This section is devoted to prove our completeness result, Theorem~\ref{thm:completeness}. We use a normal form argument: more specifically we mimic automata-theoretic results to rewrite every string diagram to a normal form corresponding to a minimal deterministic finite automaton (DFA). We achieve it by implementing Brzozowski's algorithm~\cite{brzozowski1962canonical} through diagrammatic equational reasoning. The proof proceeds in three distinct steps.
\begin{enumerate}
\item We first show (Section~\ref{sec:determinisation}) how to \emph{determinise} (the representation of) a diagram: this step consists in eliminating all subdiagrams that correspond to nondeterministic transitions in the associated automaton.
\item We use the previous step to implement a \emph{minimisation} procedure (Section~\ref{sec:minimisation}) from which we obtain a minimal representation for a given diagram: this is a representation whose associated automaton is minimal---with the fewest number of states---amongst DFAs that recognise the same language. To do this, we show how the four steps of Brzozowski's minimisation algorithm---reverse; determinise; reverse; determinise---translate into diagrammatic equational reasoning. Note that the first three steps taken together simply amount to applying in reverse the determinisation procedure we have already devised. That this is possible will be a consequence of the symmetry of~$\eqKa$.
\item Finally, from the uniqueness of minimal DFAs, any two diagrams that have the same denotation are both equal to the same minimal representation and we can derive completeness of~$\eqKa$. 
\end{enumerate}

We will now write equations in $\eqKa$ simply as $=$ to simplify notation and say that diagrams $c$ and $d$ are \emph{equal} when $c\eqKa d$. 

First, we use the symmetries of the equational theory to make simplifying assumptions about the diagrams we need to consider for the completeness proof.

\paragraph{A few simplifying assumptions.} Without loss of generality, the proof we give is restricted to string diagrams with no $\objred$ in their domain as well as in their codomain. This is simply a matter of convenience: the same proof would work for more general diagrams, that may contain $\objred$ in their (co)domain, at the cost of significantly cluttering diagrams. Henceforth, one can simply think of the labels for the action $\scalar{x}$ as uniquely identifying one open red wire in a diagram. With this convention, two or more occurrences of the same $x$ in a diagram can be seen as connected to the same red wire on the left, via $
\tikzset{x=1em, y=2.1ex}
\InputIfFileExists{ka-copy.tikz}{}{\input{./tikz/ka-copy.tikz}}
\tikzset{x=1em, y=1.5ex}
$. The completeness of $\eqKa$ restricted to the monochromatic red fragment  is a consequence of~\cite[Theorem 6.1]{bonchi2018deconstructing}.


Arbitrary objects in $\Aut$ are lists of the three generating objects. We have already motivated focusing on string diagrams with no open red wires so that the objects we care about are lists of $\objr$ and $\objl$. The following proposition implies that, without loss of generality, for the proof of completeness we can restrict further to left-to-right diagrams (Section~\ref{sec:automata-diag}).
\begin{proposition}\label{thm:left-to-right}
There is a natural bijection between sets of string diagrams of the form
\begin{equation*}

\tikzset{x=1em, y=2.1ex}
\InputIfFileExists{right-to-left.tikz}{}{\input{./tikz/right-to-left.tikz}}
\tikzset{x=1em, y=1.5ex}
\quad\leftrightarrow\quad 
\tikzset{x=1em, y=2.1ex}
\InputIfFileExists{wire-bend.tikz}{}{\input{./tikz/wire-bend.tikz}}
\tikzset{x=1em, y=1.5ex}
 \qquad\text{ where $A_i, B_i$ represent lists of $\objr$ and $\objl$.}
\end{equation*}
\end{proposition}
Proposition~\ref{thm:left-to-right} tell us that we can always bend the incoming wires to the left and outgoing wires to the right before applying some equations, and recover the original orientation of the wires by bending them into their original place later.

\subsection{Determinisation}\label{sec:determinisation}

In diagrammatic terms, a nondeterministic transition of the automaton associated to (a representation of) a given diagram, corresponds to a subdiagram of the form $
\tikzset{x=1em, y=2.1ex}
\InputIfFileExists{non-determinism.tikz}{}{\input{./tikz/non-determinism.tikz}}
\tikzset{x=1em, y=1.5ex}
$ for some $a\in \Sigma$.
Clearly, using the definition of $\scalar{a} := 
\tikzset{x=1em, y=2.1ex}
\InputIfFileExists{a-action.tikz}{}{\input{./tikz/a-action.tikz}}
\tikzset{x=1em, y=1.5ex}
$ in~\eqref{def:scalar} and the axiom
\begin{equation}

\tikzset{x=1em, y=2.1ex}
\InputIfFileExists{action-copy.tikz}{}{\input{./tikz/action-copy.tikz}}
\tikzset{x=1em, y=1.5ex}
\; \myeq{D1}\; 
\tikzset{x=1em, y=2.1ex}
\InputIfFileExists{action-copy-1.tikz}{}{\input{./tikz/action-copy-1.tikz}}
\tikzset{x=1em, y=1.5ex}
\text{, }\quad \text{ we have }\quad  
\tikzset{x=1em, y=2.1ex}
\InputIfFileExists{non-determinism.tikz}{}{\input{./tikz/non-determinism.tikz}}
\tikzset{x=1em, y=1.5ex}
\; =\; 
\tikzset{x=1em, y=2.1ex}
\InputIfFileExists{non-determinism-1.tikz}{}{\input{./tikz/non-determinism-1.tikz}}
\tikzset{x=1em, y=1.5ex}

\end{equation}
which will prove to be the engine of our determinisation procedure, along with the fact that any red expression can be copied and deleted. The next two theorems generalise the ability to copy and delete to arbitrary left-to-right diagrams. 
\begin{theorem}\label{thm:copy-merge}
For any left-to-right diagram $d\from \objr^m\to \objr^n$, we have
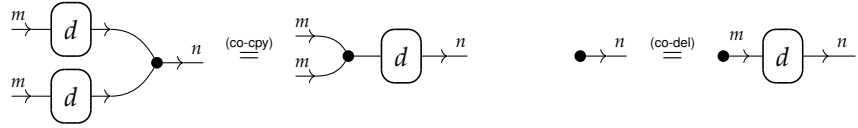
\begin{align*}

\tikzset{x=1em, y=2.1ex}
\InputIfFileExists{global-copy.tikz}{}{\input{./tikz/global-copy.tikz}}
\tikzset{x=1em, y=1.5ex}
\quad \myeq{cpy} \quad 
\tikzset{x=1em, y=2.1ex}
\InputIfFileExists{global-copy-1.tikz}{}{\input{./tikz/global-copy-1.tikz}}
\tikzset{x=1em, y=1.5ex}
 \qquad\qquad 
\tikzset{x=1em, y=2.1ex}
\InputIfFileExists{global-delete.tikz}{}{\input{./tikz/global-delete.tikz}}
\tikzset{x=1em, y=1.5ex}
\quad \myeq{del} \quad 
\tikzset{x=1em, y=2.1ex}
\InputIfFileExists{global-delete-1.tikz}{}{\input{./tikz/global-delete-1.tikz}}
\tikzset{x=1em, y=1.5ex}
\\

\tikzset{x=1em, y=2.1ex}
\InputIfFileExists{global-merge.tikz}{}{\input{./tikz/global-merge.tikz}}
\tikzset{x=1em, y=1.5ex}
\quad \myeq{co-cpy} \quad 
\tikzset{x=1em, y=2.1ex}
\InputIfFileExists{global-merge-1.tikz}{}{\input{./tikz/global-merge-1.tikz}}
\tikzset{x=1em, y=1.5ex}
\qquad\qquad  
\tikzset{x=1em, y=2.1ex}
\begin{tikzpicture}
	\begin{pgfonlayer}{nodelayer}
		\node [style=none] (16) at (-0.75, 0.5) {\scriptsize $n$};
		\node [style=none] (17) at (-1.25, 0) {};
		\node [style=black] (18) at (-2.25, 0) {};
		\node [style=none] (19) at (-0.5, 0) {};
	\end{pgfonlayer}
	\begin{pgfonlayer}{edgelayer}
		\draw (17.center) to (19.center);
		\draw [->] (18) to (17.center);
	\end{pgfonlayer}
\end{tikzpicture}
}
\tikzset{x=1em, y=1.5ex}
\quad \myeq{co-del} \quad 
\tikzset{x=1em, y=2.1ex}
\InputIfFileExists{global-co-delete-1.tikz}{}{\input{./tikz/global-co-delete-1.tikz}}
\tikzset{x=1em, y=1.5ex}

\end{align*}
\end{theorem}
For $d\from \objr^m\to \objr^n$, let $d_{ij}$ be the string diagram of type $\objr\to \objr$ obtained by composing every input with $
\tikzset{x=1em, y=2.1ex}
\begin{tikzpicture}
	\begin{pgfonlayer}{nodelayer}
		\node [style=black] (37) at (-0.5, 0) {};
		\node [style=none] (43) at (0.25, 0) {};
		\node [style=none] (44) at (0.75, 0) {};
	\end{pgfonlayer}
	\begin{pgfonlayer}{edgelayer}
		\draw [->] (37) to (43.center);
		\draw (44.center) to (43.center);
	\end{pgfonlayer}
\end{tikzpicture}
}
\tikzset{x=1em, y=1.5ex}
$ except the $i$th one, and every output with $
\tikzset{x=1em, y=2.1ex}
\begin{tikzpicture}
	\begin{pgfonlayer}{nodelayer}
		\node [style=black] (37) at (0.75, 0) {};
		\node [style=none] (43) at (0.25, 0) {};
		\node [style=none] (44) at (-0.5, 0) {};
	\end{pgfonlayer}
	\begin{pgfonlayer}{edgelayer}
		\draw (43.center) to (37);
		\draw [->] (44.center) to (43.center);
	\end{pgfonlayer}
\end{tikzpicture}
}
\tikzset{x=1em, y=1.5ex}
$ except the $j$th one. Theorem~\ref{thm:copy-merge} implies that string diagrams are fully characterised by their $\objr\to\objr$ subdiagrams.
\begin{corollary}\label{thm:1-to-1-restrict}
Given $d,e\from \objr^m\to \objr^n$, $d\eqKa e$ iff $d_{ij} \eqKa e_{ij}$, for all $1\leq i\leq m$ and $1\leq j\leq n$.
\end{corollary}
Thus, we can restrict our focus further to left-to-right $\objr\to\objr$ diagrams, without loss of generality. We are now able to devise a determinisation procedure for representation of diagrams, which we illustrate below on a simple example. 
\begin{proposition}[Determinisation]\label{thm:deterministic-rep}
Any diagram $\objr\to\objr$ has a deterministic representation.
\end{proposition}
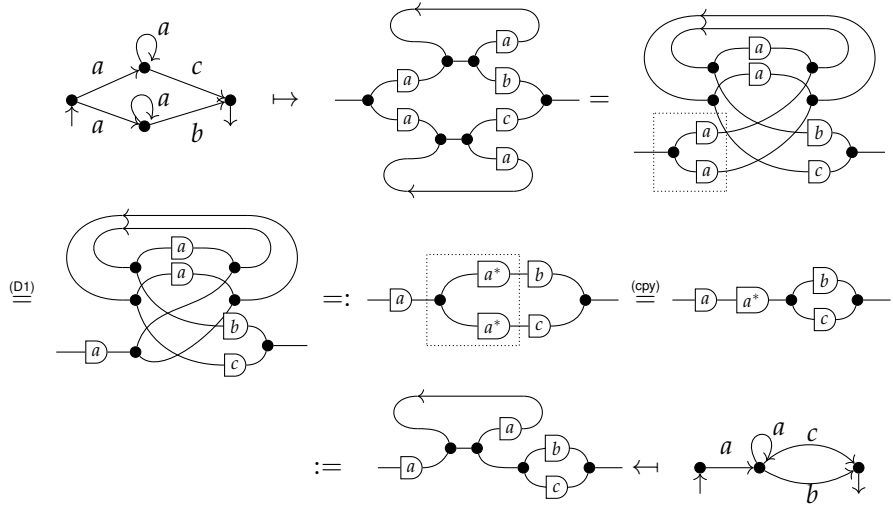
\begin{example}\label{ex:determinisation}
\begin{align*}

\tikzset{x=1em, y=2.1ex}
\InputIfFileExists{ex-graph-determinise.tikz}{}{\input{./tikz/ex-graph-determinise.tikz}}
\tikzset{x=1em, y=1.5ex}
\quad \mapsto\quad 
\tikzset{x=1em, y=2.1ex}
\InputIfFileExists{ex-diag-determinise.tikz}{}{\input{./tikz/ex-diag-determinise.tikz}}
\tikzset{x=1em, y=1.5ex}
 
 = \;
\tikzset{x=1em, y=2.1ex}
\InputIfFileExists{ex-diag-determinise-1.tikz}{}{\input{./tikz/ex-diag-determinise-1.tikz}}
\tikzset{x=1em, y=1.5ex}
\\
 \myeq{D1} \; 
\tikzset{x=1em, y=2.1ex}
\InputIfFileExists{ex-diag-determinise-2.tikz}{}{\input{./tikz/ex-diag-determinise-2.tikz}}
\tikzset{x=1em, y=1.5ex}

 =: \;
\tikzset{x=1em, y=2.1ex}
\InputIfFileExists{ex-diag-determinise-3.tikz}{}{\input{./tikz/ex-diag-determinise-3.tikz}}
\tikzset{x=1em, y=1.5ex}

 \;\myeq{cpy}\;
\tikzset{x=1em, y=2.1ex}
\InputIfFileExists{ex-diag-determinise-4.tikz}{}{\input{./tikz/ex-diag-determinise-4.tikz}}
\tikzset{x=1em, y=1.5ex}
\\
 := \quad
\tikzset{x=1em, y=2.1ex}
\InputIfFileExists{ex-diag-determinise-5.tikz}{}{\input{./tikz/ex-diag-determinise-5.tikz}}
\tikzset{x=1em, y=1.5ex}

\mapsfrom\quad 
\tikzset{x=1em, y=2.1ex}
\InputIfFileExists{ex-graph-deterministic.tikz}{}{\input{./tikz/ex-graph-deterministic.tikz}}
\tikzset{x=1em, y=1.5ex}
\quad 
\end{align*}
\end{example}

\paragraph{Dealing with useless states.}\label{sec:unreachable} Notice that our deterministic form is \emph{partial} and that the determinisation procedure disregards \emph{useless states}, i.e., parts of a string diagram that are not on a path from the input to the output wire. None of these contribute to the semantics of the diagram and can be safely eliminated using Theorem~\ref{thm:copy-merge}. If one prefers a \emph{total} deterministic form---one in which the transition relation not only contains each letter of $\Sigma$ at most once out of each state, but \emph{precisely} once. Conversely, one can use Theorem~\ref{thm:copy-merge} (del)-(co-del) to introduce an additional garbage state (corresponding to the empty set in the classical determinisation by subset construction), disconnected from the output, as a target for all undefined transitions. Rather than providing a tedious formal construction, let us illustrate this point on the preceding example: there is only one transition out of the initial state but we can add $b$ and $c$-transitions to a new state that does not lead anywhere, giving a total deterministic automaton, as follows.
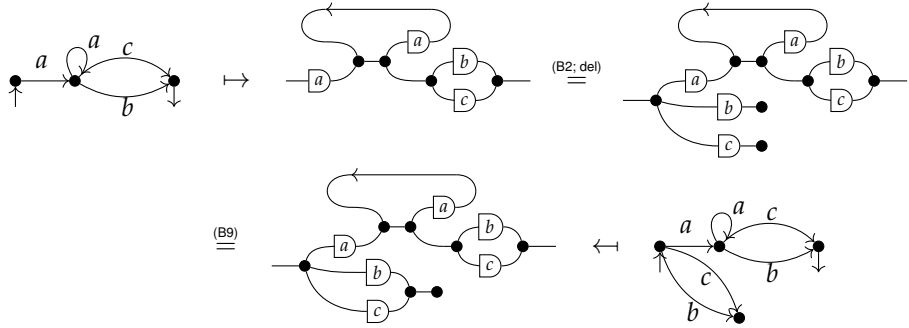
\begin{align*}

\tikzset{x=1em, y=2.1ex}
\InputIfFileExists{ex-graph-deterministic.tikz}{}{\input{./tikz/ex-graph-deterministic.tikz}}
\tikzset{x=1em, y=1.5ex}
 \quad &\;\mapsto\quad 
\tikzset{x=1em, y=2.1ex}
\InputIfFileExists{ex-diag-determinise-5.tikz}{}{\input{./tikz/ex-diag-determinise-5.tikz}}
\tikzset{x=1em, y=1.5ex}
 \quad 
\myeq{B2; del}\quad 
\tikzset{x=1em, y=2.1ex}
\InputIfFileExists{ex-diag-total-deterministic.tikz}{}{\input{./tikz/ex-diag-total-deterministic.tikz}}
\tikzset{x=1em, y=1.5ex}
\\
&\myeq{B9}\quad 
\tikzset{x=1em, y=2.1ex}
\InputIfFileExists{ex-diag-total-deterministic-1.tikz}{}{\input{./tikz/ex-diag-total-deterministic-1.tikz}}
\tikzset{x=1em, y=1.5ex}
\quad \mapsfrom \quad 
\tikzset{x=1em, y=2.1ex}
\InputIfFileExists{ex-graph-total-deterministic.tikz}{}{\input{./tikz/ex-graph-total-deterministic.tikz}}
\tikzset{x=1em, y=1.5ex}

\end{align*}

\subsection{Minimisation and completeness}
\label{sec:minimisation}

As explained above, our proof of completeness is a diagrammatic reformulation of Brzozowski's algorithm which proceeds in four steps: determinise, reverse, determinise, reverse. We already know how to determinise a given diagram. The other three steps are simply a matter of looking at string diagrams differently and showing that all the equations that we needed to determinise them, can be performed in reverse.

We say that a matrix-diagram is \emph{co-deterministic} if the converse of its associated transition relation is deterministic.

\begin{proof}[Proof of Theorem~\ref{thm:completeness} (Completeness)]
We have a procedure to show that, if $\sem{d}=\sem{e}$, then there exists a string diagram $f$ in normal form such that  $d=f=e$. This normal form is the diagrammatic counterpart of the \emph{minimal} automaton associated to $d$ and $e$. In our setting, it is the deterministic representation equal to $d$ and $e$ with the smallest number of states. This is unique because we can obtain from it the corresponding minimal automaton, which is well-known to be unique. First, given any string diagram we can obtain a representation for it by Proposition~\ref{thm:traceform}. Then we obtain a minimal representation by splitting Brzozowski's algorithm in two steps.
\begin{description}
\item[1. Reverse; determinise; reverse.] A close look at the determinisation procedure (proof of Proposition~\ref{thm:copy-merge} and Proposition~\ref{thm:deterministic-rep} in Appendix) shows that, at each step, the required laws all hold in reverse. For example, we can replace every instance of~(cpy) with~(co-cpy). 
We can thus define, in a completely analogous manner, a co-determinisation procedure which takes care of the first three steps of Brzozowski's algorithm, and obtain a co-deterministic representation for the given diagram.
\item[2. Determinise.]  Finally from a direct application of Proposition~\ref{thm:deterministic-rep}, we can obtain a deterministic representation for the co-deterministic representation of the previous step. The result is the desired minimal representation and normal form.
\end{description}
\end{proof}

\section{Discussion}\label{sec:conclusion}

In this paper, we have given a fully diagrammatic treatment of finite-state automata, with a finite equational theory that axiomatises them up to language equivalence. We have seen that this allows us to decompose the regular operations of Kleene algebra, like the star, into more primitive components, resulting in greater modularity. In this section, we compare our contributions with related work, and outline directions for future research.

Traditionally, computer scientists have used \emph{syntax or railroad diagrams} to visualise regular expressions and, more generally, context-free grammars~\cite{wirth1971programming}. These diagrams resemble very closely our syntax but have remained mostly informal and usually restricted to a single input and output. More recently, Hinze has treated the single input-output case rigorously as a pedagogical tool to teach the correspondence between finite-state automata and regular expressions~\cite{hinze2019self}. He did not, however, study their equational properties.


Bloom and {\'E}sik's \emph{iteration theories} provide a general categorical setting in which to study the equational properties of iteration for a broad range of structures that appear in the semantics of programming languages~\cite{bloom1993iteration}. They are cartesian categories equipped with a parameterised fixed-point operation which is closely related to the trace operation that we have used to represent the Kleene star. However, the monoidal category of interest in this paper is \emph{compact-closed} (only the full subcategory over $\objr$ and $\objl$ to be precise), a property that is incompatible with the existence of categorical products (any category that has both collapses to a preorder~\cite{lambek1988introduction}). Nevertheless, the subcategory of left-to-right diagrams (Section~\ref{sec:automata-diag}) is an iteration theory and, in particular a matrix iteration theory~\cite{bloom1993matrix}, a structure that Bloom and {\'E}sik have used to give an (infinitary) axiomatisation of regular languages~\cite{bloom1993equational}. 

Similarly, Stefanescu's work on \emph{network algebra} provides a unified algebraic treatment of various types of networks, including finite-state automata~\cite{stefanescu2000network}. In general, network algebras are traced monoidal categories where the product is not necessarily cartesian, and therefore more general than iteration theories. In both settings however, the trace is a global operation, that cannot be decomposed further into simpler components. In our work, on the other hand, the trace can be defined from the compact-closed structure, as was depicted in~\eqref{eq:star-decomposed}. 

Note that the compact closed subcategory in this paper can be recovered from the traced monoidal category of left-to-right diagrams, via the \emph{Int construction}~\cite{Joyal_tracedcategories}. Therefore, as far as mathematical expressiveness is concerned, the two approaches are equivalent. However, from a methodological point of view, taking the compact closed structure as primitive allows for improved compositionality, as example~\eqref{ex:decompose-automaton} in the introduction illustrates. Furthermore, the compact closed structure can be finitely presented relative to the theory of symmetric monoidal categories, whereas the trace operation cannot. This matters greatly in this paper, where finding a finite axiomatisation is our main concern.  

Finally, the idea of treating regular expressions as a free structure acting on a second algebraic structure also appeared in Pratt's treatment of \emph{dynamic algebra}, which axiomatises the propositional fragment of dynamic modal logic~\cite{pratt1988dynamic}. Like our formalism, and contrary to Kleene algebras, the variety of dynamic algebras is finitely-based. But they assume more structure: there the second algebraic structure is a Boolean algebra.

In all the formalisms we have mentioned, the difficulty typically lies in capturing the behaviour of iteration---whether as the star in Kleene algebra~\cite{kozen1994completeness,bloom1993equational}, or a trace operator~\cite{bloom1993iteration} in iteration theory and network algebra~\cite{stefanescu2000network}. The axioms should be coercive enough to force it to be \emph{the least fixed-point} of the language map $L\mapsto \{\epsilon\} \cup LK$. In Kozen's axiomatisation of Kleene algebra~\cite{kozen1994completeness} for example, this is through (a) the axiom $1+ ee^*\leq e^*$ (star is a fixpoint) and (b) the Horn clause $f+ex \leq x \Rightarrow  e^*f \leq x$ (star is the least fixpoint). In our work, (a) is a consequence of the unfolding of the star into a feedback loop and can be derived from the other axioms. (b) is more subtle, but can be seen as a consequence of (D1)-(D4) axioms. These allows us to (co)copy and (co)delete arbitrary diagrams (Theorem~\ref{thm:copy-merge}) and we conjecture that this is what forces the star to be a single definite value, not just any fixed-point, but the least one. Making this statement precise is the subject of future work.  

The difficulty in capturing the behaviour of fixed-points is also the reason why we decided to work with an additional red wire, to encode the action of regular expressions on the set of languages---without it, global (co)copying and (co)deleting (Theorem~\ref{thm:copy-merge}) cannot be reduced to the local (D1)-(D4) axioms. There is another route, that leads to an infinitary axiomatisation: we could dispense with the red generators altogether and take $\scalar{a}$ (for $a\in\Sigma$) as primitive instead, with global axioms to (co)copy and (co)delete arbitrary diagrams. This would pave the way for a reformulation of our work in the context of iteration (matrix) theories, where the ability to (co)copy and (co)delete arbitrary expressions is already built-in. We leave this for future work.

There is an intriguing parallelism between our case study and the positive fragment of relation algebra (also known as allegories~\cite{freyd1990categories}). Indeed, allegories, like Kleene algebra, do not admit a finite axiomatisation~\cite{freyd1990categories}. However, this result holds for the standard algebraic theories
. It has been shown recently that a structure equivalent to allegories can be given a finite axiomatisation when formulated in terms of string diagrams in monoidal categories~\cite{bonchi2018graphical}. It seems like the greater generality of the monoidal setting---algebraic theories correspond precisely to the particular case of cartesian monoidal categories~\cite{bonchi2018deconstructing}---allows for simpler axiomatisations in some specific cases. In the future we would like to understand whether this phenomenon, of which now we have two instances, can be understood in a general context. In the future we would like to understand whether this phenomenon, of which now we have two instances, can be understood in a general context.

Lastly, extensions of Kleene Algebra, such as Concurrent Kleene Algebra (CKA)~\cite{hoare2009concurrent,KappeB0Z18} and NetKAT~\cite{anderson2014netkat}, are increasingly relevant in current research. Enhancing our theory $\eqKa$ to encompass these extensions seems a promising research direction, for two main reasons. First, the two-dimensional nature of string diagrams  has been proven particularly suitable to reason about  concurrency (see e.g.~\cite{BonchiHPSZ19,Bruni2013}), and more generally about resource exchange between processes  (see e.g.~\cite{BonchiSZ17,pqp,jacobs2019causal,baez2015compositional,BPSZ-lics19}). Second, when trying to transfer the good meta-theoretical properties of Kleene Algebra (like completeness and decidability) to extensions such as CKA and NetKAT, the cleanest way to proceed is usually in a modular fashion. The interaction between the new operators of the extension and the Kleene star usually represents the greatest challenge to this methodology. Now, in $\eqKa$, the Kleene star is decomposable into simpler components (see~\eqref{eq:star-decomposed}) and there is only one specific axiom (C5) governing its behaviour. We believe this is a particularly favourable starting point to modularise a meta-theoretic study of CKA and NetKAT with string diagrams, taking advantage of the results we presented in this paper for finite-state automata.

\bibliographystyle{splncs04}
\bibliography{refs}

\newpage
\appendix
\section{Proofs}\label{app:completeness}

\subsection{Encoding regular expressions and automata}

We write ``$\poi$'' for relational composition, from left to right: $R\poi S = \{(x,z)\mid \exists y, (x,y)\in R, (y,z)\in S\}$.

\begin{proof}[Proof of Proposition~\ref{thm:diagram-regexp}]
By induction on the structure of regular expressions. The proposition holds by definition on the generators: $\sem{\transreg{a}} = \{(L,K)\mid aL\subseteq K\}$. There are three inductive cases to consider. Assume that $e$ and $f$ satisfy the proposition. 
\begin{itemize}
\item For the $ef$ case, $\sem{\transreg{ef}} = \sem{\transreg{e}}\poi \sem{\transreg{f}} = \{(L,K)\mid \semreg{e}L\subseteq K\}\poi \{(L,K)\mid \semreg{f}L\subseteq K\}$. Hence, by monotony of the product, we have $\sem{\transreg{ef}} =\{(L,K)\mid \semreg{e}\semreg{f}L\subseteq K\} = \{(L,K)\mid \semreg{ef}L\subseteq K\}$.
\item For the case of $e+f$ we have \begin{align*}
\sem{\transreg{e+f}} &= \{(L,K)\mid \exists K_1, K_2\subseteq K, L_1, L_2, \text{ s.t. } L\subseteq L_1, L_2, \semreg{e}L_1\subseteq K_1, \semreg{f}L_2\subseteq K_2\}\\
&= \{(L,K)\mid \exists L_1, L_2, \text{ s.t. } L\subseteq L_1, L_2, \semreg{e}L_1\subseteq K, \semreg{f}L_2\subseteq K\}\\
& = \{(L,K)\mid \exists L_1, L_2, \text{ s.t. } L\subseteq L_1, L_2, \semreg{e}L_1 \cup \semreg{f}L_2 \subseteq K\}\\
&= \{(L,K)\mid \semreg{e}L \cup \semreg{f}L \subseteq K\} \\
&= \{(L,K)\mid (\semreg{e}\cup \semreg{f})L \subseteq K\} =  \{(L,K)\mid \semreg{e+f}L \subseteq K\}\text{.}
\end{align*} 
\item Finally, for $e^*$, 
\begin{align*}
\sem{\transreg{e^*}} &=  \{(L,K)\mid \exists M,N\text{ s.t. } \, M,L\subseteq N, \semreg{e}N \subseteq M, N\subseteq K\}\\
& = \{(L,K)\mid \exists N \text{ s.t. }  \semreg{e}N \subseteq N, L\subseteq N\subseteq K\}\\
& = \{(L,K)\mid \exists N \text{ s.t. }  L\cup \semreg{e}N \subseteq N, L\subseteq N\subseteq K\}\\
& = \{(L,K)\mid \exists N \text{ s.t. }   \semreg{e}^* L \subseteq N, L\subseteq N\subseteq K\}\\
& = \{(L,K)\mid \exists N \text{ s.t. }   \semreg{e^*}L \subseteq N, L\subseteq N\subseteq K\}\\
& = \{(L,K)\mid \semreg{e^*}L \subseteq K\}
\end{align*}
where the fourth equation is a consequence of Arden's lemma~\cite{arden1961delayed}: $A^*B$ is the smallest solution (for $X$) of the language equation $B\cup AX\subseteq X$. 
\end{itemize}
\end{proof}

\begin{proof}[Proof of Proposition~\ref{thm:nfa-to-diag}]
This is the diagrammatic counterpart of the representation of automata as matrices of regular expressions given in \cite[Definition 12]{kozen1994completeness}.

We write $\mathbf{K}$ for a vector of languages $(K_1,\dots, K_Q)$ and, for a square matrix of languages $\mathbf{A}$, let $\mathbf{A}\mathbf{K}$ be the language vector resulting from applying $\mathbf{A}$ to $\mathbf{K}$ in the obvious way. By~\cite[Theorem 11]{kozen1994completeness}, square language matrices form a Kleene algebra, with the composition of matrices as product, component-wise union as sum and the star defined as in~\cite[Lemma 10]{kozen1994completeness}. We also write write $\mathbf{K}\subseteq \mathbf{K}'$ if the inclusions all hold component-wise. Furthermore, Arden's lemma holds in this slightly more general setting: the least solution of the language-matrix equation $\mathbf{B}\cup \mathbf{A}\mathbf{X}\subseteq \mathbf{X}$ is $\mathbf{X} = \mathbf{A}^*\mathbf{B}$. 

Now, for a given automaton $A$ we construct the diagram below as explained in the main text:
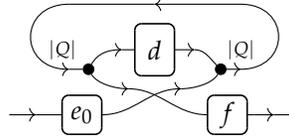
\[
\tikzset{x=1em, y=2.1ex}
\InputIfFileExists{nfa-rep.tikz}{}{\input{./tikz/nfa-rep.tikz}}
\tikzset{x=1em, y=1.5ex}
\]
with $d$ the diagram encoding the transition relation of $A$, $e_0$ the diagram encoding its initial state, and $f$ the diagram encoding its set of final states. Let $\mathbf{D}$ be the language matrix obtained from $A$ by letting $\mathbf{D}_{ij} = \{a\}$ if $(q_i, a, q_j)$ is in the transition relation of $A$.
First, we have

\noindent \[\sem{
\tikzset{x=1em, y=2.1ex}
\InputIfFileExists{nfa-rep-star.tikz}{}{\input{./tikz/nfa-rep-star.tikz}}
\tikzset{x=1em, y=1.5ex}
}
\begin{array}{l}
=  \{(\mathbf{K},\mathbf{K}')\mid \exists \mathbf{M},\mathbf{N}, \; \mathbf{M},\mathbf{K}\subseteq \mathbf{N}, \mathbf{D}\mathbf{N}\subseteq \mathbf{M}, \mathbf{N}\subseteq \mathbf{K}'\}\\
=  \{(\mathbf{K},\mathbf{K}')\mid \exists \mathbf{N}, \; \mathbf{D}\mathbf{N}\subseteq \mathbf{N}, \mathbf{K}\subseteq \mathbf{N}\subseteq \mathbf{K}'\}\\
= \{(\mathbf{K},\mathbf{K}')\mid \exists \mathbf{N}, \; \mathbf{K}\cup \mathbf{D}\mathbf{N}\subseteq \mathbf{N}, \mathbf{K}\subseteq \mathbf{N}\subseteq \mathbf{K}'\}\\
= \{(\mathbf{K},\mathbf{K}')\mid \exists \mathbf{N}, \;  \mathbf{D}^*\mathbf{K}\subseteq \mathbf{N}, \mathbf{K}\subseteq \mathbf{N}\subseteq \mathbf{K}'\}\\
= \{(\mathbf{K},\mathbf{K}')\mid  \mathbf{D}^*\mathbf{K}\subseteq \mathbf{K}'\}
\end{array}
\]
where the penultimate step holds by the matrix Arden's lemma. Then, $\sem{e_0}$ and $\sem{f}$ pick out the component languages of $\mathbf{D}^*$ that correspond to the initial states of $A$ and some final state, and takes their union. Thus, we get 
\[\sem{
\tikzset{x=1em, y=2.1ex}
\InputIfFileExists{nfa-rep.tikz}{}{\input{./tikz/nfa-rep.tikz}}
\tikzset{x=1em, y=1.5ex}
} \begin{array}{l}
= \sem{e_0}\poi \{(\mathbf{K},\mathbf{K}')\mid  \mathbf{D}^*\mathbf{K}\subseteq \mathbf{K}'\}\poi \sem{f} \\
= \left\{(K,K')\mid LK \subseteq K'\right\}
\end{array}  \]
where $L$ is the language accepted by the original automaton. 
\end{proof}


\begin{lemma}\label{thm:relations}
Every left-to-right diagram that does not contain the $
\tikzset{x=1em, y=2.1ex}
\InputIfFileExists{action.tikz}{}{\input{./tikz/action.tikz}}
\tikzset{x=1em, y=1.5ex}
$ generator is equal to a one that factors as a block of $
\tikzset{x=1em, y=2.1ex}
\InputIfFileExists{lr-copy.tikz}{}{\input{./tikz/lr-copy.tikz}}
\tikzset{x=1em, y=1.5ex}
, 
\tikzset{x=1em, y=2.1ex}
\begin{tikzpicture}
	\begin{pgfonlayer}{nodelayer}
		\node [style=black] (37) at (0.75, 0) {};
		\node [style=none] (43) at (0.25, 0) {};
		\node [style=none] (44) at (-0.5, 0) {};
	\end{pgfonlayer}
	\begin{pgfonlayer}{edgelayer}
		\draw (43.center) to (37);
		\draw [->] (44.center) to (43.center);
	\end{pgfonlayer}
\end{tikzpicture}
}
\tikzset{x=1em, y=1.5ex}
$, followed by a block of $
\tikzset{x=1em, y=2.1ex}
\InputIfFileExists{lr-merge.tikz}{}{\input{./tikz/lr-merge.tikz}}
\tikzset{x=1em, y=1.5ex}
, 
\tikzset{x=1em, y=2.1ex}
\begin{tikzpicture}
	\begin{pgfonlayer}{nodelayer}
		\node [style=black] (37) at (-0.5, 0) {};
		\node [style=none] (43) at (0.25, 0) {};
		\node [style=none] (44) at (0.75, 0) {};
	\end{pgfonlayer}
	\begin{pgfonlayer}{edgelayer}
		\draw [->] (37) to (43.center);
		\draw (44.center) to (43.center);
	\end{pgfonlayer}
\end{tikzpicture}
}
\tikzset{x=1em, y=1.5ex}
$. 
\end{lemma}
\begin{proof}
The equational theory $\eqKa$ restricted to the four generators $
\tikzset{x=1em, y=2.1ex}
\InputIfFileExists{lr-copy.tikz}{}{\input{./tikz/lr-copy.tikz}}
\tikzset{x=1em, y=1.5ex}
, 
\tikzset{x=1em, y=2.1ex}
\begin{tikzpicture}
	\begin{pgfonlayer}{nodelayer}
		\node [style=black] (37) at (0.75, 0) {};
		\node [style=none] (43) at (0.25, 0) {};
		\node [style=none] (44) at (-0.5, 0) {};
	\end{pgfonlayer}
	\begin{pgfonlayer}{edgelayer}
		\draw (43.center) to (37);
		\draw [->] (44.center) to (43.center);
	\end{pgfonlayer}
\end{tikzpicture}
}
\tikzset{x=1em, y=1.5ex}
$, $
\tikzset{x=1em, y=2.1ex}
\InputIfFileExists{lr-merge.tikz}{}{\input{./tikz/lr-merge.tikz}}
\tikzset{x=1em, y=1.5ex}
, 
\tikzset{x=1em, y=2.1ex}
\begin{tikzpicture}
	\begin{pgfonlayer}{nodelayer}
		\node [style=black] (37) at (-0.5, 0) {};
		\node [style=none] (43) at (0.25, 0) {};
		\node [style=none] (44) at (0.75, 0) {};
	\end{pgfonlayer}
	\begin{pgfonlayer}{edgelayer}
		\draw [->] (37) to (43.center);
		\draw (44.center) to (43.center);
	\end{pgfonlayer}
\end{tikzpicture}
}
\tikzset{x=1em, y=1.5ex}
$ coincides with the equational theory of relations between finite sets. The proof of this fact and a normal form that factorises as in the statement of the Lemma can be found in~\cite[Section 4]{Lafont95-equationalReasoningTwoDimDiagrams}.
\end{proof}

\begin{proof}[Proof of Proposition~\ref{thm:traceform}]
First, we claim that we can always find $c$ containing no action $\scalar{a}$ such that  
\begin{equation}\label{eq:tracecanform}
\diagbox{d}{n}{m} \quad = \quad \traceaction{c}{n}{m}{l}{x}
\end{equation}
where $\scalar{x}\from \objr^l\to\objr^l$ is simply a vertical composition of $l$ different $\scalar{a}$, $a\in\Sigma$. 

To prove this claim, we reason by structural induction on $\Aut$. For the base case, if $d$ is $\scalar{a}$, we have
\begin{equation}
\scalar{a} \quad \myeq{A1} \quad 
\tikzset{x=1em, y=2.1ex}
\InputIfFileExists{trace-proof-base-1.tikz}{}{\input{./tikz/trace-proof-base-1.tikz}}
\tikzset{x=1em, y=1.5ex}
\quad=\quad 
\tikzset{x=1em, y=2.1ex}
\InputIfFileExists{trace-proof-base.tikz}{}{\input{./tikz/trace-proof-base.tikz}}
\tikzset{x=1em, y=1.5ex}

\end{equation}
and every morphism that does not contain $\scalar{x}$ is trivially in the right form, with the trace taken over the $0$ object.
 
There are two inductive cases to consider:
\begin{itemize}
\item $d$ is given by the sequential composition of two morphisms of the appropriate form (using the induction hypothesis). Then 
\begin{align}

\tikzset{x=1em, y=2.1ex}
\InputIfFileExists{trace-proof-comp.tikz}{}{\input{./tikz/trace-proof-comp.tikz}}
\tikzset{x=1em, y=1.5ex}
\quad & =\quad
\tikzset{x=1em, y=2.1ex}
\InputIfFileExists{trace-proof-comp-1.tikz}{}{\input{./tikz/trace-proof-comp-1.tikz}}
\tikzset{x=1em, y=1.5ex}
 \\
& = \quad 
\tikzset{x=1em, y=2.1ex}
\InputIfFileExists{trace-proof-comp-2.tikz}{}{\input{./tikz/trace-proof-comp-2.tikz}}
\tikzset{x=1em, y=1.5ex}

\end{align}
\item $d$ is given as the monoidal product of two morphisms of the appropriate form. Then
\begin{align}

\tikzset{x=1em, y=2.1ex}
\InputIfFileExists{trace-proof-tensor.tikz}{}{\input{./tikz/trace-proof-tensor.tikz}}
\tikzset{x=1em, y=1.5ex}
\quad & =\quad
\tikzset{x=1em, y=2.1ex}
\InputIfFileExists{trace-proof-tensor-1.tikz}{}{\input{./tikz/trace-proof-tensor-1.tikz}}
\tikzset{x=1em, y=1.5ex}
\\
& = \quad 
\tikzset{x=1em, y=2.1ex}
\InputIfFileExists{trace-proof-tensor-2.tikz}{}{\input{./tikz/trace-proof-tensor-2.tikz}}
\tikzset{x=1em, y=1.5ex}

\end{align}
\end{itemize}
In \begin{equation}\label{eq:tracecanform}
\diagbox{d}{n}{m} \quad = \quad \traceaction{c}{n}{m}{l}{x}
\end{equation}
since $c$ is contains no action nodes, by Lemma~\ref{thm:relations}, it is equal to a matrix diagram with only entries $0$ or $\epsilon$ (i.e. a relation).
In other words, we can assume that $c$ factorises as a first layer of comonoid $\Bcomult, \Bcounit$, followed by a layer of permutations and a third layer of vertical compositions of the monoid $\Bmult, \Bunit$. 
Now, the action nodes in the trace distribute over $\Bmult$ by (D3) so that we can push them inside $c$. The resulting matrix diagram $d'$ is such that $d'_{ll}$ is $\epsilon$-free, as needed.
\end{proof}

\subsection{Completeness}

\begin{proof}[Proof of Proposition~\ref{thm:left-to-right}]
This proposition holds in any compact-closed category and relies on the ability to bend wires using $
\tikzset{x=1em, y=2.1ex}
\InputIfFileExists{cap-down.tikz}{}{\input{./tikz/cap-down.tikz}}
\tikzset{x=1em, y=1.5ex}
$ and $
\tikzset{x=1em, y=2.1ex}
\InputIfFileExists{cup-down.tikz}{}{\input{./tikz/cup-down.tikz}}
\tikzset{x=1em, y=1.5ex}
$. Explicitly, given a diagram of the first form, we can obtain one of the second as follows:
\begin{equation}

\tikzset{x=1em, y=2.1ex}
\InputIfFileExists{right-to-left.tikz}{}{\input{./tikz/right-to-left.tikz}}
\tikzset{x=1em, y=1.5ex}
\quad \mapsto\quad 
\tikzset{x=1em, y=2.1ex}
\InputIfFileExists{bent-wires.tikz}{}{\input{./tikz/bent-wires.tikz}}
\tikzset{x=1em, y=1.5ex}

\end{equation}
The inverse mapping is symmetric. That they are inverse transformations follows immediately from (A1)-(A2).
\end{proof}

\begin{lemma}\label{thm:copy-regexp}
For any $\diagregexp{e}$, we have
\begin{equation*}

\tikzset{x=1em, y=2.1ex}
\InputIfFileExists{regexp-copy.tikz}{}{\input{./tikz/regexp-copy.tikz}}
\tikzset{x=1em, y=1.5ex}
\; =\; 
\tikzset{x=1em, y=2.1ex}
\InputIfFileExists{regexp-copy-1.tikz}{}{\input{./tikz/regexp-copy-1.tikz}}
\tikzset{x=1em, y=1.5ex}
\qquad\qquad 
\tikzset{x=1em, y=2.1ex}
\begin{tikzpicture}
	\begin{pgfonlayer}{nodelayer}
		\node [style=rcoreg] (4) at (-1.25, 0) {${\color{red} e}$};
		\node [style=rn] (5) at (0.25, 0) {};
	\end{pgfonlayer}
	\begin{pgfonlayer}{edgelayer}
		\draw [red] (4) to (5);
	\end{pgfonlayer}
\end{tikzpicture}
}
\tikzset{x=1em, y=1.5ex}
\; =\;
\end{equation*}
\end{lemma}
\begin{proof}
By structural induction. It also follows from the more general case of~\cite[Theorem 2.42]{ZanasiThesis}.
\end{proof}

The class of left-to-right diagrams can be characterised inductively. We call \emph{trace} the canonical feedback operation induced by $
\tikzset{x=1em, y=2.1ex}
\InputIfFileExists{cup-down.tikz}{}{\input{./tikz/cup-down.tikz}}
\tikzset{x=1em, y=1.5ex}
$ and $
\tikzset{x=1em, y=2.1ex}
\InputIfFileExists{cap-down.tikz}{}{\input{./tikz/cap-down.tikz}}
\tikzset{x=1em, y=1.5ex}
$. Given a left-to-right diagram $d\from 1+n\to 1+m$, its trace is defined as
\begin{equation}

\tikzset{x=1em, y=2.1ex}
\InputIfFileExists{directed-trace.tikz}{}{\input{./tikz/directed-trace.tikz}}
\tikzset{x=1em, y=1.5ex}

\end{equation}
Let $\mathsf{Rat}$ be the set of diagrams of $\Aut$ containing $
\tikzset{x=1em, y=2.1ex}
\InputIfFileExists{lr-copy.tikz}{}{\input{./tikz/lr-copy.tikz}}
\tikzset{x=1em, y=1.5ex}
, 
\tikzset{x=1em, y=2.1ex}
\begin{tikzpicture}
	\begin{pgfonlayer}{nodelayer}
		\node [style=black] (37) at (0.75, 0) {};
		\node [style=none] (43) at (0.25, 0) {};
		\node [style=none] (44) at (-0.5, 0) {};
	\end{pgfonlayer}
	\begin{pgfonlayer}{edgelayer}
		\draw (43.center) to (37);
		\draw [->] (44.center) to (43.center);
	\end{pgfonlayer}
\end{tikzpicture}
}
\tikzset{x=1em, y=1.5ex}
$, $
\tikzset{x=1em, y=2.1ex}
\InputIfFileExists{lr-merge.tikz}{}{\input{./tikz/lr-merge.tikz}}
\tikzset{x=1em, y=1.5ex}
, 
\tikzset{x=1em, y=2.1ex}
\begin{tikzpicture}
	\begin{pgfonlayer}{nodelayer}
		\node [style=black] (37) at (-0.5, 0) {};
		\node [style=none] (43) at (0.25, 0) {};
		\node [style=none] (44) at (0.75, 0) {};
	\end{pgfonlayer}
	\begin{pgfonlayer}{edgelayer}
		\draw [->] (37) to (43.center);
		\draw (44.center) to (43.center);
	\end{pgfonlayer}
\end{tikzpicture}
}
\tikzset{x=1em, y=1.5ex}
$, $
\tikzset{x=1em, y=2.1ex}
\InputIfFileExists{action.tikz}{}{\input{./tikz/action.tikz}}
\tikzset{x=1em, y=1.5ex}
$, and all red generators, which is closed under the operations of vertical composition, horizontal composition, and the trace. Clearly any diagram of $\mathsf{Rat}$ is a left-to-right diagram. The converse is also true, up to $\eqKa$, as a corollary of Lemma~\ref{thm:traceform}, which proves a stronger result about the form of left-to-right diagrams.
\begin{proposition}\label{thm:inductive-lr}
Every left-to-right diagram is equal to one in $\mathsf{Rat}$.  
\end{proposition}


We can define the diagrammatic counterpart of matrices whose coefficients are regular expressions. 
\begin{definition}\label{def:generalised-matrix}
A \emph{generalised matrix-diagram} is a left-to-right diagram that factors as a block of $
\tikzset{x=1em, y=2.1ex}
\InputIfFileExists{lr-copy.tikz}{}{\input{./tikz/lr-copy.tikz}}
\tikzset{x=1em, y=1.5ex}
, 
\tikzset{x=1em, y=2.1ex}
\begin{tikzpicture}
	\begin{pgfonlayer}{nodelayer}
		\node [style=black] (37) at (0.75, 0) {};
		\node [style=none] (43) at (0.25, 0) {};
		\node [style=none] (44) at (-0.5, 0) {};
	\end{pgfonlayer}
	\begin{pgfonlayer}{edgelayer}
		\draw (43.center) to (37);
		\draw [->] (44.center) to (43.center);
	\end{pgfonlayer}
\end{tikzpicture}
}
\tikzset{x=1em, y=1.5ex}
$, followed by a block of $\scalar{e} (e\in \RegExp)$ and finally, a block of $
\tikzset{x=1em, y=2.1ex}
\InputIfFileExists{lr-merge.tikz}{}{\input{./tikz/lr-merge.tikz}}
\tikzset{x=1em, y=1.5ex}
, 
\tikzset{x=1em, y=2.1ex}
\begin{tikzpicture}
	\begin{pgfonlayer}{nodelayer}
		\node [style=black] (37) at (-0.5, 0) {};
		\node [style=none] (43) at (0.25, 0) {};
		\node [style=none] (44) at (0.75, 0) {};
	\end{pgfonlayer}
	\begin{pgfonlayer}{edgelayer}
		\draw [->] (37) to (43.center);
		\draw (44.center) to (43.center);
	\end{pgfonlayer}
\end{tikzpicture}
}
\tikzset{x=1em, y=1.5ex}
$. 
\end{definition}
Just like for matrices, we call the regular expressions $d_{ij}$ that appear in a generalised matrix-diagram $d\from \objr^n\to \objr^m$ the \emph{coefficients} of $d$ and index them by pairs of numbers $\{1,\dots n\}$, $\{1,\dots m\}$ in the usual way. The following proposition shows that left-to-right diagrams are as expressive as matrices of regular expressions. 
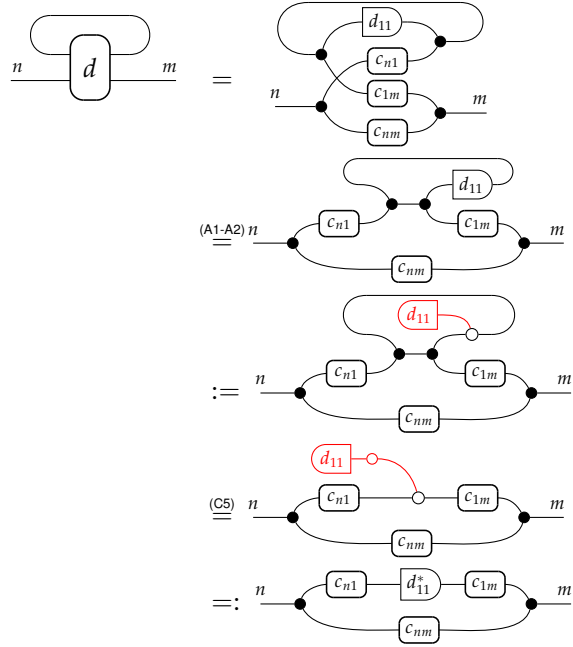
\begin{proposition}\label{thm:matrix-regexp}
Any left-to-right diagram is equal to a generalised matrix diagram.
\end{proposition}
\begin{proof}
We reason by structural induction, using the inductive characterisation of left-to-right diagrams of Proposition~\ref{thm:inductive-lr}. The base cases are those of $
\tikzset{x=1em, y=2.1ex}
\InputIfFileExists{lr-copy.tikz}{}{\input{./tikz/lr-copy.tikz}}
\tikzset{x=1em, y=1.5ex}
, 
\tikzset{x=1em, y=2.1ex}
\begin{tikzpicture}
	\begin{pgfonlayer}{nodelayer}
		\node [style=black] (37) at (0.75, 0) {};
		\node [style=none] (43) at (0.25, 0) {};
		\node [style=none] (44) at (-0.5, 0) {};
	\end{pgfonlayer}
	\begin{pgfonlayer}{edgelayer}
		\draw (43.center) to (37);
		\draw [->] (44.center) to (43.center);
	\end{pgfonlayer}
\end{tikzpicture}
}
\tikzset{x=1em, y=1.5ex}
$, $\scalar{e} (e\in \RegExp)$, $
\tikzset{x=1em, y=2.1ex}
\InputIfFileExists{lr-merge.tikz}{}{\input{./tikz/lr-merge.tikz}}
\tikzset{x=1em, y=1.5ex}
$, and $
\tikzset{x=1em, y=2.1ex}
\begin{tikzpicture}
	\begin{pgfonlayer}{nodelayer}
		\node [style=black] (37) at (-0.5, 0) {};
		\node [style=none] (43) at (0.25, 0) {};
		\node [style=none] (44) at (0.75, 0) {};
	\end{pgfonlayer}
	\begin{pgfonlayer}{edgelayer}
		\draw [->] (37) to (43.center);
		\draw (44.center) to (43.center);
	\end{pgfonlayer}
\end{tikzpicture}
}
\tikzset{x=1em, y=1.5ex}
$ which are all generalised matrix diagrams, by definition.

By Proposition~\ref{thm:inductive-lr}, there are three inductive cases to consider.
\begin{enumerate}
\item Let $d\from \objr^n\to \objr^m$ and $c\from \objr^m\to \objr^l$ be two generalised matrix-diagrams. Consider their horizontal composite. By applying the bimonoid laws \hyperref[ax:rel]{(B6)-(B8)} as many times as needed we can commute the block of $
\tikzset{x=1em, y=2.1ex}
\InputIfFileExists{lr-merge.tikz}{}{\input{./tikz/lr-merge.tikz}}
\tikzset{x=1em, y=1.5ex}
, 
\tikzset{x=1em, y=2.1ex}
\begin{tikzpicture}
	\begin{pgfonlayer}{nodelayer}
		\node [style=black] (37) at (-0.5, 0) {};
		\node [style=none] (43) at (0.25, 0) {};
		\node [style=none] (44) at (0.75, 0) {};
	\end{pgfonlayer}
	\begin{pgfonlayer}{edgelayer}
		\draw [->] (37) to (43.center);
		\draw (44.center) to (43.center);
	\end{pgfonlayer}
\end{tikzpicture}
}
\tikzset{x=1em, y=1.5ex}
$ of $d$ past the block of $
\tikzset{x=1em, y=2.1ex}
\InputIfFileExists{lr-copy.tikz}{}{\input{./tikz/lr-copy.tikz}}
\tikzset{x=1em, y=1.5ex}
, 
\tikzset{x=1em, y=2.1ex}
\begin{tikzpicture}
	\begin{pgfonlayer}{nodelayer}
		\node [style=black] (37) at (0.75, 0) {};
		\node [style=none] (43) at (0.25, 0) {};
		\node [style=none] (44) at (-0.5, 0) {};
	\end{pgfonlayer}
	\begin{pgfonlayer}{edgelayer}
		\draw (43.center) to (37);
		\draw [->] (44.center) to (43.center);
	\end{pgfonlayer}
\end{tikzpicture}
}
\tikzset{x=1em, y=1.5ex}
$ of $c$. Then, in the same way, we can apply \hyperref[ax:action-homomorphism]{(D3)-(D4)} to commute the block of $
\tikzset{x=1em, y=2.1ex}
\InputIfFileExists{lr-merge.tikz}{}{\input{./tikz/lr-merge.tikz}}
\tikzset{x=1em, y=1.5ex}
, 
\tikzset{x=1em, y=2.1ex}
\begin{tikzpicture}
	\begin{pgfonlayer}{nodelayer}
		\node [style=black] (37) at (-0.5, 0) {};
		\node [style=none] (43) at (0.25, 0) {};
		\node [style=none] (44) at (0.75, 0) {};
	\end{pgfonlayer}
	\begin{pgfonlayer}{edgelayer}
		\draw [->] (37) to (43.center);
		\draw (44.center) to (43.center);
	\end{pgfonlayer}
\end{tikzpicture}
}
\tikzset{x=1em, y=1.5ex}
$ of $d$ past the block of $\scalar{e}$ of $c$. As a result, their horizontal composite is equal to a diagram that factors as a block of $
\tikzset{x=1em, y=2.1ex}
\InputIfFileExists{lr-copy.tikz}{}{\input{./tikz/lr-copy.tikz}}
\tikzset{x=1em, y=1.5ex}
, 
\tikzset{x=1em, y=2.1ex}
\begin{tikzpicture}
	\begin{pgfonlayer}{nodelayer}
		\node [style=black] (37) at (0.75, 0) {};
		\node [style=none] (43) at (0.25, 0) {};
		\node [style=none] (44) at (-0.5, 0) {};
	\end{pgfonlayer}
	\begin{pgfonlayer}{edgelayer}
		\draw (43.center) to (37);
		\draw [->] (44.center) to (43.center);
	\end{pgfonlayer}
\end{tikzpicture}
}
\tikzset{x=1em, y=1.5ex}
$, followed by \emph{two} blocks of $\scalar{e} (e\in \RegExp)$ and finally, a block of $
\tikzset{x=1em, y=2.1ex}
\InputIfFileExists{lr-merge.tikz}{}{\input{./tikz/lr-merge.tikz}}
\tikzset{x=1em, y=1.5ex}
, 
\tikzset{x=1em, y=2.1ex}
\begin{tikzpicture}
	\begin{pgfonlayer}{nodelayer}
		\node [style=black] (37) at (-0.5, 0) {};
		\node [style=none] (43) at (0.25, 0) {};
		\node [style=none] (44) at (0.75, 0) {};
	\end{pgfonlayer}
	\begin{pgfonlayer}{edgelayer}
		\draw [->] (37) to (43.center);
		\draw (44.center) to (43.center);
	\end{pgfonlayer}
\end{tikzpicture}
}
\tikzset{x=1em, y=1.5ex}
$. Now, we may have two types of subdiagrams to eliminate.
\begin{itemize}
\item If the diagram contains 
\tikzset{x=1em, y=2.1ex}
\begin{tikzpicture}
	\begin{pgfonlayer}{nodelayer}
		\node [style=none] (0) at (-2.25, 0) {};
		\node [style=reg] (1) at (-0.75, 0) {$e$};
		\node [style=reg] (2) at (1, 0) {$e'$};
		\node [style=none] (3) at (2.75, 0) {};
	\end{pgfonlayer}
	\begin{pgfonlayer}{edgelayer}
		\draw (0.center) to (1);
		\draw (1) to (2);
		\draw (2) to (3.center);
	\end{pgfonlayer}
\end{tikzpicture}
}
\tikzset{x=1em, y=1.5ex}
 then we can turn this into a single coefficient as follows:
\begin{equation*}

\tikzset{x=1em, y=2.1ex}
\begin{tikzpicture}
	\begin{pgfonlayer}{nodelayer}
		\node [style=none] (0) at (-2.25, 0) {};
		\node [style=reg] (1) at (-0.75, 0) {$e$};
		\node [style=reg] (2) at (1, 0) {$e'$};
		\node [style=none] (3) at (2.75, 0) {};
	\end{pgfonlayer}
	\begin{pgfonlayer}{edgelayer}
		\draw (0.center) to (1);
		\draw (1) to (2);
		\draw (2) to (3.center);
	\end{pgfonlayer}
\end{tikzpicture}
}
\tikzset{x=1em, y=1.5ex}
 \;=\; 
\tikzset{x=1em, y=2.1ex}
\InputIfFileExists{regexps-product.tikz}{}{\input{./tikz/regexps-product.tikz}}
\tikzset{x=1em, y=1.5ex}
 \;=\; 
\tikzset{x=1em, y=2.1ex}
\InputIfFileExists{regexps-product-1.tikz}{}{\input{./tikz/regexps-product-1.tikz}}
\tikzset{x=1em, y=1.5ex}
 \;=\; \scalar{e\,e'}
\end{equation*}
\item If the diagram contains 
\tikzset{x=1em, y=2.1ex}
\InputIfFileExists{regexps-union.tikz}{}{\input{./tikz/regexps-union.tikz}}
\tikzset{x=1em, y=1.5ex}
 then we can turn this into a single coefficient as follows:
\begin{equation*}

\tikzset{x=1em, y=2.1ex}
\InputIfFileExists{regexps-union.tikz}{}{\input{./tikz/regexps-union.tikz}}
\tikzset{x=1em, y=1.5ex}
 \;=\; 
\tikzset{x=1em, y=2.1ex}
\InputIfFileExists{regexps-sum.tikz}{}{\input{./tikz/regexps-sum.tikz}}
\tikzset{x=1em, y=1.5ex}
 \;=\; 
\tikzset{x=1em, y=2.1ex}
\InputIfFileExists{regexps-sum-1.tikz}{}{\input{./tikz/regexps-sum-1.tikz}}
\tikzset{x=1em, y=1.5ex}
 \;=\; \scalar{e_1+e_2}
\end{equation*}
\end{itemize}
In this way we merge the two blocks of regular expression coefficients into a single block, as needed to obtain a generalised matrix-diagram.
\item The case of vertical composition is immediate: if $d_1\from \objr^{n_1}\to \objr^{m_1}$ and $d_2\from \objr^{n_2}\to \objr^{m_2}$ are two generalised matrix-diagrams then so is their vertical composite.
\item The remaining case is that of the trace. Suppose $d\from \objr^{1+n}\to \objr^{1+m}$ is a generalised matrix-diagram.  Then, there exists a regular expression $d_{11}$ and generalised matrix-subdiagrams $c_{1m}, c_{n1}$, and $c_{nm}$ such that
\begin{equation*}
d = \quad 
\tikzset{x=1em, y=2.1ex}
\InputIfFileExists{d-submatrices.tikz}{}{\input{./tikz/d-submatrices.tikz}}
\tikzset{x=1em, y=1.5ex}

\end{equation*}  
Then 
\begin{align*}
\traceform{d}{n}{m}{}\quad &=\quad 
\tikzset{x=1em, y=2.1ex}
\InputIfFileExists{regexp-inductive-trace.tikz}{}{\input{./tikz/regexp-inductive-trace.tikz}}
\tikzset{x=1em, y=1.5ex}
\\
& \myeq{A1-A2}\; 
\tikzset{x=1em, y=2.1ex}
\InputIfFileExists{regexp-inductive-trace-1.tikz}{}{\input{./tikz/regexp-inductive-trace-1.tikz}}
\tikzset{x=1em, y=1.5ex}
\\
& :=\; 
\tikzset{x=1em, y=2.1ex}
\InputIfFileExists{regexp-inductive-trace-2.tikz}{}{\input{./tikz/regexp-inductive-trace-2.tikz}}
\tikzset{x=1em, y=1.5ex}
\\
& \myeq{C5}\; 
\tikzset{x=1em, y=2.1ex}
\InputIfFileExists{regexp-inductive-trace-3.tikz}{}{\input{./tikz/regexp-inductive-trace-3.tikz}}
\tikzset{x=1em, y=1.5ex}
\\
& =:\; 
\tikzset{x=1em, y=2.1ex}
\InputIfFileExists{regexp-inductive-trace-4.tikz}{}{\input{./tikz/regexp-inductive-trace-4.tikz}}
\tikzset{x=1em, y=1.5ex}

\end{align*}
Then we can use the inductive case of composition above (1.) to obtain a generalised matrix diagram from the horizontal composite of $c_{n_1}$, $d_{11}^*$ and $c_{1m}$, and therefore for the whole diagram, thus concluding the proof.
\end{enumerate}
\end{proof}


\begin{proof}[Proof of Theorem~\ref{thm:copy-merge}]
By Proposition~\ref{thm:matrix-regexp}, any $d$ as in the statement of the theorem is equal to a generalised matrix-diagram. These are made up of consecutive blocks of $
\tikzset{x=1em, y=2.1ex}
\InputIfFileExists{lr-copy.tikz}{}{\input{./tikz/lr-copy.tikz}}
\tikzset{x=1em, y=1.5ex}
, 
\tikzset{x=1em, y=2.1ex}
\begin{tikzpicture}
	\begin{pgfonlayer}{nodelayer}
		\node [style=black] (37) at (0.75, 0) {};
		\node [style=none] (43) at (0.25, 0) {};
		\node [style=none] (44) at (-0.5, 0) {};
	\end{pgfonlayer}
	\begin{pgfonlayer}{edgelayer}
		\draw (43.center) to (37);
		\draw [->] (44.center) to (43.center);
	\end{pgfonlayer}
\end{tikzpicture}
}
\tikzset{x=1em, y=1.5ex}
$, $
\tikzset{x=1em, y=2.1ex}
\InputIfFileExists{regexp-action.tikz}{}{\input{./tikz/regexp-action.tikz}}
\tikzset{x=1em, y=1.5ex}
$, and $
\tikzset{x=1em, y=2.1ex}
\InputIfFileExists{lr-merge.tikz}{}{\input{./tikz/lr-merge.tikz}}
\tikzset{x=1em, y=1.5ex}
, 
\tikzset{x=1em, y=2.1ex}
\begin{tikzpicture}
	\begin{pgfonlayer}{nodelayer}
		\node [style=black] (37) at (-0.5, 0) {};
		\node [style=none] (43) at (0.25, 0) {};
		\node [style=none] (44) at (0.75, 0) {};
	\end{pgfonlayer}
	\begin{pgfonlayer}{edgelayer}
		\draw [->] (37) to (43.center);
		\draw (44.center) to (43.center);
	\end{pgfonlayer}
\end{tikzpicture}
}
\tikzset{x=1em, y=1.5ex}
$. Each equation in the statement holds for all of these components. For example, (cpy) holds for $
\tikzset{x=1em, y=2.1ex}
\InputIfFileExists{lr-merge.tikz}{}{\input{./tikz/lr-merge.tikz}}
\tikzset{x=1em, y=1.5ex}
$ by (B7), for $
\tikzset{x=1em, y=2.1ex}
\begin{tikzpicture}
	\begin{pgfonlayer}{nodelayer}
		\node [style=black] (37) at (-0.5, 0) {};
		\node [style=none] (43) at (0.25, 0) {};
		\node [style=none] (44) at (0.75, 0) {};
	\end{pgfonlayer}
	\begin{pgfonlayer}{edgelayer}
		\draw [->] (37) to (43.center);
		\draw (44.center) to (43.center);
	\end{pgfonlayer}
\end{tikzpicture}
}
\tikzset{x=1em, y=1.5ex}
$ by (B8), and for the $
\tikzset{x=1em, y=2.1ex}
\InputIfFileExists{regexp-action.tikz}{}{\input{./tikz/regexp-action.tikz}}
\tikzset{x=1em, y=1.5ex}
$ block by (D1) in conjunction with Lemma~\ref{thm:copy-regexp} to copy $\diagregexp{e}$. 

Thus, a simple structural induction on the form of generalised matrix-diagrams suffices to prove the theorem. 
\end{proof}

\begin{lemma}\label{lem:absorb-merge}
For any deterministic matrix-diagram $d$, there exists a deterministic matrix-diagram $d'$ such that
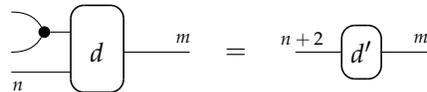
\[
\tikzset{x=1em, y=2.1ex}
\InputIfFileExists{absorb-merge.tikz}{}{\input{./tikz/absorb-merge.tikz}}
\tikzset{x=1em, y=1.5ex}
\quad =\quad \diagbox{d'}{n+2}{m}\]
\end{lemma}
\begin{proof}
This is a straightforward induction on the structure of matrix-diagrams: they are formed of a layer of $\Bcomult$ followed by a layer of action nodes $\scalar{x}$ and finally a layer of $\Bmult$. We can use the bimonoid axioms (B7)-(B9) to push the $\Bmult$ past the $\Bcomult$-layer and the ability to merge two $\scalar{a}$ using co-copying axiom (D3), to push $\Bmult$ past the $\scalar{x}$-layer. 
\end{proof}

The next lemma performs the key step in removing nondeterminism. In more transparent language, it asserts that if we have a diagram that correspond to a deterministic automaton, and we identify two inputs to two of its states, we can get rid of the potential nondeterminism that we have introduced, using equational reasoning. 
\begin{lemma}\label{thm:elim-nondeterminism}
For a matrix-diagram $d\from \objr^{l+3}\to \objr^{l+1}$ with $d_{ll}$ and $d_{3l}$ deterministic, there exists a matrix-diagram $d'\from \objr^{l'+2}\to \objr^{l'+1}$ with $d'_{l'l'}$ and $d'_{2l'}$ deterministic such that 
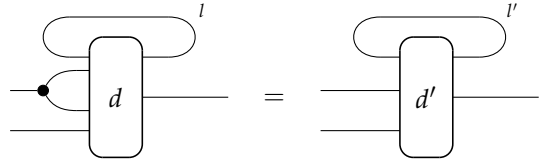
\[
\tikzset{x=1em, y=2.1ex}
\InputIfFileExists{lemma-elim-nondeterminism.tikz}{}{\input{./tikz/lemma-elim-nondeterminism.tikz}}
\tikzset{x=1em, y=1.5ex}
 \quad =\quad 
\tikzset{x=1em, y=2.1ex}
\InputIfFileExists{lemma-elim-nondeterminism-absorbed.tikz}{}{\input{./tikz/lemma-elim-nondeterminism-absorbed.tikz}}
\tikzset{x=1em, y=1.5ex}
\]
\end{lemma}
\begin{proof}
The idea is to identify the largest equivalent (i.e., that give the same language) subdiagrams, starting from any of the branches of $\Bcomult$, and pull them through $\Bcomult$ to merge them. Thinking of the diagram as an automaton, this amounts to identify the intersection of the languages recognised by the two states that $\Bcomult$ merges, to pull them through this generator, and therefore create a new state that recognises the intersection.

Following this idea, we take the largest submatrix-diagram $c$ of $d$ such that
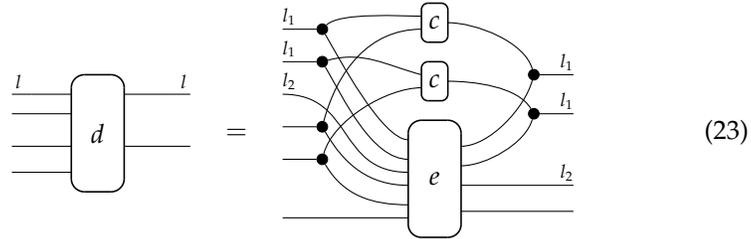
\begin{equation}\label{eq:two-states-intersection}

\tikzset{x=1em, y=2.1ex}
\InputIfFileExists{deterministic-to-merge.tikz}{}{\input{./tikz/deterministic-to-merge.tikz}}
\tikzset{x=1em, y=1.5ex}
\quad=\quad
\tikzset{x=1em, y=2.1ex}
\InputIfFileExists{two-states-intersection.tikz}{}{\input{./tikz/two-states-intersection.tikz}}
\tikzset{x=1em, y=1.5ex}

\end{equation}
for some deterministic $e$, and $l=l_1+l_1+l_2$.
Note that if there is no such subdiagram, we are done, since merging the two states does not introduce nondeterminism. Otherwise we proceed as follows. 
First, replacing~\eqref{eq:two-states-intersection} in the context of the statement, we obtain
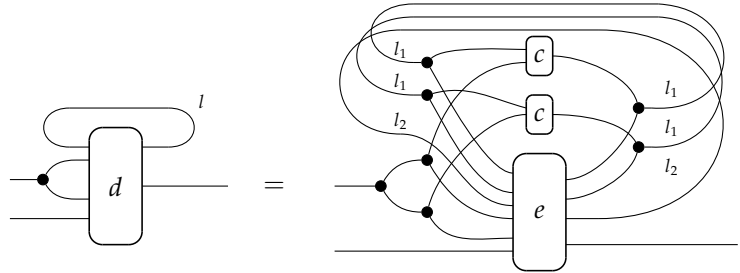
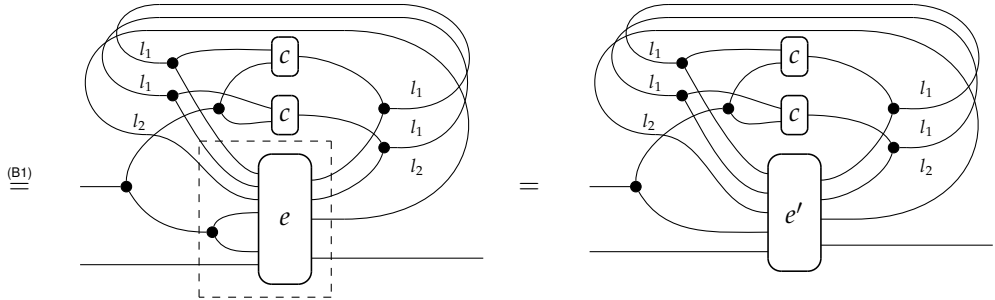
\begin{align*}

\tikzset{x=1em, y=2.1ex}
\InputIfFileExists{lemma-elim-nondeterminism.tikz}{}{\input{./tikz/lemma-elim-nondeterminism.tikz}}
\tikzset{x=1em, y=1.5ex}
\quad &=\quad
\tikzset{x=1em, y=2.1ex}
\InputIfFileExists{lemma-elim-nondeterminism-1.tikz}{}{\input{./tikz/lemma-elim-nondeterminism-1.tikz}}
\tikzset{x=1em, y=1.5ex}
\\
\myeq{B1}\quad
\tikzset{x=1em, y=2.1ex}
\InputIfFileExists{lemma-elim-nondeterminism-2.tikz}{}{\input{./tikz/lemma-elim-nondeterminism-2.tikz}}
\tikzset{x=1em, y=1.5ex}
 &=\quad 
\tikzset{x=1em, y=2.1ex}
\InputIfFileExists{lemma-elim-nondeterminism-2bis.tikz}{}{\input{./tikz/lemma-elim-nondeterminism-2bis.tikz}}
\tikzset{x=1em, y=1.5ex}

\end{align*}
where $e'$ is the dashed box in the previous diagram. Note that we have not introduced more nondeterminism since, by construction, $e'_{ll}$ and $e'_{1l}$ are deterministic. Indeed, otherwise, $c$ would not be the largest subdiagram satisfying~\eqref{eq:two-states-intersection} and we could add any additional nondeterministic transition in $e'$ to it.

We now focus on transforming the following subdiagram, which we isolate for clarity. First, we can split $c$ into two submatrix-diagrams $c_1$ and $c_2$ such that
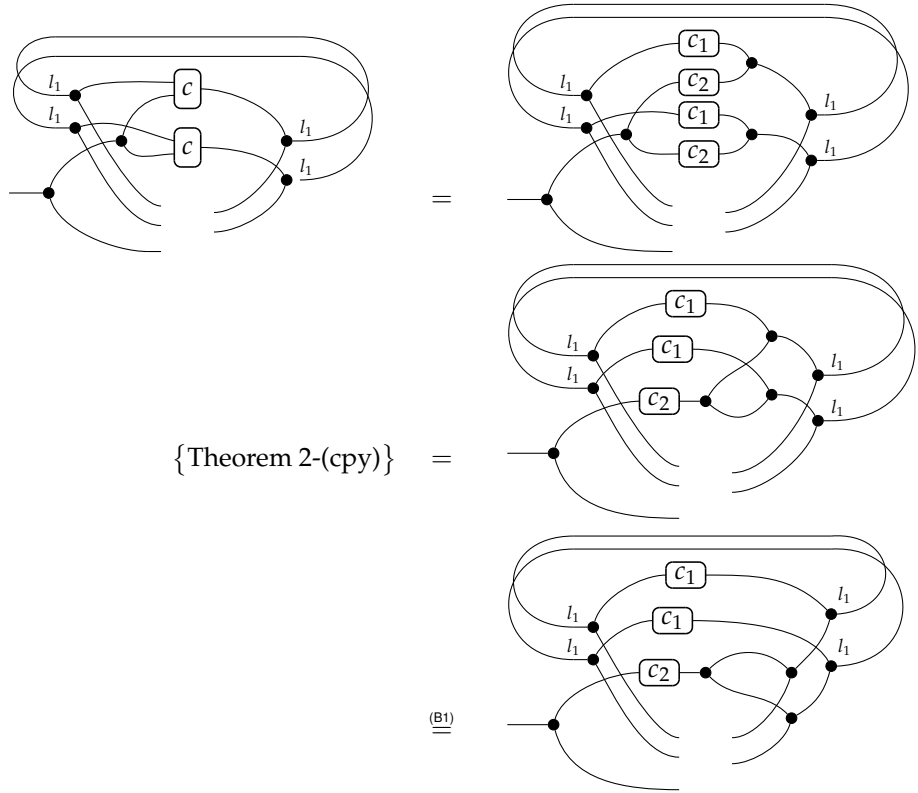
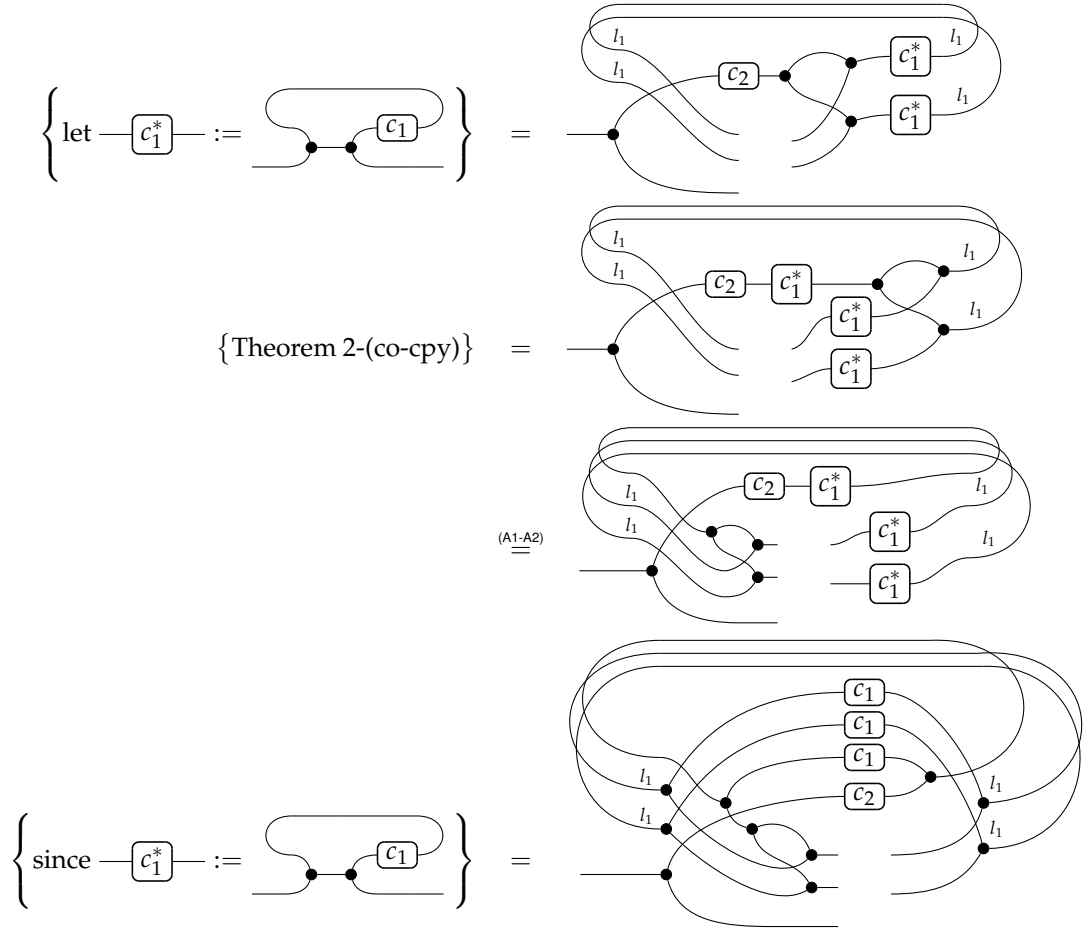
\begin{align*}

\tikzset{x=1em, y=2.1ex}
\InputIfFileExists{lemma-elim-nondeterminism-3.tikz}{}{\input{./tikz/lemma-elim-nondeterminism-3.tikz}}
\tikzset{x=1em, y=1.5ex}
\quad & = \quad 
\tikzset{x=1em, y=2.1ex}
\InputIfFileExists{merge-non-determinism.tikz}{}{\input{./tikz/merge-non-determinism.tikz}}
\tikzset{x=1em, y=1.5ex}
\\
\big\{\text{Theorem~\ref{thm:copy-merge}-(cpy)}\big\}\quad & = \quad 
\tikzset{x=1em, y=2.1ex}
\InputIfFileExists{merge-non-determinism-2.tikz}{}{\input{./tikz/merge-non-determinism-2.tikz}}
\tikzset{x=1em, y=1.5ex}
\\
\quad & \myeq{B1} \quad 
\tikzset{x=1em, y=2.1ex}
\InputIfFileExists{merge-non-determinism-3.tikz}{}{\input{./tikz/merge-non-determinism-3.tikz}}
\tikzset{x=1em, y=1.5ex}
\\
\left\{\text{let }
\tikzset{x=1em, y=2.1ex}
\InputIfFileExists{star-c_i.tikz}{}{\input{./tikz/star-c_i.tikz}}
\tikzset{x=1em, y=1.5ex}
 :=
\tikzset{x=1em, y=2.1ex}
\InputIfFileExists{star-c_i-def.tikz}{}{\input{./tikz/star-c_i-def.tikz}}
\tikzset{x=1em, y=1.5ex}
 \right\}\quad & = \quad 
\tikzset{x=1em, y=2.1ex}
\InputIfFileExists{merge-non-determinism-4.tikz}{}{\input{./tikz/merge-non-determinism-4.tikz}}
\tikzset{x=1em, y=1.5ex}
\\
\big\{\text{Theorem~\ref{thm:copy-merge}-(co-cpy)}\big\}\quad & = \quad 
\tikzset{x=1em, y=2.1ex}
\InputIfFileExists{merge-non-determinism-5.tikz}{}{\input{./tikz/merge-non-determinism-5.tikz}}
\tikzset{x=1em, y=1.5ex}
\\
& \myeq{A1-A2} \quad 
\tikzset{x=1em, y=2.1ex}
\InputIfFileExists{merge-non-determinism-6.tikz}{}{\input{./tikz/merge-non-determinism-6.tikz}}
\tikzset{x=1em, y=1.5ex}
\\
\left\{\text{since }
\tikzset{x=1em, y=2.1ex}
\InputIfFileExists{star-c_i.tikz}{}{\input{./tikz/star-c_i.tikz}}
\tikzset{x=1em, y=1.5ex}
 :=
\tikzset{x=1em, y=2.1ex}
\InputIfFileExists{star-c_i-def.tikz}{}{\input{./tikz/star-c_i-def.tikz}}
\tikzset{x=1em, y=1.5ex}
 \right\}\quad& =  
\tikzset{x=1em, y=2.1ex}
\InputIfFileExists{merge-non-determinism-6-1.tikz}{}{\input{./tikz/merge-non-determinism-6-1.tikz}}
\tikzset{x=1em, y=1.5ex}

\end{align*}
In context, this gives
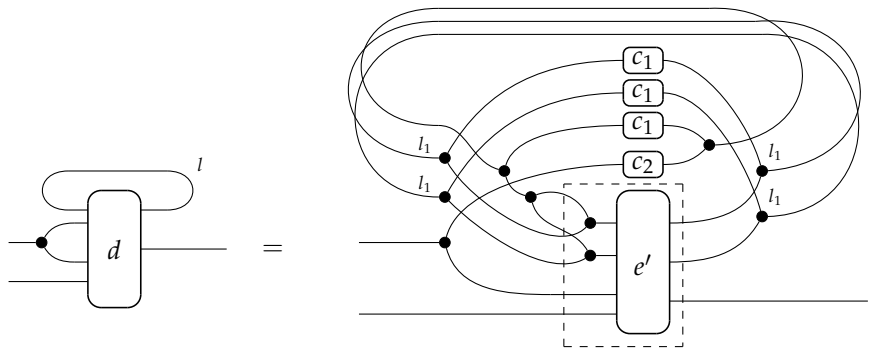
\begin{equation*}

\tikzset{x=1em, y=2.1ex}
\InputIfFileExists{lemma-elim-nondeterminism.tikz}{}{\input{./tikz/lemma-elim-nondeterminism.tikz}}
\tikzset{x=1em, y=1.5ex}
\quad=\quad
\tikzset{x=1em, y=2.1ex}
\InputIfFileExists{merge-non-determinism-7.tikz}{}{\input{./tikz/merge-non-determinism-7.tikz}}
\tikzset{x=1em, y=1.5ex}

\end{equation*}
Now, we can use Lemma~\ref{lem:absorb-merge} to absorb the two occurences of $\Bmult$ into $e'$. We get $e''$ such that
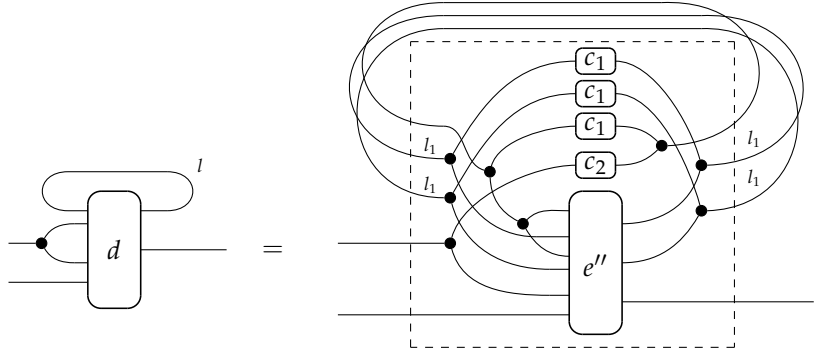
\begin{equation*}

\tikzset{x=1em, y=2.1ex}
\InputIfFileExists{lemma-elim-nondeterminism.tikz}{}{\input{./tikz/lemma-elim-nondeterminism.tikz}}
\tikzset{x=1em, y=1.5ex}
\quad=\quad
\tikzset{x=1em, y=2.1ex}
\InputIfFileExists{merge-non-determinism-8.tikz}{}{\input{./tikz/merge-non-determinism-8.tikz}}
\tikzset{x=1em, y=1.5ex}

\end{equation*}
Now, the dashed box in the diagram above is the $d'$ required by the statement of the lemma: indeed, any further nondeterministic transition would mean that $c$ as chosen above was not be the largest subdiagram satisfying the conditions of~\eqref{eq:two-states-intersection}.
\end{proof}


\begin{proof}[Proof of Proposition~\ref{thm:deterministic-rep} (Determinisation)]
First, by Proposition~\ref{thm:traceform}, we can obtain a representation for any give left-to-right diagram. Thus, we only need to prove that, for any matrix-diagram $d\from \objr^{l+1}\to \objr^{l+1}$ with $d_{ll}$ $\epsilon$-free,
there exists $l'\in \N$ and a  matrix-diagram $d'\from \objr^{l'+1}\to \objr^{l'+1}$, with $d'_{l'l'}$ deterministic 
such that
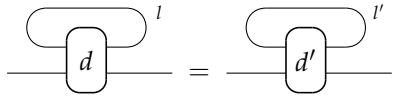
\[\traceform{d}{}{}{l}\; = \; \traceform{d'}{}{}{l'}\]
We proceed by induction on the number of nondeterministic transitions in $d$.
Recall that, in diagrammatic terms, a nondeterministic transition of the associated automaton corresponds to a subdiagram of the form, for some $a\in \Sigma$:
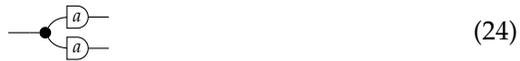
\begin{equation}\label{diag:non-determinism}

\tikzset{x=1em, y=2.1ex}
\InputIfFileExists{non-determinism.tikz}{}{\input{./tikz/non-determinism.tikz}}
\tikzset{x=1em, y=1.5ex}

\end{equation}
If there are none there is nothing to do. Assuming that the statement of the theorem holds for any matrix-diagram containing $n$ nondeterministic transitions, let $d$ be a matrix diagram with $n+1$ such transitions. We can choose one, so that there exists some diagram $d^{(1)}$ with
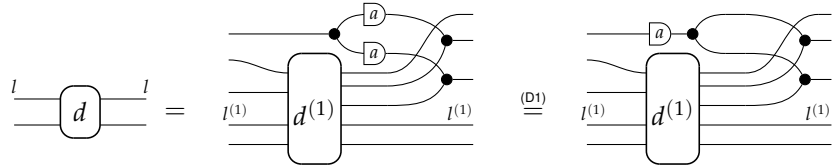
\begin{equation*}
\diagstate{d}{}{}{l} \;= \;
\tikzset{x=1em, y=2.1ex}
\InputIfFileExists{non-deterministic-ih.tikz}{}{\input{./tikz/non-deterministic-ih.tikz}}
\tikzset{x=1em, y=1.5ex}
  \;\myeq{D1} \;
\tikzset{x=1em, y=2.1ex}
\InputIfFileExists{non-deterministic-merged-ih.tikz}{}{\input{./tikz/non-deterministic-merged-ih.tikz}}
\tikzset{x=1em, y=1.5ex}

\end{equation*}
and $l^{(1)} = l - 3$, for $k$ the arity of the nondeterministic transition we picked. The second equation is an immediate application of Theorem~\ref{thm:copy-merge}.
Then, by tracing out, we get
\begin{align}\label{eq:absord-merge}
\traceform{d}{}{}{l}\quad &= \quad
\tikzset{x=1em, y=2.1ex}
\InputIfFileExists{non-deterministic-merged-ih-1.tikz}{}{\input{./tikz/non-deterministic-merged-ih-1.tikz}}
\tikzset{x=1em, y=1.5ex}
 \\
& = \quad
\tikzset{x=1em, y=2.1ex}
\InputIfFileExists{non-deterministic-merged-ih-2.tikz}{}{\input{./tikz/non-deterministic-merged-ih-2.tikz}}
\tikzset{x=1em, y=1.5ex}

\end{align}
Let $d^{(2)}$ be the subdiagram in the dashed box above and let $l^{(2)} = l^{(1)} +1$. We now have
\begin{equation}\label{eq:absord-action}
\traceform{d}{}{}{l}\quad = \quad
\tikzset{x=1em, y=2.1ex}
\InputIfFileExists{non-deterministic-merged-ih-3.tikz}{}{\input{./tikz/non-deterministic-merged-ih-3.tikz}}
\tikzset{x=1em, y=1.5ex}

\end{equation}
Note, that, by construction $d^{(2)}$ contains $n$ nondeterministic transitions. We can therefore apply the induction hypothesis to determinise the following subdiagram:
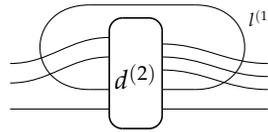
\[
\tikzset{x=1em, y=2.1ex}
\InputIfFileExists{to-determinise-ih.tikz}{}{\input{./tikz/to-determinise-ih.tikz}}
\tikzset{x=1em, y=1.5ex}
\]
From this, we obtain $d^{(3)}$ deterministic such that
\begin{align*}
\traceform{d}{}{}{l}\quad &= \quad
\tikzset{x=1em, y=2.1ex}
\InputIfFileExists{non-deterministic-merged-ih-4.tikz}{}{\input{./tikz/non-deterministic-merged-ih-4.tikz}}
\tikzset{x=1em, y=1.5ex}
 \\
& \myeq{A1-A2} \:
\tikzset{x=1em, y=2.1ex}
\InputIfFileExists{non-deterministic-merged-ih-5.tikz}{}{\input{./tikz/non-deterministic-merged-ih-5.tikz}}
\tikzset{x=1em, y=1.5ex}
 
\end{align*}
Applying Lemma~\ref{lem:absorb-merge} twice, we obtain $d^{(4)}$ such that
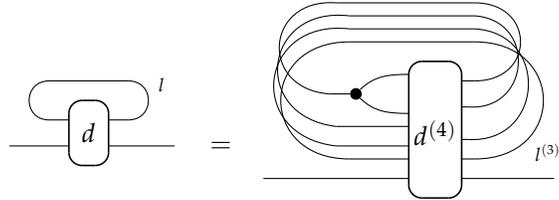
\begin{equation*}
\traceform{d}{}{}{l}\quad = \:
\tikzset{x=1em, y=2.1ex}
\InputIfFileExists{non-deterministic-merged-ih-6.tikz}{}{\input{./tikz/non-deterministic-merged-ih-6.tikz}}
\tikzset{x=1em, y=1.5ex}

\end{equation*}
and we are now able to apply Lemma~\ref{thm:elim-nondeterminism} to eliminate the copying node $\Bcomult$, obtaining $d'$ as needed, and concluding the proof.
\end{proof}

\end{document}